\documentclass[notitlepage,12pt]{article}
\usepackage{amsfonts}
\usepackage{amsmath}
\usepackage{amssymb}
\usepackage{float}
\usepackage{geometry}
\usepackage{float}
\usepackage{subfig}
\usepackage{graphicx}
\usepackage[onehalfspacing]{setspace}

\newtheorem{theorem}{Theorem}

\newtheorem{remark}[theorem]{Remark}

\newenvironment{proof}[1][Proof]{\noindent \textbf{#1.} }{\  \rule{0.5em}{0.5em}}
\input{tcilatex}
\newtheorem{thm}{Theorem}

\newtheorem{assu}{Assumption}
\newtheorem{lem}{Lemma}
\newtheorem{prop}{Proposition}

\begin{document}

\title{Estimation and Inference for High Dimensional Factor Model with
Regime Switching}
\author{Giovanni Urga\thanks{%
Faculty of Finance, Bayes Business School (formerly Cass), 106 Bunhill Row,
London, EC1Y 8TZ, U.K.} \and Fa Wang\thanks{%
School of Economics, Peking University, 5 Yiheyuan Road, Beijing, 100871,
China.}}
\date{%
\today%
}
\maketitle

\begin{abstract}
\baselineskip=16.0pt

This paper proposes maximum (quasi)likelihood estimation for high
dimensional factor models with regime switching in the loadings. The model
parameters are estimated jointly by the EM (expectation maximization)
algorithm, which in the current context only requires iteratively
calculating regime probabilities and principal components of the weighted
sample covariance matrix. When regime dynamics are taken into account,
smoothed regime probabilities are calculated using a recursive algorithm.
Consistency, convergence rates and limit distributions of the estimated
loadings and the estimated factors are established under weak
cross-sectional and temporal dependence as well as heteroscedasticity. It is
worth noting that due to high dimension, regime switching can be identified
consistently after the switching point with only one observation. Simulation
results show good performance of the proposed method. An application to the
FRED-MD dataset illustrates the potential of the proposed method for
detection of business cycle turning points.



\thispagestyle{empty}%
\strut \bigskip \bigskip
\end{abstract}

\baselineskip=20.0pt\ \ \ \pagebreak 
\setcounter{page}{0}%

\section{Introduction\label{intro}}

A great deal of attention has focused on the loading instability issue in
high dimensional factor models. For empirical evidences of parameter
instability in macroeconomic and financial time series, see for example,
Banerjee, Marcellino and Masten (2008), Stock and Watson (2009) and
Korobilis (2013). Several procedures are proposed to detect and/or estimate
common abrupt breaks in the loadings, including Cheng, Liao and Shorfheide
(2016), Baltagi, Kao and Wang (2017, 2021), Bai, Han and Shi (2020), and Ma
and Su (2018), to mention a few. Other models of time varying loadings, such
as i.i.d./random walk, smooth change, vector autoregression and threshold
type, are studied in Bates, Plagborg-Moller, Stock and Watson (2013), Su and
Wang (2017), Mikkelsen, Hillebrand and Urga (2019) and Massacci (2017),
respectively.

An alternative approach of modeling loading instability is common regime
switching. In business cycle analysis, several unobservable factors
summarize the comovements of many economic variables and the loadings
measure the importance of factors for each economic variable. The importance
of each factor may be different depending on fiscal policy (expansionary,
contractionary, neutral), or monetary policy (expansionary, contractionary),
or the stage of the business cycle (peak, trough, expansion, contraction),
hence the loadings may switch synchronously between several states under
different scenarios. In stock return analysis, the loadings measure the
impact of the factor return on the expected return of each individual stock,
hence the loadings may switch synchronously depending on the stock market
scenarios (bull versus bear markets, high versus low volatility), see for
example Gu (2005) and Guidolin and Timmermann (2008) for related
discussions. In bond return analysis, the yields of bonds with different
maturities are well captured by the level factor, the slope factor and the
curvature factor, see for example Cochrane and Piazzesi (2005) and Diebold
and Li (2006). The importance of each factor could be different depending on
the stock market volatility, or the stage of the business cycle, or the
unemployment rate, hence the loadings may also switch synchronously
according to these state variables. In general, large factor models with
regime switching in the loadings could also be useful for other topics, such
as tracking labor productivity.

There are only a few related results on large factor models with regime
switching in the loadings. Liu and Chen (2016) proposes an iterative
algorithm for estimating the model parameters and the hidden states based on
eigen-decomposition and the Viterbi algorithm, however, the asymptotic
properties of the estimated parameters are established only when the true
states are known. Considering loadings as general functions of some
recurrent states, Pelger and Xiong (2021) develops nonparametric kernel
estimator for the loadings and the factors, and establishes the relevant
asymptotic theory. However, Pelger and Xiong (2021) requires observable
state variables. In general, state variables may be misspecified or
unobservable.

This paper proposes maximum (quasi)likelihood estimation for high
dimensional factor models with regime switching in the loadings when the
state variables are unobservable. This paper also proposes new criteria to
consistently determine the number of regimes and the number of factors in
each regime. The model parameters are estimated jointly by the EM algorithm,
which in the current context only requires calculating principal components
iteratively.

More specifically, in the E-step, the probabilities of each regime at each
time $t$ are calculated based on the observed data and the parameter values
at the current iteration using a recursive algorithm modified from Hamilton
(1990), and then the joint (log)likelihood of the observed data and the
unobserved states are averaged with respect to the calculated regime
probabilities. In the M-step, the estimated loadings for each regime are the
principal components of the weighted sample covariance matrix of the
observed time series, where the weight on $x_{t}$ (the observed time series
at time $t$) equals the probability of that regime at time $t$. Since
principal components can be easily calculated even when $N$ (the dimension
of time series) is large, our method is very easy to implement.

For the proposed algorithm, this paper establishes the convergence rates of
the estimated loading spaces and the estimated factor spaces, the limit
distributions of the estimated loadings and the estimated factors, the
consistency of the estimated regime probabilities, and the consistency of
the estimated transition probability matrix when the true state process is
Markovian. Note that asymptotic analysis under the regime switching setup is
more difficult than under the structural break setup, because the pattern of
regimes for the latter is much simpler.

These asymptotic results are essential in many empirical contexts. First,
the limit distributions of the estimated factors allow us to construct
confidence intervals for the true factors, which represent economic indices
in many applications. The result on the estimated factor spaces implies that
if the estimated factors are used in factor-augmented forecasting (or
factor-augmented VAR), the forecasting equation (or the VAR equation) would
have induced regime switching in the model parameters. Second, for asset
management, the estimated loadings of each regime allow us to construct
portfolios according to each specific market scenario. For structural
dynamic factor analysis, consistently estimated loadings are also crucial
for recovering the impulse responses. Third, the consistency of the
estimated regime probabilities implies that for each $x_{t}$, we can
consistently identify which regime $x_{t}$ belongs to as $N\rightarrow
\infty $. For asset management, this allows us to consistently identify the
current market scenario. For business cycle analysis, this allows us to
consistently date turning points of the business cycle and detect new
recessions or expansions, especially when high frequency (weekly, daily)
data is utilized.

For cases with small $N$, various methods have been proposed for estimating
factor models with regime switching. Kim (1994) proposes approximate Kalman
filter for likelihood evaluation and uses nonlinear optimization for
likelihood maximization. Kim and Yoo (1995) and Chauvet (1998) apply Kim
(1994)'s method to a small number of economic series and obtain recession
probabilities and turning points very close to the official NBER dates. Kim
(1994) allows for regime switching in both the factor mean and the factor
loadings, but when $N$ is large, Kim (1994)'s method would be very time
consuming and may have convergence problems\footnote{%
This is because the number of parameters grows proportionally to $N$ and the
likelihood function is calculated numerically and maximized by nonlinear
optimization algorithm.}. Other methods, such as Diebold and Rudebusch
(1996) and Kim and Nelson (1998), assume stable loadings and only focus on
regime switching in the factor mean. If the loadings are unstable, these
methods are not applicable. More importantly, if there is only regime
switching in the factor mean, we can not consistently identify each regime
even when $N$ is large.

In contrast with Kim (1994), our method is fast and easy to implement even
when $N$ is very large. The crucial point behind our EM algorithm is to
ignore factor dynamics\footnote{%
Factor dynamics are still allowed for the data generating process.
\par
\bigskip} and integrate out the factors in the likelihood function. If
factors dynamics are taken into account or factors are treated as parameters
in the likelihood function, the estimated loadings would not be the
principal components of the weighted sample covariance matrix, and
consequently both the algorithm and the asymptotic analysis would become
infeasible. On the other hand, the efficiency loss of ignoring factor
dynamics is small when $N$ is large.

This paper may also contribute to the literature on dating turning points of
the business cycle. Currently there are two main approaches for dating
business cycle using multiple time series. The first approach, aggregating
then dating, is to date business cycle by focusing on a few highly
aggregated time series such as GDP, industrial production and nonfarm
employment. The second approach, dating then aggregating, is to date turning
points in each disaggregated series and then aggregate these turning points
in some appropriate way, see Burns and Mitchell (1946), Harding and Pagan
(2006) and Chauvet and Piger (2008). These papers only use a small number of
time series. Stock and Watson (2010, 2014) studies this issue using many
time series. This paper shows that it is possible to consistently identify
turning points if regime switching is synchronous and $N$ is large enough.
If $N$ is small, consistency is not possible no matter how large $T$ is.
This paper also shows that if $N$ is large, it is possible to consistently
detect regime switching right after the turning point with only one
observation, thus the speed of detection could be improved significantly. If 
$N$ is small, we have to wait for enough observations from the new regime.

The rest of the paper is organized as follows. Section \ref{id&es}
introduces the model setup and the estimation procedures. Section \ref{asym}
presents the assumptions and the asymptotic results. Section \ref{rJ}
proposes criteria for determining the number of regimes and the number of
factors in each regime. Section \ref{Simu} presents simulation results.\
Section \ref{App} presents an empirical application of the proposed method
to the FRED-MD dataset. Section \ref{Con} concludes. All proofs are
relegated to the appendix.

Through out the paper, $(N,T)\rightarrow \infty $\ denotes $N$\ and $T$\
going to infinity jointly, $\delta _{NT}=\min \{\sqrt{N},\sqrt{T}\}$. $%
\overset{p}{\rightarrow }$ and $\overset{d}{\rightarrow }$ denotes
convergence in probability and convergence in distribution, respectively.
For matrix $A$, let $\left\Vert A\right\Vert $, $\left\Vert A\right\Vert
_{F} $, $\rho _{\max }(A)$ and $\rho _{\min }(A)$ denote its spectral norm,
Frobenius norm, largest eigenvalue and smallest eigenvalue, respectively.
Let $P_{A}=A(A^{\prime }A)^{-1}A^{\prime }$ denote the projection matrix and 
$M_{A}=I-P_{A}$. "w.p.a.1" denotes with probability approaching one.

\section{Identification and Estimation\label{id&es}}

Consider the following factor model with regime switching: for $i=1,...,N$
and $t=1,...,T,$ 
\begin{equation}
x_{it}=f_{t}^{0\prime }\lambda _{ji}^{0}+e_{it}\text{ if }z_{t}=j,\text{ }
\end{equation}%
where $\lambda _{ji}^{0}$ is an $r_{j}^{0}$ dimensional vector of loadings
for regime $j$, $f_{t}^{0}$ is an $r_{z_{t}}^{0}$ dimensional vector of
factors, $z_{t}$ is the state variable indicating which regime $x_{it}$
belongs to, and $e_{it}$ is the error term allowed to have cross-sectional
and temporal dependence as well as heteroscedasticity. $x_{it}$ is
observable and all of the right hand side variables are unobservable. The
number of regimes $J^{0}$ and the number of factors in each regime $%
r_{j}^{0} $ (could be different across $j$) are fixed as $(N,T)\rightarrow
\infty $ and assumed to be known in this section and Section \ref{asym}. How
to consistently determine $r_{j}^{0}$\ and $J^{0}$\ will be studied in
Section \ref{rJ}.

The factor process $\{f_{t}^{0},t=1,...,T\}$ is allowed to be dynamic, and
similar to the principal component estimator (PCE) in Bai (2003) and the
maximum likelihood estimator (MLE) in Bai and Li (2012, 2016), factor
dynamics are ignored when estimating the model parameters, thus there is no
need to model factor dynamics.

For the state process $\{z_{t},t=1,...,T\}$, the asymptotic results in
Section \ref{AsyDyn} and Section \ref{rJ} are valid as long as $\frac{1}{T}%
\sum\nolimits_{t=1}^{T}1_{z_{t}=j}\overset{p}{\rightarrow }q_{j}^{0}>0$\ for 
$j=1,...,J^{0}$\ ($q_{j}^{0}=\Pr (z_{t}=j)$\ is the unconditional
probability of regime $j$\ and $1_{z_{t}=j}=1$\ if $z_{t}=j$\ and $0$\
otherwise), and Assumptions \ref{factors}-\ref{error} and \ref{moments}-\ref%
{dist} in Section \ref{assu} hold conditioning on $\{z_{t},t=1,...,T\}$.
Thus $\{z_{t},t=1,...,T\}$\ is allowed to be correlated with $f_{s}^{0}$\
and $e_{is}$\ for all $i$\ and $s$, and we do not need to know the true
model of $\{z_{t},t=1,...,T\}$.

In vector form, the model can be written as:%
\begin{equation}
x_{t}=\Lambda _{j}^{0}f_{t}^{0}+e_{t}\text{ if }z_{t}=j\text{, for }%
t=1,...,T,
\end{equation}%
where $\Lambda _{j}^{0}=(\lambda _{j1}^{0},...,\lambda _{jN}^{0})^{\prime }$%
, $x_{t}=(x_{1t},...,x_{Nt})^{\prime }$ and $e_{t}=(e_{1t},...,e_{Nt})^{%
\prime }$. Let $\Lambda ^{0}=(\Lambda _{1}^{0},...,\Lambda _{J^{0}}^{0})$
and let $E=(e_{1},...,e_{T})^{\prime }$ be the $T\times N$ matrix of errors.
When there are no superscripts, $\Lambda _{j}$ and $\Lambda $ denote
parameters as variables.

\subsection{Identification\label{id}}

Since the factors are unobservable, regimes are defined in terms of the
linear spaces spanned by the loadings. Two regimes are different if their
loading spaces are different, and vice versa. More specifically, the
identification condition is: for any $j$\ and $k$,%
\begin{equation}
\min_{t}\frac{1}{N}\left\Vert M_{\Lambda _{k}^{0}}\Lambda
_{j}^{0}f_{t}^{0}\right\Vert ^{2}=\min_{t}\frac{1}{N}f_{t}^{0\prime }\Lambda
_{j}^{0\prime }M_{\Lambda _{k}^{0}}\Lambda _{j}^{0}f_{t}^{0}\geq C\text{ for
some }C>0.  \label{ide}
\end{equation}%
A sufficient condition for (\ref{ide}) is:%
\begin{eqnarray}
&&\text{ }\lim_{N\rightarrow \infty }\frac{1}{N}\Lambda _{k}^{0\prime
}M_{\Lambda _{j}^{0}}\Lambda _{k}^{0}\ \text{is positive definite for any }j%
\text{\ and }k\text{,}  \label{ar} \\
&&\text{and }\min_{t}\left\Vert f_{t}\right\Vert \text{ is nonzero.}  \notag
\end{eqnarray}%
Condition (\ref{ar}) requires $\lim\limits_{N\rightarrow \infty }\frac{1}{N}%
(\Lambda _{j}^{0},\Lambda _{k}^{0})^{\prime }(\Lambda _{j}^{0},\Lambda
_{k}^{0})\ $to be full rank for any $j$\ and $k$. Thus $\Lambda _{j}^{0}$\
and $\Lambda _{k}^{0}$\ are not allowed to share some columns, and columns
of $\Lambda _{j}^{0}$\ could not be linear combination of $\Lambda _{k}^{0}$%
\ and vice versa. An alternative sufficient condition for (\ref{ide}) is:%
\begin{eqnarray}
&&\lim_{N\rightarrow \infty }\frac{1}{N}\Lambda _{k}^{0\prime }M_{\Lambda
_{j}^{0}}\Lambda _{k}^{0}\neq 0\text{ for any }j\ \text{and }k\text{,}
\label{ar'} \\
&&\text{and }\min_{t}\left\vert g_{jk}^{\prime }f_{t}\right\vert \text{ is
nonzero,}  \notag
\end{eqnarray}%
where $g_{jk}$ is the eigenvector of $\lim\limits_{N\rightarrow \infty }%
\frac{1}{N}\Lambda _{k}^{0\prime }M_{\Lambda _{j}^{0}}\Lambda _{k}^{0}$
corresponding to nonzero eigenvalue. Condition (\ref{ar'}) only requires
that the linear spaces spanned by $\Lambda _{j}^{0}$ and $\Lambda _{k}^{0}$
are different. Thus $\Lambda _{k}^{0}$\ and $\Lambda _{j}^{0}$\ are allowed
to share some columns, and some columns of $\Lambda _{k}^{0}$ are allowed to
be linear combinations of the columns of $\Lambda _{j}^{0}$ and vice versa,
but $\Lambda _{k}^{0}$\ is not allowed to be a subset of $\Lambda _{j}^{0}$.
For example, if there are two regimes with two factors in each regime and
only the loadings of $f_{2t}$ (the second factor) switch across the regimes,
then condition (\ref{ar'}) requires that $\min_{t}\left\vert
f_{2t}\right\vert $\ is nonzero.

Note that condition (\ref{ar}) does not rule out the possibility that any
regime $j$ can be further decomposed into multiple regimes. Suppose the true
model is $x_{t}=\Lambda _{j}^{0}f_{t}^{0}+e_{t}$\ if $z_{t}=j$, $j=1,2,3$,
and $(\Lambda _{1}^{0},\Lambda _{2}^{0},\Lambda _{3}^{0})$\ satisfies
condition (\ref{ar}). If we consider $\Lambda _{1}^{0}$\ as the first regime
and $(\Lambda _{2}^{0},\Lambda _{3}^{0})$\ as the second regime, the true
model can be equivalently written as $x_{t}=\Lambda _{1}^{0}f_{t}^{0}+e_{t}$%
\ if $z_{t}=1$, and $x_{t}=(\Lambda _{2}^{0},\Lambda _{3}^{0})f_{t}^{\ast
}+e_{t}$\ if $z_{t}=2$\ $or$\ $3$, where $f_{t}^{\ast }=\left(
f_{t}^{0\prime },0^{\prime }\right) ^{\prime }$\ if $z_{t}=2$\ and $%
f_{t}^{\ast }=\left( 0^{\prime },f_{t}^{0\prime }\right) ^{\prime }$\ if $%
z_{t}=3$. The equivalent model also satisfies condition (\ref{ar}). However,
while $plim\frac{1}{T}\sum\nolimits_{z_{t}=2\text{ }or\text{ }3}f_{t}^{\ast
}f_{t}^{\ast \prime }$\ is positive definite, $plim\frac{1}{T}%
\sum\nolimits_{z_{t}=2}f_{t}^{\ast }f_{t}^{\ast \prime }$\ and $plim\frac{1}{%
T}\sum\nolimits_{z_{t}=3}f_{t}^{\ast }f_{t}^{\ast \prime }$\ are not
positive definite. To rule out the possibility that any regime $j$ can be
further decomposed, we assume that%
\begin{equation}
plim\frac{1}{\left\vert A_{j}\right\vert }\sum\nolimits_{t\in
A_{j}}f_{t}^{0}f_{t}^{0\prime }\text{ is positive definite,}
\end{equation}%
where $A_{j}$ denotes any subset of $\{t:z_{t}=j\}$ with cardinality $%
\left\vert A_{j}\right\vert $ and $\lim \frac{\left\vert A_{j}\right\vert }{T%
}>0$. If $\frac{1}{\left\vert A_{j}\right\vert }\sum\nolimits_{t\in
A_{j}}f_{t}^{0}f_{t}^{0\prime }$\ is not positive definite as $T\rightarrow
\infty $\ for some $A_{j}$, then $A_{j}$\ and $\{t:z_{t}=j,t\notin A_{j}\}$
are considered as two separate regimes.

\subsection{First Order Conditions\label{focsDyn}}

Consider the following log-likelihood function for Gaussian mixture in
covariance:%
\begin{equation}
l(\Lambda ,\sigma ^{2})=\log
[\sum\nolimits_{z_{T}=1}^{J^{0}}...\sum\nolimits_{z_{1}=1}^{J^{0}}\prod%
\nolimits_{t=1}^{T}L(x_{t}\left\vert z_{t};\Lambda ,\sigma ^{2}\right. )\Pr
(z_{1},...,z_{T})],  \label{ae}
\end{equation}%
where $\prod\nolimits_{t=1}^{T}L(x_{t}\left\vert z_{t};\Lambda ,\sigma
^{2}\right. )$ is the density of $(x_{1},...,x_{T})$ conditioning on $%
(z_{1},...,z_{T})$, $\Pr (z_{1},...,z_{T})$ is the joint probability of $%
(z_{1},...,z_{T})$,%
\begin{equation}
L(x_{t}\left\vert z_{t}=j;\Lambda _{j},\sigma ^{2}\right. )=(2\pi )^{-\frac{N%
}{2}}\left\vert \Sigma _{j}\right\vert ^{-\frac{1}{2}}e^{-\frac{1}{2}%
x_{t}^{\prime }\Sigma _{j}{}^{-1}x_{t}}\text{,}  \label{ab}
\end{equation}%
$\Sigma _{j}$ is the covariance matrix of $x_{t}$ for regime $j$, and%
\begin{equation}
\Sigma _{j}=\Lambda _{j}\Lambda _{j}^{\prime }+\sigma ^{2}I_{N}.  \label{aq}
\end{equation}

The above log-likelihood function avoids estimating the factors. If the
factors are estimated jointly with the loadings, we would not have the
analytical first order conditions presented below, and consequently the EM
algorithm would become infeasible.

Equation (\ref{ae}) is a misspecified log-likelihood function. First, the
state process $\{z_{t},t=1,...,T\}$\ is not specified yet, and the
probability $\Pr (z_{1},...,z_{T})$\ depends on how we model the state
process. Second, similar to the principal component estimator in Stock and
Watson (2002) and Bai (2003), equation (\ref{aq}) ignores the
cross-sectional and serial dependence and heteroscedasticity of the error
term. We may also take into account the heteroscedasticity as Doz, Giannone
and Reichlin (2012) and Bai and Li (2012, 2016). With regime switching, the
algorithm and the asymptotic analysis would be much more complicated, but
the results should be conceptually similar.

Third, the factor dynamics are ignored. As shown in Bai (2003) for PCE and
in Bai and Li (2012, 2016) for MLE, when there is no regime switching, the
asymptotic properties of the estimated factors and the estimated loadings
are robust to the presence of the factor dynamics if both $N$ and $T$ are
large. We shall show in Section \ref{asym} that when there is regime
switching, the asymptotic results are also robust to the presence of the
factor dynamics. More importantly, ignoring the factor dynamics greatly
simplifies the computation algorithm for regime switching factor models. As
shown below, with factor dynamics ignored, we just need to calculate
principal components iteratively. If the factor dynamics are not ignored,
Kim (1994)'s method would be very time consuming and may have convergence
problems if $N$\ is large\footnote{%
When there is no regime switching, as suggested by Doz et al. (2012), large $%
N$ factor model with factor dynamics can be calculated by the EM algorithm.
However, when there are both regime switching and factor dynamics, the EM
algorithm also fails. This is because in the E-step we need to calculate the
likelihood for each possible state chain $z_{1},....,z_{T}$ and there are $%
(J^{0})^{T}$ possibilities, and in the M-step numerical optimization is
still needed.{}}.

Fourth, equation (\ref{ae}) implicitly assumes that $\mathbb{E}(f_{t}^{0})=0$
and $\mathbb{E}(f_{t}^{0}f_{t}^{0\prime })$ is stable within each regime,
and $\mathbb{E}(f_{t}^{0}f_{t}^{0\prime })$ is absorbed into $\Lambda
_{j}\Lambda _{j}^{\prime }$ in equation (\ref{aq}). This does not matter,
since all results of this paper still hold when $\mathbb{E}(f_{t}^{0})\neq 0$
and $\mathbb{E}(f_{t}^{0}f_{t}^{0\prime })$ is unstable within regime, as
long as Assumption \ref{factors} is satisfied.

\begin{description}
\item[First order conditions for $\Lambda $ and $\protect\sigma ^{2}$] 
\end{description}

The parameters $\Lambda $ and $\sigma ^{2}$ are estimated by maximizing $%
l(\Lambda ,\sigma ^{2})$. Define $x_{1:t}\equiv (x_{1},...,x_{t})$ and $%
z_{1:t}\equiv (z_{1},...,z_{t})$, and let $p_{tj\left\vert T\right. }\equiv
\Pr (z_{t}=j\left\vert x_{1:T};\Lambda ,\sigma ^{2}\right. )$ denote the
probability of $z_{t}=j$ conditional on $x_{1:T}$. Based on equation (\ref%
{ae}), it can be easily verified that%
\begin{eqnarray}
\frac{\partial l(\Lambda ,\sigma ^{2})}{\partial \Lambda _{j}}
&=&\sum\nolimits_{t=1}^{T}\frac{\partial \log L(x_{t}\left\vert
z_{t}=j;\Lambda _{j},\sigma ^{2}\right. )}{\partial \Lambda _{j}}%
p_{tj\left\vert T\right. }  \notag \\
&=&\sum\nolimits_{t=1}^{T}p_{tj\left\vert T\right. }(-\Sigma
_{j}^{-1}\Lambda _{j}+\Sigma _{j}^{-1}x_{t}x_{t}^{\prime }\Sigma
_{j}^{-1}\Lambda _{j}),  \label{bw} \\
\frac{\partial l(\Lambda ,\sigma ^{2})}{\partial \sigma ^{2}}
&=&\sum\nolimits_{t=1}^{T}\sum\nolimits_{j=1}^{J^{0}}\frac{\partial \log
L(x_{t}\left\vert z_{t}=j;\Lambda _{j},\sigma ^{2}\right. )}{\partial \sigma
^{2}}p_{tj\left\vert T\right. },
\end{eqnarray}%
where equation (\ref{bw}) follows from%
\begin{eqnarray}
\frac{\partial \log \left\vert \Sigma _{j}\right\vert }{\partial \Lambda _{j}%
} &=&2\Sigma _{j}^{-1}\Lambda _{j},  \label{af} \\
\frac{\partial x_{t}^{\prime }\Sigma _{j}^{-1}x_{t}}{\partial \Lambda _{j}}
&=&-2\Sigma _{j}^{-1}x_{t}x_{t}^{\prime }\Sigma _{j}^{-1}\Lambda _{j},
\label{ag}
\end{eqnarray}%
see Chapter 14.3 in Andersen (2003) for the details on calculating these
derivatives. Set $\frac{\partial l(\Lambda ,\sigma ^{2})}{\partial \Lambda
_{j}}$ to $0$, we have%
\begin{eqnarray}
\Sigma _{j}^{-1}\Lambda _{j} &=&\Sigma _{j}^{-1}S_{j}\Sigma _{j}^{-1}\Lambda
_{j},  \label{a} \\
\text{and }S_{j} &=&\sum\nolimits_{t=1}^{T}p_{tj\left\vert T\right.
}x_{t}x_{t}^{\prime }/\sum\nolimits_{t=1}^{T}p_{tj\left\vert T\right. }. 
\notag
\end{eqnarray}%
$S_{j}$ can be considered as sample covariance matrix for $\Sigma _{j}$
based on importance sampling. The weights $p_{tj\left\vert T\right.
}/\sum\nolimits_{t=1}^{T}p_{tj\left\vert T\right. }$ depend on the
importance of the sample $x_{t}$ for regime $j$, the larger $p_{tj\left\vert
T\right. }$ is, the more important $x_{t}$ is for regime $j$.

From equation (\ref{aq}), we have $\Sigma _{j}\Lambda _{j}=\Lambda
_{j}(\Lambda _{j}^{\prime }\Lambda _{j}+\sigma ^{2}I_{r_{j}^{0}})$. Left
multiply $S_{j}\Sigma _{j}^{-1}$ on both sides, we have $S_{j}\Lambda
_{j}=S_{j}\Sigma _{j}^{-1}\Lambda _{j}(\Lambda _{j}^{\prime }\Lambda
_{j}+\sigma ^{2}I_{r_{j}^{0}})$. From equation (\ref{a}), we have $\Lambda
_{j}=S_{j}\Sigma _{j}^{-1}\Lambda _{j}$, thus%
\begin{equation}
S_{j}\Lambda _{j}=\Lambda _{j}(\Lambda _{j}^{\prime }\Lambda _{j}+\sigma
^{2}I_{r_{j}^{0}}).  \label{foc1}
\end{equation}%
If $\Lambda _{j}$ is a solution for equation (\ref{foc1}) and $\Lambda
_{j}^{\ast }$ equals post-multiplying $\Lambda _{j}$ by the eigenvector
matrix of $\Lambda _{j}^{\prime }\Lambda _{j}$, then $\Lambda _{j}^{\ast }$
is also a solution for equation (\ref{foc1}) and $\Lambda _{j}^{\ast \prime
}\Lambda _{j}^{\ast }$ is diagonal. Thus we can directly choose the solution 
$\Lambda _{j}$ with $\Lambda _{j}^{\prime }\Lambda _{j}$ being diagonal. It
follows that the solution $\Lambda _{j}$ is the eigenvectors of $S_{j}$ and $%
\Lambda _{j}^{\prime }\Lambda _{j}+\sigma ^{2}I_{r_{j}^{0}}$ is the
corresponding eigenvalues. We show in Appendix \ref{secfoc} that $\sigma
^{2} $ satisfies the following condition:%
\begin{equation}
\sigma ^{2}=\frac{1}{N}tr(\frac{1}{T}\sum\nolimits_{t=1}^{T}x_{t}x_{t}^{%
\prime }-\sum\nolimits_{j=1}^{J^{0}}\frac{1}{T}\sum\nolimits_{t=1}^{T}p_{tj%
\left\vert T\right. }\Lambda _{j}\Lambda _{j}^{\prime }).  \label{ca}
\end{equation}%
Note that we do not need to specify the state process $\{z_{1},...,z_{T}\}$\
when deriving the first order conditions (\ref{foc1}) and (\ref{ca}), and
different models of $\{z_{1},...,z_{T}\}$\ correspond to different ways of
calculating $p_{tj\left\vert T\right. }$. In the EM algorithm presented
below, we consider $\{z_{1},...,z_{T}\}$\ as a Markov process regardless of
what the true process of $\{z_{1},...,z_{T}\}$\ is.

\subsection{EM Algorithm\label{EM}}

Let $q^{0}=(q_{1}^{0},...,q_{J^{0}}^{0})^{\prime }$ denote the unconditional
regime probabilities, $\phi ^{0}=(\phi _{1}^{0},...,\phi
_{J^{0}}^{0})^{\prime }$ denote the initial probabilities of $z_{1}$, $Q^{0}$
denote the $(J^{0}\times J^{0})$ matrix of transition probabilities and $%
Q_{jk}^{0}$ denote the probability of switching from state $k$ to state $j$.
If there are no superscripts, $q$, $Q$ and $\phi $ denote parameters as
variables.

For any given $Q$ and $\phi $, at the $h$-th iteration, let $\tilde{\Lambda}%
^{(h)}$ denote the estimated loadings, $\tilde{\sigma}^{2(h)}$ denote the
estimated variance, and $\Pr (z_{1},...,z_{T}\left\vert x_{1:T};\tilde{\theta%
}^{(h)}\right. )$ denote the probability of $z_{1:T}$ conditioning on $%
x_{1:T}$ and evaluated at $\tilde{\theta}^{(h)}=(\tilde{\Lambda}^{(h)},%
\tilde{\sigma}^{2(h)},Q,\phi )$. The EM algorithm maximizes the expectation
of the log-likelihood of $(x_{1:T},z_{1:T})$ with respect to $\Pr
(z_{1},...,z_{T}\left\vert x_{1:T};\tilde{\theta}^{(h)}\right. )$, i.e.,%
\begin{eqnarray*}
l^{(h)}(\Lambda ,\sigma ^{2},Q,\phi ) &\equiv
&\sum\nolimits_{z_{T}=1}^{J^{0}}...\sum\nolimits_{z_{1}=1}^{J^{0}}\log
[\prod\nolimits_{t=1}^{T}L(x_{t}\left\vert z_{t};\Lambda ,\sigma ^{2}\right.
)\Pr (z_{1},...,z_{T}\left\vert Q,\phi \right. )] \\
&&\times \Pr (z_{1},...,z_{T}\left\vert x_{1:T};\tilde{\theta}^{(h)}\right.
).
\end{eqnarray*}%
Considering $z_{t}$ as a Markov process,\textbf{\ }$\Pr
(z_{1},...,z_{T}\left\vert Q,\phi \right. )=\Pr (z_{1}\left\vert \phi
\right. )\prod\nolimits_{t=2}^{T}\Pr (z_{t}\left\vert z_{t-1};Q\right. )$.
Thus 
\begin{eqnarray}
l^{(h)}(\Lambda ,\sigma ^{2},Q,\phi )
&=&\sum\nolimits_{z_{T}=1}^{J^{0}}...\sum\nolimits_{z_{1}=1}^{J^{0}}[\sum%
\nolimits_{t=1}^{T}\log L(x_{t}\left\vert z_{t};\Lambda ,\sigma ^{2}\right. )
\notag \\
&&+\sum\nolimits_{t=2}^{T}\log \Pr (z_{t}\left\vert z_{t-1};Q\right. )+\log
\Pr (z_{1}\left\vert \phi \right. )]\Pr (z_{1},...,z_{T}\left\vert x_{1:T};%
\tilde{\theta}^{(h)}\right. )  \notag \\
&=&\sum\nolimits_{t=1}^{T}\sum\nolimits_{j=1}^{J^{0}}\log L(x_{t}\left\vert
z_{t}=j;\Lambda _{j},\sigma ^{2}\right. )\tilde{p}_{tj\left\vert T\right.
}^{(h)}  \notag \\
&&+\sum\nolimits_{t=2}^{T}\sum\nolimits_{j=1}^{J^{0}}\sum%
\nolimits_{k=1}^{J^{0}}\log Q_{jk}\tilde{p}_{tjk\left\vert T\right.
}^{(h)}+\sum\nolimits_{k=1}^{J^{0}}\log \phi _{k}\tilde{p}_{1k\left\vert
T\right. }^{(h)},  \label{ap}
\end{eqnarray}%
where $\tilde{p}_{tjk\left\vert T\right. }^{(h)}=\Pr
(z_{t}=j,z_{t-1}=k\left\vert x_{1:T};\tilde{\theta}^{(h)}\right. )$ and $%
\tilde{p}_{tj\left\vert T\right. }^{(h)}=\Pr (z_{t}=j\left\vert x_{1:T};%
\tilde{\theta}^{(h)}\right. )=\sum\nolimits_{k=1}^{J^{0}}\tilde{p}%
_{tjk\left\vert T\right. }^{(h)}$ are the smoothed probabilities based on $%
x_{1:T}$ and $\tilde{\theta}^{(h)}$. Appendix \ref{ptjkT} presents a
recursive algorithm for calculating $\tilde{p}_{tjk\left\vert T\right.
}^{(h)}$. From equations (\ref{af}) and (\ref{ag}), we have $\frac{\partial
\log L(x_{t}\left\vert z_{t}=j;\Lambda _{j},\sigma ^{2}\right. )}{\partial
\Lambda _{j}}=-\Sigma _{j}^{-1}\Lambda _{j}+\Sigma
_{j}^{-1}x_{t}x_{t}^{\prime }\Sigma _{j}^{-1}\Lambda _{j}$. Thus 
\begin{equation*}
\frac{\partial \sum\nolimits_{t=1}^{T}\log L(x_{t}\left\vert z_{t}=j;\Lambda
_{j},\sigma ^{2}\right. )\tilde{p}_{tj\left\vert T\right. }^{(h)}}{\partial
\Lambda _{j}}=\sum\nolimits_{t=1}^{T}(-\Sigma _{j}^{-1}\Lambda _{j}+\Sigma
_{j}^{-1}x_{t}x_{t}^{\prime }\Sigma _{j}^{-1}\Lambda _{j})\tilde{p}%
_{tj\left\vert T\right. }^{(h)}=0,
\end{equation*}%
and it follows that%
\begin{eqnarray}
\Sigma _{j}^{-1}\Lambda _{j} &=&\Sigma _{j}^{-1}\tilde{S}_{j}^{(h)}\Sigma
_{j}^{-1}\Lambda _{j},  \label{ah} \\
\text{where }\tilde{S}_{j}^{(h)} &=&\sum\nolimits_{t=1}^{T}\tilde{p}%
_{tj\left\vert T\right. }^{(h)}x_{t}x_{t}^{\prime }/\sum\nolimits_{t=1}^{T}%
\tilde{p}_{tj\left\vert T\right. }^{(h)}.  \notag
\end{eqnarray}%
Similar to equation (\ref{foc1}), equation (\ref{ah}) implies that%
\begin{equation}
\tilde{S}_{j}^{(h)}\tilde{\Lambda}_{j}^{(h+1)}=\tilde{\Lambda}_{j}^{(h+1)}(%
\tilde{\Lambda}_{j}^{(h+1)\prime }\tilde{\Lambda}_{j}^{(h+1)}+\tilde{\sigma}%
^{2(h+1)}I_{r_{j}^{0}}),  \label{an}
\end{equation}%
thus the columns of $\tilde{\Lambda}_{j}^{(h+1)}$ are the eigenvectors of $%
\tilde{S}_{j}^{(h)}$ and the diagonal elements of $\tilde{\Lambda}%
_{j}^{(h+1)\prime }\tilde{\Lambda}_{j}^{(h+1)}+\tilde{\sigma}%
^{2(h+1)}I_{r_{j}^{0}}$ are the corresponding eigenvalues. To save space, we
show in Appendix \ref{secfoc} that 
\begin{equation}
\tilde{\sigma}^{2(h+1)}=\frac{1}{N}tr(\frac{1}{T}\sum%
\nolimits_{t=1}^{T}x_{t}x_{t}^{\prime }-\sum\nolimits_{j=1}^{J^{0}}\frac{1}{T%
}\sum\nolimits_{t=1}^{T}\tilde{p}_{tj\left\vert T\right. }^{(h)}\tilde{%
\Lambda}_{j}^{(h+1)}\tilde{\Lambda}_{j}^{(h+1)\prime }).  \label{ao}
\end{equation}

\begin{remark}
The second equality of equation (\ref{ap}) is crucial. Since factor dynamics
are ignored, $L(x_{1:T}\left\vert z_{1:T};\Lambda ,\sigma ^{2}\right.
)=\prod\nolimits_{t=1}^{T}L(x_{t}\left\vert z_{t};\Lambda ,\sigma
^{2}\right. )$, thus we only need to calculate $\tilde{p}_{tj\left\vert
T\right. }^{(h)}$ rather than the probability of the whole chain $\Pr
(z_{1},...,z_{T}\left\vert x_{1:T};\tilde{\theta}^{(h)}\right. )$. The
latter requires $(J^{0})^{T}$ calculations, which is hopeless when $T$ is
large. If factor dynamics are not ignored, then $L(x_{1:T}\left\vert
z_{1:T};\Lambda ,\sigma ^{2}\right. )=L(x_{1}\left\vert z_{1:T};\Lambda
,\sigma ^{2}\right. )\prod\nolimits_{t=2}^{T}L(x_{t}\left\vert
x_{1:t-1},z_{1:T};\Lambda ,\sigma ^{2}\right. )$. $L(x_{t}\left\vert
x_{1:t-1},z_{1:T};\Lambda ,\sigma ^{2}\right. )$ depends on the chain $%
(z_{1},...,z_{T})$ through $z_{1:t}$, thus we need to calculate $\Pr
(z_{1:t}\left\vert x_{1:T};\tilde{\theta}^{(h)}\right. )$. This requires $%
(J^{0})^{t}$ calculations, which is hopeless when $t$ is large.
\end{remark}

\begin{description}
\item[EM algorithm for $\Lambda $ and $\protect\sigma ^{2}$] 
\end{description}

\textit{Choose any }$Q$\textit{\ and }$\phi $\textit{\ such that }$Q_{jk}>0$%
\textit{\ for any }$j$\textit{\ and }$k$\textit{\ and }$\phi _{k}>0$\textit{%
\ for all }$k$. \textit{Start from randomly generated initial values of }$%
\tilde{\Lambda}^{(0)}$ \textit{and }$\tilde{\sigma}^{2(0)}=1$\textit{. For }$%
h=0,1,...,$

\textit{(E-step): calculate }$\tilde{p}_{tjk\left\vert T\right. }^{(h)}$%
\textit{\ using the algorithm in Appendix \ref{ptjkT}, and calculate }$%
\tilde{S}_{j}^{(h)}=\sum\nolimits_{t=1}^{T}\tilde{p}_{tj\left\vert T\right.
}^{(h)}x_{t}x_{t}^{\prime }/\sum\nolimits_{t=1}^{T}\tilde{p}_{tj\left\vert
T\right. }^{(h)}$ with $\tilde{p}_{tj\left\vert T\right.
}^{(h)}=\sum\nolimits_{k=1}^{J^{0}}\tilde{p}_{tjk\left\vert T\right. }^{(h)}$%
\textit{;}

\textit{(M-step): given }$\tilde{p}_{tjk\left\vert T\right. }^{(h)}$\textit{%
\ and }$\tilde{S}_{j}^{(h)}$\textit{, calculate }$\tilde{\Lambda}%
_{j}^{(h+1)} $\textit{\ as the eigenvectors of }$\tilde{S}_{j}^{(h)}$\textit{%
\ corresponding to the }$r_{j}^{0}$\textit{\ largest eigenvalues, and then
normalize }$\tilde{\Lambda}_{j}^{(h+1)}$\textit{\ such that }$\left\Vert 
\tilde{\Lambda}_{jl}^{(h+1)}\right\Vert ^{2}+\tilde{\sigma}^{2(h+1)}$\textit{%
\ equals the }$l$\textit{-th largest eigenvalue of }$\tilde{S}_{j}^{(h)}$%
\textit{\ for }$l=1,...,r_{j}^{0}$\textit{\ and equation (\ref{ao}) is also
satisfied, where }$\tilde{\Lambda}_{jl}^{(h+1)}$\textit{\ is the }$l$\textit{%
-th column of }$\tilde{\Lambda}_{j}^{(h+1)}$\textit{. Note that the
computation of }$\left\Vert \tilde{\Lambda}_{jl}^{(h+1)}\right\Vert ^{2}$%
\textit{\ and }$\tilde{\sigma}^{2(h+1)}$\textit{\ requires iteration between
equations (\ref{an}) and (\ref{ao}).}

\textit{Iterate the E-step and the M-step until converge. Let }$\tilde{%
\Lambda}_{j}=(\tilde{\lambda}_{j1},...,\tilde{\lambda}_{jN})^{\prime }$%
\textit{, }$\tilde{\Lambda}=(\tilde{\Lambda}_{1},...,\tilde{\Lambda}%
_{J^{0}}) $\textit{\ and }$\tilde{\sigma}^{2}$\textit{\ denote the estimated
parameters, and let }$\tilde{p}_{tj\left\vert T\right. }$\textit{\ and }$%
\tilde{p}_{tjk\left\vert T\right. }$\textit{\ denote the smoothed
probabilities based on }$x_{1:T}$\textit{\ and }$(\tilde{\Lambda},\tilde{%
\sigma}^{2},Q,\phi )$\textit{. }

\bigskip

A special case of the above EM algorithm is when we choose $\phi =q$\ and $%
Q=q1_{J^{0}}^{\prime }$\ ($1_{J^{0}}$\ denotes the $J^{0}\times 1$\ vector
of ones), i.e., we consider $\{z_{1},...,z_{T}\}$\ as an independent
process. For this case, the computation of $\tilde{p}_{tj\left\vert T\right.
}^{(h)}$\ is simplified because the unsmoothed regime probabilities can be
calculated directly by%
\begin{equation*}
\tilde{p}_{tj\left\vert T\right. }^{(h)}=q_{j}L(x_{t}\left\vert z_{t}=j;%
\tilde{\Lambda}_{j}^{(h)},\tilde{\sigma}^{2(h)}\right.
)/\sum\nolimits_{k=1}^{J^{0}}q_{k}L(x_{t}\left\vert z_{t}=k;\tilde{\Lambda}%
_{k}^{(h)},\tilde{\sigma}^{2(h)}\right. ).
\end{equation*}%
This case is preferable if we knew the true process of $\{z_{1},...,z_{T}\}$%
\ is independent. In practice, since the state process of the business
cycle/stock market is highly persistent, smoothed regime probabilities that
capture the persistence should perform significantly better, especially when
mixed frequency data or ragged edge data (data released at non-synchronized
dates) are used. The asymptotic results in Section \ref{AsyDyn} and Section %
\ref{rJ} hold for any $Q$\ and $\phi $\ as long as $\phi _{j}>0$\ for any $j$%
\ and $Q_{jk}>0$\ for any $j$\ and $k$, i.e., they hold for both the
smoothed algorithm and the unsmoothed algorithm.

If the true process of $\{z_{1},...,z_{T}\}$ is Markovian, $Q_{jk}^{0}$ and $%
\phi _{k}^{0}$ can be estimated by%
\begin{eqnarray}
\tilde{Q}_{jk} &=&\sum\nolimits_{t=2}^{T}\tilde{p}_{tjk\left\vert T\right.
}/\sum\nolimits_{j=1}^{J^{0}}\sum\nolimits_{t=2}^{T}\tilde{p}_{tjk\left\vert
T\right. },  \label{Qtilde-jk} \\
\tilde{\phi}_{k} &=&\tilde{p}_{1k\left\vert T\right.
}=\sum\nolimits_{j=1}^{J^{0}}\tilde{p}_{2jk\left\vert T\right. }.
\end{eqnarray}%
We can also plug $\tilde{Q}_{jk}$ and $\tilde{\phi}_{k}$ back in the above
EM algorithm and iterate between $(\tilde{\Lambda},\tilde{\sigma}^{2})$ and $%
(\tilde{Q},\tilde{\phi})$ until convergence. This is the maximum likelihood
estimator when $(Q,\phi )$ is estimated jointly with $(\Lambda ,\sigma ^{2})$%
, see Appendix \ref{secfoc} for details.

The asymptotic results in Section \ref{AsyDyn} and Section \ref{rJ} also
hold as long as $\tilde{\sigma}^{2}$ is bounded and bounded away from zero
in probability. Consistency of $\tilde{\sigma}^{2}$ is not needed. We could
restrict $\tilde{\sigma}^{2}$ in $[\frac{1}{C^{2}},C^{2}]$ for some large $C$
or simply fix down $\tilde{\sigma}^{2}=1$ to avoid the iteration between $%
\tilde{\Lambda}_{j}^{(h+1)}$ and $\tilde{\sigma}^{2(h+1)}$. This only
affects the Euclidean norm of the columns of $\tilde{\Lambda}_{j}^{(h+1)}$.

\begin{remark}
Pelger and Xiong (2021) also considers the model\footnote{%
We changed Pelger and Xiong (2021)'s notation to our notation for better
comparison.} $x_{t}=\Lambda (z_{t})f_{t}^{0}+e_{t}$. The state variable $%
z_{t}$\ is discrete and unobservable in this paper, while in Pelger and
Xiong (2021) $z_{t}$\ is continuous and observable. Also, in this paper $%
\tilde{\Lambda}_{j}$\ are eigenvectors of $\tilde{S}_{j}=\frac{1}{%
\sum\nolimits_{t=1}^{T}\tilde{p}_{tj\left\vert T\right. }}%
\sum\nolimits_{t=1}^{T}\tilde{p}_{tj\left\vert T\right. }x_{t}x_{t}^{\prime
} $, while in Pelger and Xiong (2021) $\hat{\Lambda}(s)$\ are eigenvectors
of $\frac{1}{\sum\nolimits_{t=1}^{T}K_{s}(z_{t})}\sum%
\nolimits_{t=1}^{T}K_{s}(z_{t})x_{t}x_{t}^{\prime }$, where $K_{s}(z_{t})=%
\frac{1}{h}K(\frac{z_{t}-s}{h})$\ is the kernel function. The key difference
is that the weight $K_{s}(z_{t})$\ is observable because $z_{t}$\ is
observable in Pelger and Xiong (2021), but in this paper the weight $\tilde{p%
}_{tj\left\vert T\right. }$\ is unobservable and need to be estimated
jointly with $\Lambda _{j}$.
\end{remark}

\begin{remark}
We can take into account cross-sectional heteroscedasticity as Bai and Li
(2012, 2016) by replacing equation (\ref{aq}) by $\Sigma _{j}=\Lambda
_{j}\Lambda _{j}^{\prime }+\Sigma _{e}$, where $\Sigma _{e}$\ is a $N\times
N $\ diagonal matrix. We show in Appendix \ref{secfoc} that the first order
conditions are%
\begin{eqnarray}
\Sigma _{e}^{-\frac{1}{2}}S_{j}\Sigma _{e}^{-1}\Lambda _{j} &=&\Sigma _{e}^{-%
\frac{1}{2}}\Lambda _{j}(\Lambda _{j}^{\prime }\Sigma _{e}^{-1}\Lambda
_{j}+I_{r_{j}^{0}}),  \label{ac} \\
\Sigma _{e} &=&diag(\frac{1}{T}\sum\nolimits_{t=1}^{T}x_{t}x_{t}^{\prime
}-\sum\nolimits_{j=1}^{J^{0}}\frac{1}{T}\sum\nolimits_{t=1}^{T}p_{tj\left%
\vert T\right. }\Lambda _{j}\Lambda _{j}^{\prime }),  \label{ad}
\end{eqnarray}%
i.e., columns of $\Sigma _{e}^{-\frac{1}{2}}\Lambda _{j}$\ are the
eigenvectors of $\Sigma _{e}^{-\frac{1}{2}}S_{j}\Sigma _{e}^{-\frac{1}{2}}$\
and diagonal elements of $\Lambda _{j}^{\prime }\Sigma _{e}^{-1}\Lambda
_{j}+I_{r_{j}^{0}}$\ are the corresponding eigenvalues. Accordingly, in the
M-step of the EM algorithm, we iterate%
\begin{eqnarray*}
\tilde{\Sigma}_{e}^{-\frac{1}{2}(h)}\tilde{S}_{j}^{(h)}\tilde{\Sigma}%
_{e}^{-1(h)}\tilde{\Lambda}_{j}^{(h+1)} &=&\tilde{\Sigma}_{e}^{-\frac{1}{2}%
(h)}\tilde{\Lambda}_{j}^{(h+1)}(\tilde{\Lambda}_{j}^{(h+1)\prime }\tilde{%
\Sigma}_{e}^{-1(h)}\tilde{\Lambda}_{j}^{(h+1)}+I_{r_{j}^{0}}), \\
\text{and }\tilde{\Sigma}_{e}^{(h+1)} &=&diag(\frac{1}{T}\sum%
\nolimits_{t=1}^{T}x_{t}x_{t}^{\prime }-\sum\nolimits_{j=1}^{J^{0}}\frac{1}{T%
}\sum\nolimits_{t=1}^{T}\tilde{p}_{tj\left\vert T\right. }^{(h)}\tilde{%
\Lambda}_{j}^{(h+1)}\tilde{\Lambda}_{j}^{(h+1)\prime }).
\end{eqnarray*}%
The other steps of the EM algorithm remain unchanged. If we further take
into account cross-sectional dependence, then $\Sigma _{j}=\Lambda
_{j}\Lambda _{j}^{\prime }+\Sigma _{e}$\ and $\Sigma _{e}$\ is non-diagonal.
It can be verified that for this case equation (\ref{ac}) is still valid,
but equation (\ref{ad}) is not. Since $\Sigma _{e}$\ is of dimension $%
N\times N$\ and $N\rightarrow \infty $\ jointly with $T$, certain sparsity
condition has to be imposed on $\Sigma _{e}$\ to consistently estimate $%
\Sigma _{e}$. Results on this topic are very rare (if any) even for factor
model with single regime.
\end{remark}

\subsection{Estimate the Factors}

If the factor dynamics are taken into account, the expectation of $f_{t}$
conditioning on $x_{1:t}$ is%
\begin{equation*}
\sum\nolimits_{z_{1}=1}^{J^{0}}...\sum\nolimits_{z_{t}=1}^{J^{0}}\mathbb{E}%
(f_{t}\left\vert x_{1:t},z_{1:t};\tilde{\Lambda},\tilde{\sigma}^{2}\right.
)\Pr (z_{1:t}\left\vert x_{1:t};\tilde{\Lambda},\tilde{\sigma}^{2},Q,\phi
\right. ),
\end{equation*}%
which is formidable since we need to calculate $\Pr (z_{1:t}\left\vert
x_{1:t};\tilde{\Lambda},\tilde{\sigma}^{2},Q,\phi \right. )$ for each
possible $z_{1:t}$, i.e., we need to calculate $(J^{0})^{t}$ probabilities.
For large $N$, the benefit of considering factor dynamics is marginal and
outweighed by the computational simplicity of ignoring factor dynamics. If
the factor dynamics are ignored, the expectation of $f_{t}$ conditioning on $%
x_{1:T}$ is%
\begin{eqnarray}
\tilde{f}_{t} &=&\sum\nolimits_{j=1}^{J^{0}}\mathbb{E}(f_{t}\left\vert
x_{1:T},z_{t}=j;\tilde{\Lambda}_{j},\tilde{\sigma}^{2}\right. )\tilde{p}%
_{tj\left\vert T\right. }=\sum\nolimits_{j=1}^{J^{0}}\mathbb{E}%
(f_{t}\left\vert x_{t},z_{t}=j;\tilde{\Lambda}_{j},\tilde{\sigma}^{2}\right.
)\tilde{p}_{tj\left\vert T\right. }  \notag \\
&=&\sum\nolimits_{j=1}^{J^{0}}\tilde{\Lambda}_{j}^{\prime }(\tilde{\Lambda}%
_{j}\tilde{\Lambda}_{j}^{\prime }+\tilde{\sigma}^{2}I_{N})^{-1}x_{t}\tilde{p}%
_{tj\left\vert T\right. }.  \label{ftilde-t}
\end{eqnarray}%
Note that the dimension of $\tilde{\Lambda}_{j}^{\prime }(\tilde{\Lambda}_{j}%
\tilde{\Lambda}_{j}^{\prime }+\tilde{\sigma}^{2}I_{N})^{-1}x_{t}$\ is
different across $j$\ if $r_{j}^{0}$\ is different across $j$. Here and also
in the proof of Theorem \ref{factor}, when we add two vectors of different
dimensions, we implicitly augment the vector of smaller dimension with zeros
to make the dimensions of these two vectors equal. Thus $\tilde{f}_{t}$\ is
a $\max r_{j}^{0}$\ dimensional vector.

\section{Asymptotic Results\label{asym}}

\subsection{Assumptions\label{assu}}

We assume the following conditions hold as $(N,T)\rightarrow \infty $. These
conditions are mainly Assumptions A-G in Bai (2003) adapted to the current
regime switching setup.

\begin{assu}
\label{factors}(1) For $j=1,...,J^{0}$, $\frac{1}{Tq_{j}^{0}}%
\sum\nolimits_{t=1}^{T}f_{t}^{0}f_{t}^{0\prime }1_{z_{t}=j}\overset{p}{%
\rightarrow }\Sigma _{F_{j}}$\ for some positive definite $\Sigma _{F_{j}}$,
and $plim\frac{1}{\left\vert A_{j}\right\vert }\sum\nolimits_{t\in
A_{j}}f_{t}^{0}f_{t}^{0\prime }$ is also positive definite, where $A_{j}$ is
defined in section \ref{id}.

(2) For some $\alpha >16$, there exists $M>0$\ such that $\mathbb{E}%
(\left\Vert f_{t}^{0}\right\Vert ^{\alpha })\leq M$\ for all $t$.
\end{assu}

Assumption \ref{factors} corresponds to Assumption A in Bai (2003).
Assumption \ref{factors}(1) rules out the possibility that for regime $j$,
the subsample $\{t:z_{t}=j\}$ can be further decomposed into multiple
regimes, see the discussion in Section \ref{id}. The factor process is
allowed to be dynamic such that $C(L)f_{t}=\epsilon _{t}$. Assumption \ref%
{factors}(2) assumes that the factors have bounded moments.

\begin{assu}
\label{loadings}(1) For $j=1,...,J^{0}$, $\frac{1}{N}\Lambda _{j}^{0\prime
}\Lambda _{j}^{0}\rightarrow \Sigma _{\Lambda _{j}}$\ for some positive
definite $\Sigma _{\Lambda _{j}}$ and $\left\Vert \lambda
_{ji}^{0}\right\Vert \leq M$\ for any $i=1,...,N$.\textbf{\ }

(2) For any $j=1,...,J^{0}$\ and $k=1,...,J^{0}$, $\min_{t}\frac{1}{N}%
f_{t}^{0\prime }\Lambda _{j}^{0\prime }M_{\Lambda _{k}^{0}}\Lambda
_{j}^{0}f_{t}^{0}\geq C$ for some $C>0$.
\end{assu}

Assumption \ref{loadings}(1) corresponds to Assumption B in Bai (2003).
Assumption \ref{loadings}(1) ensures that each factor has a nontrivial
contribution within each regime, and $\left\Vert \lambda
_{ji}^{0}\right\Vert $ is assumed to be uniformly bounded over $i$.
Assumption \ref{loadings}(2) is the identification condition for determining
which regime each $x_{t}$ belongs to, see Section \ref{id} for details on
the implication of this condition.

\begin{assu}
\label{error}(1) $\mathbb{E}(e_{it})=0$, $\mathbb{E}(e_{it}^{\alpha })\leq M$
for some $\alpha >16$.

(2) $\sum\nolimits_{k=1}^{N}\tau _{ik}\leq M$ for any $i$, where $\mathbb{E}%
(e_{it}e_{kt})=\tau _{ik,t}$ with $\left\vert \tau _{ik,t}\right\vert \leq
\tau _{ik}$ for some $\tau _{ik}>0$ and for all $t$.

(3) $\sum\nolimits_{s=1}^{T}\gamma _{ts}\leq M$ for all $t$, where $%
E(e_{it}e_{is})=\gamma _{i,ts}$ with $\left\vert \gamma _{i,ts}\right\vert
\leq \gamma _{ts}$ for some $\gamma _{ts}>0$ and for all $i$.

(4) $\mathbb{E}(\left\Vert \frac{1}{\sqrt{T}}\sum%
\nolimits_{t=1}^{T}(e_{it}e_{kt}-\mathbb{E(}e_{it}e_{kt}))1_{z_{t}=j}\right%
\Vert ^{2})\leq M$\ for all $i=1,...,N$, $k=1,...,N$ and $j=1,...,J^{0}$.
\end{assu}

Assumption \ref{error} is modified slightly from Assumption C in Bai (2003).
The error term is allowed to have limited cross-sectional and serial
dependence as well as heteroscedasticity.

\begin{assu}
\label{state}For $j=1,...,J^{0}$, $\frac{1}{T}\sum%
\nolimits_{t=1}^{T}1_{z_{t}=j}\overset{p}{\rightarrow }q_{j}^{0}$ and $%
0<q_{j}^{0}<1$.
\end{assu}

The asymptotic results in Section \ref{AsyDyn} and Section \ref{rJ} are
valid as long as Assumption \ref{state} holds and the other assumptions in
this section hold conditioning on $\{z_{t},t=1,...,T\}$. Thus the state
process $\{z_{t},t=1,...,T\}$\ is allowed to be non-Markovian and correlated
with $f_{s}^{0}$\ and $e_{is}$\ for all $i$\ and $s$. Knowledge of the true
state process is not needed.

\begin{assu}
\label{moments}(1)For some $\beta \geq 2$, $\mathbb{E}(\left\Vert \frac{1}{%
\sqrt{N}}\sum\nolimits_{i=1}^{N}\lambda _{ji}^{0}e_{it}\right\Vert ^{\beta
})\leq M$ for all $j=1,...,J^{0}$ and all $t$.

(2) $\mathbb{E}(\left\Vert \frac{1}{\sqrt{T}}\sum%
\nolimits_{t=1}^{T}f_{t}^{0}e_{it}1_{z_{t}=j}\right\Vert ^{2})\leq M$ for
all $j=1,...,J^{0}$ and all $i$.
\end{assu}

Assumption \ref{moments} is modified slightly from Assumption D in Bai
(2003). Assumption \ref{moments}(1) assumes that the errors are weakly
correlated across $i$ for each $t$. When $\beta =2$, Assumption \ref{moments}%
(1) is implied by Assumptions \ref{loadings}(1), \ref{error}(1) and \ref%
{error}(2). Assumption \ref{moments}(2) assumes that the errors are weakly
correlated across $t$ for each $i$. Assumption \ref{moments}(2) is implied
by Assumptions \ref{factors}(2), \ref{error}(1) and \ref{error}(4) if we
further assume the factors are nonrandom or independent with the errors.

\begin{assu}
\label{diff-eig}For each $j=1,...,J^{0}$, the eigenvalues of $\Sigma
_{\Lambda _{j}}^{\frac{1}{2}}\Sigma _{F_{j}}\Sigma _{\Lambda _{j}}^{\frac{1}{%
2}}$\ are different.
\end{assu}

Assumption \ref{diff-eig} corresponds to Assumption G in Bai (2003). With
Assumption \ref{diff-eig}, the loadings and the factors are identifiable up
to a rotation. For identification of the loading space and the factor space,
Assumption \ref{diff-eig} is not needed.

\begin{assu}
\label{dist}(1) $\mathbb{E}(\left\Vert \frac{1}{\sqrt{NT}}%
\sum\nolimits_{k=1}^{N}\sum\nolimits_{t=1}^{T}\lambda _{i}^{0}(e_{it}e_{kt}-%
\mathbb{E(}e_{it}e_{kt}))1_{z_{t}=j}\right\Vert ^{2})\leq M$\ for all $%
i=1,...,N$ and $j=1,...,J^{0}$; and $\mathbb{E}(\left\Vert \frac{1}{\sqrt{NT}%
}\sum\nolimits_{i=1}^{N}\sum\nolimits_{t=1}^{T}(e_{it}e_{is}-\mathbb{E(}%
e_{it}e_{is}))f_{t}^{0}1_{z_{t}=j}\right\Vert ^{2})\leq M$\ for all $%
s=1,...,T$ and $j=1,...,J^{0}$.

(2) $\mathbb{E}(\left\Vert \frac{1}{\sqrt{NT}}\sum\nolimits_{k=1}^{N}\sum%
\nolimits_{t=1}^{T}\lambda _{k}^{0}f_{t}^{0\prime
}e_{kt}1_{z_{t}=j}\right\Vert ^{2})\leq M$\ for $j=1,...,J^{0}$.

(3) Define $\Phi _{ji}=plim\frac{1}{T}\sum\nolimits_{s=1}^{T}\sum%
\nolimits_{t=1}^{T}\mathbb{E}(f_{t}^{0}f_{s}^{0\prime
}e_{is}e_{it}1_{z_{s}=j}1_{z_{t}=j})$. For $j=1,...,J^{0}$, $\frac{1}{\sqrt{%
Tq_{j}^{0}}}\sum\nolimits_{t=1}^{T}f_{t}^{0}e_{it}1_{z_{t}=j}\overset{d}{%
\rightarrow }\mathcal{N}(0,\Phi _{ji})$.

(4) Define $\Gamma _{jt}=lim\frac{1}{N}\sum\nolimits_{i=1}^{N}\sum%
\nolimits_{k=1}^{N}\lambda _{ji}^{0}\lambda _{jk}^{0}\mathbb{E}%
(e_{it}e_{kt}) $. For $j=1,...,J^{0}$, $\frac{1}{\sqrt{N}}%
\sum\nolimits_{i=1}^{N}\lambda _{ji}^{0}e_{it}\overset{d}{\rightarrow }%
\mathcal{N}(0,\Gamma _{jt})$.
\end{assu}

Assumption \ref{dist} corresponds to Assumption F in Bai (2003). Part (3)
and part (4) are just central limit theorems and will be used for deriving
the limit distributions of the estimated factors and loadings.

\subsection{Asymptotic Results\label{AsyDyn}}

\begin{description}
\item[Consistency of the estimated loading space] 
\end{description}

\begin{thm}
\label{consis}Under Assumptions \ref{factors}, \ref{loadings}(1), \ref{error}
and \ref{state}, $\frac{1}{N}\left\Vert M_{\tilde{\Lambda}_{j}}\Lambda
_{j}^{0}\right\Vert _{F}^{2}=O_{p}(\frac{1}{\sqrt{\delta _{NT}}})$ for each $%
j$ as $(N,T)\rightarrow \infty $.
\end{thm}

Theorem \ref{consis} shows that the estimated loading space is consistent
without observing the state variable $z_{t}$. Note that the estimated
loadings $\tilde{\Lambda}_{j}$\ and the estimated regime probabilities $%
\tilde{p}_{tj\left\vert T\right. }$\ depend on each other, but the standard
technique in Bai (2003) for analyzing $\tilde{\Lambda}_{j}$\ is applicable
only when $\tilde{p}_{tj\left\vert T\right. }=1_{z_{t}=j}$. This is the
first technical difficulty we encounter in going from one regime to multiple
regimes. The crucial point for Theorem \ref{consis} is that if the linear
spaces spanned by $\tilde{\Lambda}_{j}$\ and $\Lambda _{j}^{0}$\ differ too
much, as long as $\min \phi _{j}>0$\ and $\min Q_{jk}>0$, the likelihood of $%
\tilde{\Lambda}$\ would be smaller than the likelihood of $\Lambda ^{0}$\
uniformly over all possible $\{z_{1},...,z_{T}\}$, i.e.,%
\begin{eqnarray*}
e^{l(\tilde{\Lambda},\tilde{\sigma}^{2})} &\leq
&\sup_{\{z_{1},...,z_{T}\}}\prod\nolimits_{t=1}^{T}L(x_{t}\left\vert z_{t};%
\tilde{\Lambda},\tilde{\sigma}^{2}\right. ) \\
&<&\sum\nolimits_{z_{T}=1}^{J^{0}}...\sum\nolimits_{z_{1}=1}^{J^{0}}\prod%
\nolimits_{t=1}^{T}L(x_{t}\left\vert z_{t};\Lambda ^{0},\tilde{\sigma}%
^{2}\right. )\Pr (z_{1},...,z_{T})=e^{l(\Lambda _{j}^{0},\tilde{\sigma}%
^{2})}.
\end{eqnarray*}%
This crucial point is due to large $N$, see the Appendix for the formal
proof. Based on Theorem \ref{consis}, we show that the estimated regime
probabilities are consistent.

\begin{description}
\item[Consistency of the estimated regime probabilities] 
\end{description}

\begin{thm}
\label{ptjT}Under Assumptions \ref{factors}-\ref{state} and \ref{moments}%
(1), as $(N,T)\rightarrow \infty $, for each $j$ and for any fixed $\eta >0$,

(1) $\sup_{t}\left\vert \tilde{p}_{tj\left\vert T\right.
}-1_{z_{t}=j}\right\vert =o_{p}(\frac{1}{N^{\eta }})$ if $T^{^{\frac{16}{%
\alpha }}}/N\rightarrow 0$ and $T^{\frac{2}{\alpha }+\frac{2}{\beta }%
}/N\rightarrow 0$,

(2) $\left\vert \tilde{p}_{tj\left\vert T\right. }-1_{z_{t}=j}\right\vert
=o_{p}(\frac{1}{N^{\eta }})$.
\end{thm}

Note that $\eta $ could be large but it is fixed as $(N,T)\rightarrow \infty 
$. $\alpha $ and $\beta $ could also be large as long as Assumptions \ref%
{factors}(2), \ref{error}(1) and \ref{moments}(1) are satisfied. Theorem \ref%
{ptjT} shows that $\tilde{p}_{tj\left\vert T\right. }$ is consistent as $%
N\rightarrow \infty $ and is uniformly consistent if $T$ is relatively small
compared to $N$. The proof utilizes the exponential likelihood ratio.

Theorem \ref{ptjT} implies that we can consistently identify which regime $%
x_{t}$ belongs to for all $t$, if there is common regime switching in the
loadings and the dimension of $x_{t}$ tends to infinity. Theorem \ref{ptjT}
also implies that we can consistently detect regime switching right after
the turning point with only one observation, so that we do not need to wait
for many observations of the time series from the new regime. This could
improve the speed of detection of new turning points, especially when high
frequency data is used.

An interesting special case is when the proposed algorithm is applied to
factor models with common breaks in the loadings. Various methods are
proposed recently for estimating the break points, Theorem \ref{ptjT}
implies that we can also consistently estimate the break points using the
proposed EM algorithm.

\begin{description}
\item[Convergence rate of the estimated loading space] 
\end{description}

If the true states $z_{t}$ were known, asymptotic properties of the
estimated loadings and factors are straightforward. Based on Theorem \ref%
{ptjT}, we shall show that using estimated regime probabilities does not
affect the asymptotic results. Define $W_{jNT}=\frac{1}{N}(\tilde{\Lambda}%
_{j}^{\prime }\tilde{\Lambda}_{j}+\tilde{\sigma}^{2}I_{r_{j}^{0}})(\frac{1}{T%
}\sum\nolimits_{t=1}^{T}\tilde{p}_{tj\left\vert T\right. })$ and $H_{j}=%
\frac{\sum\nolimits_{t=1}^{T}f_{t}^{0}f_{t}^{0\prime }1_{z_{t}=j}}{T}\frac{%
\Lambda _{j}^{0\prime }\tilde{\Lambda}_{j}}{N}W_{jNT}^{-1}$, then we have:

\begin{prop}
\label{VjHj}Let $V_{j}$ be a $r_{j}^{0}\times r_{j}^{0}$ diagonal matrix
consisting of eigenvalues of $\Sigma _{\Lambda _{j}}^{\frac{1}{2}}\Sigma
_{F_{j}}\Sigma _{\Lambda _{j}}^{\frac{1}{2}}$ in descending order and $%
\Upsilon _{j}$ be the corresponding eigenvectors. Under Assumptions \ref%
{factors}-\ref{diff-eig}, and assume $T^{^{\frac{16}{\alpha }}}/N\rightarrow
0$ and $T^{\frac{2}{\alpha }+\frac{2}{\beta }}/N\rightarrow 0$, as $%
(N,T)\rightarrow \infty $,

(1) $W_{jNT}\overset{p}{\rightarrow }q_{j}^{0}V_{j}$ for each $j$,

(2) $H_{j}\overset{p}{\rightarrow }\Sigma _{\Lambda _{j}}^{-\frac{1}{2}%
}\Upsilon _{j}V_{j}^{\frac{1}{2}}$ for each $j$.\footnote{$H_{j}$
corresponds to $(H^{-1})^{\prime }$ for the rotation matrix $H$ in Bai
(2003).}
\end{prop}

Proposition \ref{VjHj} is an important auxiliary result, and part (1) and
part (2) corresponds to Lemma A.3 and Proposition 1 in Bai (2003),
respectively. Lemma A.3 in Bai (2003) is based on the fact that the
estimated factors are $\sqrt{T}$\ times the eigenvectors corresponding to
the $r$\ largest eigenvalues\footnote{$r$ denotes the number of factors in
Bai (2003).} of $XX^{\prime }$\ and consequently $\tilde{\Lambda}^{\prime }%
\tilde{\Lambda}$\ is a diagonal matrix consisting of the $r$\ largest
eigenvalues of $\frac{1}{T}\sum\nolimits_{t=1}^{T}x_{t}x_{t}^{\prime }$.
However, here the first order condition (\ref{foc1})\ only tells us the
columns of $\tilde{\Lambda}_{j}$\ are the eigenvectors of $S_{j}$\ and $%
\tilde{\Lambda}_{j}^{\prime }\tilde{\Lambda}_{j}+\tilde{\sigma}%
^{2}I_{r_{j}^{0}}$\ are the corresponding eigenvalues. Condition (\ref{foc1}%
) does not tells us whether these eigenvalues are the $r_{j}^{0}$\ largest
eigenvalues of $S_{j}$\ or not.\ This is the second technical difficulty we
encounter in going from one regime to multiple regimes. Our proof strategy
of Proposition \ref{VjHj} utilizes Theorem \ref{consis} and is totally
different from Bai (2003)'s proof for his Proposition 1.

\begin{thm}
\label{rate}Under Assumptions \ref{factors}-\ref{diff-eig}, and assume $T^{^{%
\frac{16}{\alpha }}}/N\rightarrow 0$ and $T^{\frac{2}{\alpha }+\frac{2}{%
\beta }}/N\rightarrow 0$, as $(N,T)\rightarrow \infty $, $\frac{1}{N}%
\left\Vert \tilde{\Lambda}_{j}-\Lambda _{j}^{0}H_{j}\right\Vert
_{F}^{2}=O_{p}(\frac{1}{\delta _{NT}^{2}})$ for each $j$.
\end{thm}

Theorem \ref{rate} establishes the convergence rate of the estimated loading
space for each regime. This could help us study the effect of using
estimated loadings on subsequent applications. For example, if the estimated
loadings are used to construct portfolios, Theorem \ref{rate} could help us
calculate how the estimation error contained in $\tilde{\Lambda}_{j}$ would
affect the performance of these portfolios.

\begin{description}
\item[Limit distributions of the estimated loadings] 
\end{description}

\begin{thm}
\label{ld}Under Assumptions \ref{factors}-\ref{dist}, and assume $\sqrt{T}%
/N\rightarrow 0$, $T^{^{\frac{16}{\alpha }}}/N\rightarrow 0$ and $T^{\frac{2%
}{\alpha }+\frac{2}{\beta }}/N\rightarrow 0$, as $(N,T)\rightarrow \infty $, 
$\sqrt{Tq_{j}^{0}}(\tilde{\lambda}_{ji}-H_{j}^{\prime }\lambda _{ji}^{0})%
\overset{d}{\rightarrow }\mathcal{N}(0,V_{j}^{-\frac{1}{2}}\Upsilon
_{j}^{\prime }\Sigma _{\Lambda _{j}}^{\frac{1}{2}}\Phi _{ji}\Sigma _{\Lambda
_{j}}^{\frac{1}{2}}\Upsilon _{j}V_{j}^{-\frac{1}{2}})$ for each $j$.
\end{thm}

Theorem \ref{ld} shows that for each $j$ and $i$, $\tilde{\lambda}_{ji}$ has
a limiting normal distribution. This allows us to construct confidence
interval for the estimated loadings. Also note that the rotation matrix $%
H_{j}$ is different for different regime.

\begin{remark}
We can also prove the consistency and limit distribution of $\tilde{\sigma}%
^{2}$ (the probability limit of $\tilde{\sigma}^{2}$ is $\lim\limits_{N%
\rightarrow \infty }\frac{1}{N}\sum\nolimits_{i=1}^{N}\sigma _{i}^{2}$), we
omit it since this is not our focus.
\end{remark}

\begin{description}
\item[Asymptotic properties of the estimated factors] 
\end{description}

\begin{thm}
\label{factor}Under Assumptions \ref{factors}-\ref{dist}, and assume $\sqrt{N%
}/T\rightarrow 0$, $T^{^{\frac{16}{\alpha }}}/N\rightarrow 0$\ and $T^{\frac{%
2}{\alpha }+\frac{2}{\beta }}/N\rightarrow 0$, as $(N,T)\rightarrow \infty $,

(1) $\frac{1}{T}\sum\nolimits_{t=1}^{T}\left\Vert \tilde{f}%
_{t}-[(H_{z_{t}}^{-1}f_{t}^{0})^{\prime },0_{\max
r_{j}^{0}-r_{z_{t}}^{0}}^{\prime }]^{\prime }\right\Vert ^{2}=O_{p}(\frac{1}{%
\delta _{NT}^{2}})$,

(2) $\sqrt{N}(\tilde{f}_{t}-[(H_{z_{t}}^{-1}f_{t}^{0})^{\prime },0_{\max
r_{j}^{0}-r_{z_{t}}^{0}}^{\prime }]^{\prime })$

$\overset{d}{\rightarrow }\mathcal{N}(0,\left[ 
\begin{array}{cc}
V_{z_{t}}^{-\frac{1}{2}}\Upsilon _{z_{t}}^{\prime }\Sigma _{\Lambda
_{z_{t}}}^{-\frac{1}{2}}\Gamma _{z_{t}t}\Sigma _{\Lambda _{z_{t}}}^{-\frac{1%
}{2}}\Upsilon _{z_{t}}V_{z_{t}}^{-\frac{1}{2}} & 0_{r_{z_{t}}^{0}\times
(\max r_{j}^{0}-r_{z_{t}}^{0})} \\ 
0_{(\max r_{j}^{0}-r_{z_{t}}^{0})\times r_{z_{t}}^{0}} & 0_{(\max
r_{j}^{0}-r_{z_{t}}^{0})\times (\max r_{j}^{0}-r_{z_{t}}^{0})}%
\end{array}%
\right] )$.
\end{thm}

Theorem \ref{factor}(2) shows that the limit distribution of $\tilde{f}_{t}$
is mixed normal, since the rotation matrix $H_{z_{t}}^{-1}$ and the
asymptotic variance depend on the state variable $z_{t}$. Theorem \ref%
{factor}(1) establishes the convergence rate of the estimated factor space.
Note that if $\{\tilde{f}_{t},t=1,...T\}$ is used as proxies for the true
factors in factor-augmented forecasting (or factor-augmented VAR), the
forecasting equation (or the VAR equation) would have induced regime
switching in the model parameters, because $H_{z_{t}}^{-1}$ depends on $%
z_{t} $. For illustration, consider the following $h$-period ahead
forecasting model using factors and some other observable variables $W_{t}$: 
$y_{t+h}=a^{\prime }f_{t}^{0}+b^{\prime }W_{t}+u_{t+h}$. If $\tilde{f}_{t}$
is used as proxies for $f_{t}^{0}$, the model can be written as%
\begin{equation*}
y_{t+h}=-a^{\prime }H_{z_{t}}(\tilde{f}_{t}-H_{z_{t}}^{-1}f_{t}^{0})+a^{%
\prime }H_{z_{t}}\tilde{f}_{t}+b^{\prime }W_{t}+u_{t+h}.
\end{equation*}%
The first term on the right hand side is asymptotically negligible. It is
easy to see that the coefficient $a^{\prime }H_{z_{t}}$ depends on $z_{t}$
and this need to be taken into account when we estimate the forecasting
equation. Finally, we show that the estimated transition probability matrix
is also consistent when $\{z_{1},...,z_{T}\}$\ is a Markov process.

\begin{thm}
\label{Q}Assume that $\{z_{1},...,z_{T}\}$ is a Markov process, under
Assumptions \ref{factors}-\ref{state} and \ref{moments}(1), $\tilde{Q}_{jk}%
\overset{p}{\rightarrow }Q_{jk}^{0}$ for each $j$ and $k$ as $%
(N,T)\rightarrow \infty $ if $T^{^{\frac{16}{\alpha }}}/N\rightarrow 0$\ and 
$T^{\frac{2}{\alpha }+\frac{2}{\beta }}/N\rightarrow 0$.
\end{thm}

\section{Determine the Number of Factors and the Number of Regimes\label{rJ}}

Given the number of factors $(r_{1},...,r_{J})$ and the number of regimes $J$%
, let $(\tilde{\Lambda}_{1,r_{1}},...,\tilde{\Lambda}_{J,r_{J}})$ be the
solution for maximizing the log-likelihood $l(\Lambda _{1,r_{1}},...,\Lambda
_{J,r_{J}},\sigma ^{2},Q,\phi )$. Here we use $\Lambda _{j,r_{j}}$ to
emphasize that $\Lambda _{j,r_{j}}$ is of dimension $N\times r_{j}$. The
criterion we propose for model selection is:%
\begin{equation}
PC(r_{1},...,r_{J})=\frac{1}{NT}l(\tilde{\Lambda}_{1,r_{1}},...,\tilde{%
\Lambda}_{J,r_{J}},\sigma ^{2},Q,\phi
)-\sum\nolimits_{j=1}^{J}(g(N,T))^{b(r_{j})},  \label{bt}
\end{equation}%
where $g(N,T)$ is a penalty function depending on both $N$ and $T$, and $%
b(\cdot )$ is a positive and decreasing function with $b(1)=1$, e.g., $%
b(r_{j})=\frac{1}{r_{j}}$. For each $J$, the numbers of factors are
estimated by%
\begin{equation}
(\tilde{r}_{1},...,\tilde{r}_{J})=\arg \max\nolimits_{r_{j}\leq \bar{r}%
,j=1,...,J}PC(r_{1},...,r_{J})\text{,}  \label{bu}
\end{equation}%
and then the number of regimes is estimated by%
\begin{equation}
\tilde{J}=\arg \max\nolimits_{J\leq \bar{J}}PC(\tilde{r}_{1},...,\tilde{r}%
_{J})\text{,}  \label{bv}
\end{equation}%
where $\bar{r}$ is the maximal number of factors in each regime and $\bar{J}$
is the maximal number of regimes. In the following theorem we show that $(%
\tilde{r}_{1},...,\tilde{r}_{J})$ and $\tilde{J}$ are consistent.

\begin{thm}
\label{r0&J0}Under Assumptions \ref{factors}, \ref{loadings}(1), \ref{error}%
, \ref{state} and assume $\lim\limits_{N\rightarrow \infty }\frac{1}{N}%
\Lambda _{k}^{0\prime }M_{\Lambda _{j}^{0}}\Lambda _{k}^{0}\neq 0$ for any $%
j\ $and $k$, we have $\Pr (\tilde{J}=J^{0}$ and $\tilde{r}_{j}=r_{j}^{0}$
for all $j)\rightarrow 1$ as $(N,T)\rightarrow \infty $ if (i) $%
g(N,T)\rightarrow 0$, (ii) $\delta _{NT}g(N,T)\rightarrow \infty $, and
(iii) $b(\cdot )$ is a positive and decreasing function with $b(1)=1$.
\end{thm}

Note that the condition $\lim\limits_{N\rightarrow \infty }\frac{1}{N}%
\Lambda _{k}^{0\prime }M_{\Lambda _{j}^{0}}\Lambda _{k}^{0}\neq 0$ allows $%
\Lambda _{k}^{0}$ and $\Lambda _{j}^{0}$ to share some columns, i.e.,
Theorem \ref{r0&J0} holds for the case where the loadings of some (but not
all) factors remain the same across different regimes.

The basic idea behind Theorem \ref{r0&J0} is similar to Theorem 2 of Bai and
Ng (2002), i.e., add a penalty term that converges to zero but slowly enough
so that underparameterized models and overparameterized models will not be
chosen. Here the penalty $(g(N,T))^{b(r_{j})}$ converges to zero because $%
g(N,T)\rightarrow 0$ and $b(r_{j})$ is positive, and $\delta
_{NT}(g(N,T))^{b(r_{j})}\rightarrow \infty $ because $\delta
_{NT}g(N,T)\rightarrow \infty $ and $b(r_{j})$ is a decreasing function of $%
r_{j}$ with $b(1)=1$. Compared to Bai and Ng (2002), the difference and
difficulty here is that the number of regimes is unknown and the number of
factors in each regime may be different. For example, suppose the true model
is $(r_{1}=2,r_{2}=1,J=2)$ and the two columns in $\Lambda _{1}^{0}$ are
linearly independent with $\Lambda _{2}^{0}$. This model can be equivalently
written as%
\begin{equation*}
x_{t}=(\Lambda _{1}^{0},\Lambda _{2}^{0})\left( 
\begin{array}{c}
f_{1t}^{0} \\ 
f_{2t}^{0} \\ 
0%
\end{array}%
\right) +e_{t}\text{ if }z_{t}=1\text{, and }x_{t}=(\Lambda _{1}^{0},\Lambda
_{2}^{0})\left( 
\begin{array}{c}
0 \\ 
0 \\ 
f_{1t}^{0}%
\end{array}%
\right) +e_{t}\text{ if }z_{t}=2,
\end{equation*}%
i.e., there is only one regime and there are three factors in this regime.
The difference between the log-likelihood of the true model $%
(r_{1}=2,r_{2}=1,J=2)$ and the log-likelihood of the equivalent model $%
(r_{1}=3,J=1)$ is negligible and clearly Bai and Ng (2002) is not applicable
to this example.

Our solution is to add a penalty term for each regime and let the penalty
term of different regime have different asymptotic order, so that
overestimating the number of factors in one regime can not be compensated by
underestimating the number of factors in another regime. For example, the
penalty for the equivalent model is $(g(N,T))^{b(3)}$ while the penalty for
the true model is $(g(N,T))^{b(2)}+(g(N,T))^{b(1)}$. Since $\frac{%
(g(N,T))^{b(3)}}{(g(N,T))^{b(2)}+(g(N,T))^{b(1)}}\rightarrow \infty $ as $%
(N,T)\rightarrow \infty $, the true model would be chosen with probability
approaching one as $(N,T)\rightarrow \infty $. The formal proof of Theorem %
\ref{r0&J0} is provided in the Appendix.

Our method can also be used to consistently determine the number of factors
and the number of breaks for factor models with multiple common breaks in
the loadings. If we replace $l(\tilde{\Lambda}_{1,r_{1}},...,\tilde{\Lambda}%
_{J,r_{J}},\sigma ^{2},Q,\phi )$ in expression (\ref{bt}) by minus the
minimum of the least squares over all possible break points and calculate $(%
\tilde{r}_{1},...,\tilde{r}_{J},\tilde{J})$ as expressions (\ref{bu})-(\ref%
{bv}), then it is not difficult to prove that we still have $\Pr (\tilde{J}%
=J^{0}$ and $\tilde{r}_{j}=r_{j}^{0}$ for all $j)\rightarrow 1$ as $%
(N,T)\rightarrow \infty $. As we discussed in the Introduction, recently the
literature on the factor loading instability issues developed quite a lot,
but as far as we know, there are very few (if any) consistent model
selection procedures that allow $r_{j}$ to be different across $j$ and allow 
$\Lambda _{k}^{0}$ and $\Lambda _{j}^{0}$ to share some columns.

\section{Simulations\label{Simu}}

In this section, we perform simulations to confirm the theoretical results
and examine the finite sample performance of our methods under various
empirically relevant scenarios.

\subsection{Simulation Design\label{simu desi}}

The data is generated as follows:%
\begin{equation*}
x_{it}=\left\{ 
\begin{array}{l}
f_{t}^{0\prime }\lambda _{1i}^{0}+e_{it}\text{ if }z_{t}=1, \\ 
f_{t}^{0\prime }\lambda _{2i}^{0}+e_{it}\text{ if }z_{t}=2,%
\end{array}%
\right. \text{for }i=1,...,N\text{ and }t=1,...,T,
\end{equation*}%
i.e., we consider two regimes. For the factors and the loadings, we consider
four data generating processes (DGP) as listed below:

\textit{DGP 1: There are two factors in both regimes and the loadings of
both factors have regime switching.}

\textit{DGP 2: There are two factors in both regimes and only the loadings
of the second factor have regime switching.}

\textit{DGP 3: There is one factor in both regimes and its loadings have
regime switching.}

For DGP1 - DGP3, the factors are generated as follows: 
\begin{equation*}
f_{t,p}^{0}=\rho f_{t-1,p}^{0}+\epsilon _{t,p}\text{ for }t=2,...,T\text{
and }p=1,...,r^{0}.
\end{equation*}%
$\epsilon _{t,p}$\ is i.i.d. $N(0,1)$, and $f_{1,p}^{0}$\ is i.i.d. $N(0,%
\frac{1}{1-\rho ^{2}})$ so that the distributions of the factors are
stationary. Serial correlation of the factors is controlled by the scalar $%
\rho $.

\textit{DGP 4: The loadings are generated in the same way as DGP2. }$%
f_{t,1}^{0}$\textit{\ is generated as i.i.d. }$N(0,1)$\textit{\ and }$%
f_{t,2}^{0}$\textit{\ is generated as uniform }$(0.5,1.5)$\textit{. }

The errors are generated as follows:%
\begin{equation*}
e_{it}=\zeta e_{i,t-1}+v_{it}\text{ for }i=1,...,N\text{ and }t=2,...,T\text{%
,}
\end{equation*}%
where $v_{t}=(v_{1,t},...,v_{N,t})^{\prime }$\ is i.i.d. $N(0,\Omega )$ for $%
t=2,...,T$ and $(e_{1,1},...,e_{N,1})^{\prime }$\ is $N(0,\frac{1}{1-\zeta
^{2}}\Omega )$\ so that the distributions of the errors are stationary.
Serial correlation of the errors is controlled by the scalar $\zeta $. For $%
\Omega $, we set $\Omega _{ij}=\xi ^{\left\vert i-j\right\vert }$ for some $%
\xi $ between 0 and 1, thus cross-sectional dependence of the errors is
controlled by $\xi $. In addition, the processes $\{\epsilon _{t,p}\}$\ and $%
\{v_{it}\}$\ are mutually independent for all $p$ and $i$.

The loadings are generated as follows: For DGP1, both $\lambda _{1i}^{0}$
and $\lambda _{2i}^{0}$\ are generated as i.i.d. $N(0,\frac{1-\rho ^{2}}{%
1-\zeta ^{2}}\frac{2R^{2}}{1-R^{2}}I_{2})$ across $i$, and $\lambda
_{1i}^{0} $ and $\lambda _{2i}^{0}$ are also independent with each other.
For DGP2, $\lambda _{1i}^{0}$ and the second element of $\lambda _{2i}^{0}$
are generated as i.i.d. $N(0,\frac{1-\rho ^{2}}{1-\zeta ^{2}}\frac{2R^{2}}{%
1-R^{2}}I_{3})$ across $i$. For DGP3, both $\lambda _{1i}^{0}$ and $\lambda
_{2i}^{0}$\ are generated as i.i.d. $N(0,\frac{1-\rho ^{2}}{1-\zeta ^{2}}%
\frac{R^{2}}{1-R^{2}})$ across $i$, and $\lambda _{1i}^{0}$ and $\lambda
_{2i}^{0}$ are also independent with each other. All loadings are
independent of the factors and the errors. The variance $\frac{1-\rho ^{2}}{%
1-\zeta ^{2}}\frac{2R^{2}}{1-R^{2}}$ guarantees that the regression
R-square\ of each series $i$\ is equal to $R^{2}$, this controls the
signal-noise ratio. Following the literature, we set $R^{2}=0.5$.

For the state process $\{z_{t},t=1,...,T\}$, we consider four cases as
listed below:

\textit{Regime Pattern 1: US business cycle 1945Q2-2020Q1}

\textit{Regime Pattern 2: single common break at }$t=T/2$

\textit{Regime Pattern 3: two common breaks at }$t=T/3$ and $t=2T/3$, 
\textit{and the loadings switch back after the second break}

\textit{Regime Pattern 4: a randomly generated Markov process }

Regime pattern 1 is based on the US business cycle from 1945 Quarter 2 to
2020 Quarter 1, as determined by the NBER business cycle dating committee.
There are 75 years (300 quarters) in total, thus we have $T=300$. For $%
t=1,...,300$, $z_{t}=1$ if the US economy at time $t$ is in expansion and $%
z_{t}=2$ if the US economy at time $t$ is in recession. The transition
probabilities of the state process calibrated to the US business cycle is $%
Q_{11}^{0}=0.95$ and $Q_{22}^{0}=0.72$ (average duration of expansion is $%
1/(1-Q_{11}^{0})=20$ and average duration of recession is $%
1/(1-Q_{22}^{0})\approx 3.5$).

Regime patterns 2 and 3 correspond to the case where loadings have single
common break\ and multiple common breaks, respectively. Regime patterns 3 is
especially interesting since the case where there are multiple breaks and
the loadings switch back to their original values after the second break is
rarely studied in the literature. Various methods are proposed in the
literature recently for estimating the break points, here we perform
simulations for regime patterns 2 and 3 to evaluate the finite sample
performance of our method when it is applied to these interesting cases.

Regime pattern 4 is a Markov process randomly generated with transition
probabilities $Q_{11}^{0}=0.95$ and $Q_{22}^{0}=0.72$, and $%
\{z_{t},t=1,...,T\}$ is independent with $f_{s}^{0}$ and $e_{is}$ for all $i$
and $s$. Regime patterns 1-3 are prespecified and are not necessarily Markov
processes, thus here we consider regime pattern 4 to evaluate the
performance of our method when applied to a Markov state process.

We study both the unsmoothed algorithm and the smoothed algorithm.\textbf{\ }%
The key difference is that in the E-step, the former uses unsmoothed regime
probabilities while the latter uses smoothed regime probabilities. Both
algorithms start from randomly generated initial values of the loadings and
iterate between the E-step and the M-step until convergence. To search for
the global maximum of the likelihood function, we generate initial values
randomly for many times and take the one with the largest likelihood. For
other parameters, we set $\sigma ^{2}=1$, $q_{j}=0.5$\ for $j=1,2$, $\phi
_{k}=0.5$ for $k=1,2$, $Q_{11}=0.95$ and $Q_{22}=0.72$. $Q_{11}$ and $Q_{22}$
are calibrated to regime pattern 1. Once we get the estimated regime
probabilities and the estimated loadings, $\tilde{Q}_{11}$ and $\tilde{Q}%
_{22}$ are estimated by equation (\ref{Qtilde-jk}), and the factors are
estimated by equation (\ref{ftilde-t}).

\subsection{Simulation Results\label{simu resu}}

Figure 1 displays the smoothed probabilities of regime 2 for DGP 1 with $%
(N,T)=(100,300)$ and $(\rho ,\zeta ,\xi )=(0,0,0)$. Subfigures 1-4 of Figure
1 correspond to regime patterns 1-4, respectively. It is easy to see that in
all subfigures when the true regime is regime 1, the smoothed probabilities
stay at zero with only a few short and mild spikes. At the beginning of each
shaded region, the smoothed probabilities increase to one instantly, and at
the end of each shaded region, the smoothed probabilities instantly decrease
to zero. Figure 2 displays the unsmoothed probabilities of regime 2 for DGP1
under the four regime patterns with $(N,T)=(100,300)$ and $(\rho ,\zeta ,\xi
)=(0,0,0)$. The estimated probabilities still stay at zero when it's regime
1 and instantly increase to one (decrease to zero) when there is regime
switching, but compared to Figure 1, Figure 2 shows more and sharper spikes
(upward or downward). These spikes are false positives in detecting regime
switching. Figure 3 and Figure 4 display the smoothed and the unsmoothed
probabilities of regime 2 for DGP2, respectively. The performance of the
estimated probabilities deteriorates since for DGP2 only one factor has
regime switching in its loadings. Overall, Figures 1-4 confirm the
theoretical results that turning points (break points) can be identified
consistently if $N$ is large.

Comparing Figure 2 to Figure 1 and Figure 4 to Figure 3, it is obvious that
the smoothed probabilities performs much better than the unsmoothed
probabilities. Many false positives in Figure 3 and Figure 4 are eliminated
by the smoother. This is because for each $t$, regimes at $t-1$ and $t+1$
contains information for detecting the regime at period $t$. Comparing
subfigures 2-3 to subfigures 1 and 4 in Figures 1-4, we can see that the
performance of the estimated probabilities under regime patterns 2-3 is
better than the performance under patterns 1 and 4. This is also because
regimes at the neighborhood periods provide information for the current
regime. Roughly speaking, the performance is better when the regime pattern
is relatively simple. In addition, we can also see that the performance
under regime pattern 1 is slightly better than the performance under regime
pattern 4. This is because the subsample size of regime 2 under pattern 4 is
larger than the subsample size under pattern 1 (72 vs 45). In general, we
find that to guarantee good performance, the subsample size for each regime
should be not less than 40.

Figure 5 focuses on regime pattern 1 and displays the estimated
probabilities of regime 2 for DGP1 and DGP2 with $N=200$. Comparing the
subfigure 3 and subfigure 4 of Figure 5 to subfigure 1 of Figure 3 and
subfigure 1 of Figure 4, it is easy to see that $N=200$ improves the
performance of the estimated probabilities. Figure 6 also focuses on regime
pattern 1 and displays the smoothed probabilities of regime 2 for DGP1 and
DGP2 with $(\rho ,\zeta ,\xi )=(0.5,0,0)$ or $(0,0.5,0.5)$. Comparing to
subfigure 1 of Figure 1 and subfigure 1 of Figure 3, it seems that the value
of $(\rho ,\zeta ,\xi )$ does not affect the performance too much if they
were far away from 1.

Figure 7 displays the smoothed and unsmoothed probabilities for regime
pattern 4 (a randomly generated Markov process) under DGP1 with $%
(N,T)=(50,500)$\ or $(N,T)=(500,50)$. The smoothed probabilities still
perform well even under such extreme case, especially when $T=50$, the
subsample size of regime 2 is just 13. However, the unsmoothed probabilities
in subfigures 3-4 deteriorate obviously, compared to subfigure 4 of Figure 2.

Figure 8 displays the smoothed and unsmoothed probabilities under DGP4 for
regime patterns 1 and 2 with $(N,T)=(100,300)$. DGP4 modifies DGP2 so that
the second factor stays away from zero. Comparing subfigures 1-2 to
subfigures 1-2 of Figure 3 and subfigures 3-4 to subfigures 1-2 of Figure 4,
we can see that the performance improvement is quite significant. This is
because under DGP2, the second factor $f_{t,2}^{0}$ is likely to be close to
zero and it is difficult to identify the regime of $x_{t}$\ when $%
f_{t,2}^{0} $\ is close to zero. Thus Figures 3-4 reflect more of the
identification problem when the factors equal zeros.

Finally, to access the adequacy of the asymptotic distributions of the
estimated loadings and factors in approximating their finite sample
counterparts, we display in Figures 9-12 the histograms of the standardized
estimated factors for $t=T/2$ and the standardized estimated loadings for $%
i=N/2$ under DGP3. The number of simulations is 1000. The histograms are
normalized to be a density function and the standard normal density curve is
overlaid on them for comparison. It is easy to see that in all subfigures of
Figures 9-12, the standard normal density curve provides good approximation
to the normalized histograms. The histograms of the estimated factors in
Figure 9 are slightly fat-tailed because of bad initial values. Comparing
the four rows in each of Figures 9-12, we can see that the estimated
loadings and factors using the smoothed algorithm perform better than using
the unsmoothed algorithm, $(\rho ,\zeta ,\xi )=(0.5,0.5,0.5)$ does not
matter too much, and $N=200$ significantly improves the performance.

The number of initial value trials also significantly affect the
performance. We find that for regime pattern 1, normally 5 trials are
enough, but to guarantee good performance in all of 1000 replications, 30
trials are needed. For regime pattern 4, normally 2 trials are enough and 15
trials are needed to guarantee good performance in all replications. For
regime patterns 2-3, 5 trials are enough for all replications. In general,
more trials are needed when the regime pattern is complex and the subsample
size is small.

In addition, we also present in Table \ref{Table} the average $R^{2}$ of the
estimated loadings of regime 1 and regime 2 projecting on the true loadings,
the average $R^{2}$ of the estimated factors, and the average absolute error
of the estimated transition probabilities. It is easy to see that in Table
1, $R_{l1}^{2}$ and $R_{l2}^{2}$ are always close to one. $R_{Hf}^{2}$ is
always close to one but $R_{f}^{2}$ is much smaller than $R_{Hf}^{2}$. This
is because $R_{Hf}^{2}$ considers the regime specific rotation matrix, as
shown in Theorem \ref{factor}(1). In summary, results in Figures 1-12 and
Table 1 lend strong support to the theoretical results and illustrate the
usefulness of the proposed EM algorithms.

\section{Empirical Application\label{App}}

In this section we apply the proposed method to detect turning points of US
business cycle from 02/1980 to 01/2023 in real-time using the FRED-MD
(Federal Reserve Economic Data - Monthly Data) data set. The FRED database
is maintained by the Research division of the Federal Reserve Bank of St.
Louis, and is publicly accessible and updated in real-time. The 02/2023
vintage of the FRED-MD data set contains 128 unbalanced monthly time series
from 01/1959 to 01/2023, including eight groups (output and income, labor
market, housing, consumption and inventories, money and credit, prices,
stock market). After removing those series with missing values and data
transformation\footnote{%
See the Appendix of McCracken and Ng (2016) for the details of data
description and transformation.}, we have 106 balanced monthly series
ranging from 03/1959 to 01/2023. Finally, the data is demeaned and
standardized.

For each month from 02/1980 to 01/2023 (516 months in total), we use the
data from 03/1959 to that month for calculating the probability of recession
of that month, i.e., we behave as if we were standing at that month\footnote{%
For simplicity, we do not use the vintage data of that month. Compared to
the vintage data, the data we use contains revision in some series if more
accurate observations were available after that month, but previous studies
on business cycle dating show that data revisions have little effects on the
results.}. More specifically, we apply the EM algorithm in Section \ref{EM}
to the data from 03/1959 to the previous month to estimate the model
parameters\footnote{%
US business cycle from 03/1959 to the previous month as determined by NBER
is used as the initial values for probabilities.}, and then use the
estimated parameters and the data from 03/1959 to that month to calculate
the filtered probability of recession for that month. Since the data of that
month is available at the end of that month or the beginning of the next
month, new recession or expansion starting from the beginning of that month
could only be detected with at least one month delay. \ 

To convert the recession probability of each month into a binary variable
that indicates the state of the economy in that month, we compare the
estimated recession probability to a prespecified threshold. More
specifically, if the previous turning point is a trough and the recession
probability of month $t$ exceeds 0.8 for the first time after the previous
turning point, month $t$ would be considered as a new turning point from
expansion to recession. Similarly, if the previous turning point is a peak
and the recession probability of month $t$ falls below 0.2 for the first
time after the previous turning point, month $t$ would be considered as a
new turning point from recession to expansion. For robustness check, we also
consider $(0.9,0.1)$ as the threshold, the results are quite similar.

We consider the turning points determined by the NBER BCDC (business cycle
dating committee) as the benchmark for comparison and we mainly focus on the
accuracy and speed of the proposed method in detecting turning points. The
proposed method is applied to both the whole panel and a subset of the whole
panel which consists of only the first 50 series among all 106 series. The
results of using only the first 50 series are better. We conjecture that
this is mainly because not all 106 series had regime switching in the factor
loadings at each turning point determined by the NBER BCDC\footnote{%
The NBER BCDC mainly focuses on four series, (1) non-farm payroll
employment, (2) industrial production, (3) real manufacturing and trade
sales, and (4) real personal income excluding transfer payments.}, or some
series had regime switching in their loadings at time periods that are
different from the NBER BCDC turning points. Thus we may further improve the
performance of the proposed method by selecting series that are most
relevant to and synchronous with the business cycle. A careful selection is
out of the scope of this paper.

Table \ref{Table 2} presents the real-time results of 02/1980-02/2020 using
the first 50 series. The number of factors in each regime is set to be six.
mm/yyyy in the second and the seventh row indicate the starting month of
each recession and expansion. The row corresponds to "NBER BCDC", "Chauvet
Piger" and "This paper" shows the number of months it takes the NBER BCDC,
Chauvet and Piger (2008) and this paper to detect each recession and
expansion, respectively. For example, the recession starting from the
beginning of February 1980 would be detected by the NBER BCDC at the
beginning of June 1980, by Chauvet and Piger (2008) at the beginning of
August 1980, and by this paper at the beginning of May 1980, respectively.
Overall, it is easy to see that this paper detects turning points much
faster than NBER BCDC and slightly faster than Chauvet and Piger (2008). On
average, this paper detects recessions with 6.25 months delay and expansions
with 5.4 months delay, NBER BCDC detects recessions with 7.4 months delay
and expansions with 14.8 months delay, and Chauvet and Piger (2008) detects
recessions with 8.6 months delay and expansions with 6.2 months delay.

We also detect two recessions after the Covid-19 pandemic, one from 03/2020
to 08/2020 and the second from 02/2021 to 05/2021, so we would have detected
the 03/2020-08/2020 recession in 04/2020 because the data for 03/2020 is
available with one month delay. This is quite interesting, given that our
method detects recessions with 6.25 months delay on average during the
period 02/1980-02/2020.

Table \ref{Table 2} shows that using more series could improve the speed of
turning points detection. However, using more series could also bring in
false positives (turning points detected by the proposed method using many
series but not detected by NBER BCDC), because the extra series may not be
synchronous with the NBER business cycle. Here we detect eight false
recessions: 09/1983-11/1983, 10/1986-02/1987, 07/1989-10/1989,
01/1993-02/1993, 01/1995-03/1995, 08/1998, 05/2000-08/2000, 06/2010-10/2010,
and one false expansion: 02/1982. While these false positives should not be
ignored, most of them only last for a very short periods and would have
little effect on macroeconomic policy. Overall, our results illustrate the
potential of using a large number of series and factor models with common
loading switching for real-time detection of the business cycle turning
points.

To get rid of those false recessions, one possible solution is to select
time series that are synchronous with the NBER business cycle, and another
promising solution is to extend our results to the case where regime
switching in the loadings is approximately synchronous rather than exactly
synchronous. In fact, Stock and Watson (2014) mainly focuses on how to
combine different turning points of many individual series (determined by
the Bry-Boschan algorithm) into a single point. \ 

\section{Conclusions\label{Con}}

The exposure of economic time series to common factors may switch depending
on state variables such as fiscal policy, monetary policy, business cycle
stage, stock market volatility, technology and so on. For consistent
estimation of the factor structure, it is crucial to take into account such
regime switching phenomena. This paper considers maximum likelihood
estimation for large factor models with common regime switching in the
loadings and proposes EM algorithm for computation, which is easy to
implement and runs fast even when $N$ is large. Convergence rates and limit
distributions of the estimated loadings and the estimated factors are
established under the approximate factor model setup. This paper also shows
that when $N$ is large, regime switching can be identified consistently and
only one observation after the switching point is needed. This allows us to
detect regime switching at very early times. Monte Carlo simulations confirm
the theoretical results and good performance of our method. An application
to the FRED-MD dataset demonstrates the potential of using many time series
with our method for detection of the business cycle turning points.

Some related topics are worth further study. First, it would be interesting
to see the performance of the portfolio constructed using regime specific
loadings, and how the identified regime is related to exogenous variables
such as market volatility and money growth. Second, our results imply that
the forecasting equation would have induced regime switching if the
estimated factors are used for forecasting, so we want to know whether it
indeed matters. Finally, a selection of time series that are most
synchronous with or related to business cycle could improve the speed and
accuracy of our method for turning points detection, so we would like to see
how much we can achieve after careful selection.

\begin{description}
\item[Acknowledgements] 
\end{description}

We are deeply indebted to the Editor, Serena Ng, the Associate Editor, and
two referees for very useful comments and suggestions which have helped to
improve and develop the paper further. We wish to thank Daniele Massacci for
some useful discussion on a previous version of the paper. The usual
disclaimer applies. We acknowledge financial support from the Centre for
Econometric Analysis, Bayes Business School, London (UK), and the
Fundamental Research Funds for the Central Universities, Peking University,
Beijing (China).

\begin{figure}[tbp]%
\begin{center}%
\caption{Smoothed Probabilities of Regime 2 for DGP 1}%
\begin{tabular}{l}
\FRAME{itbpF}{6.0226in}{1.6561in}{0in}{}{}{1-a.png}{\special{language
"Scientific Word";type "GRAPHIC";maintain-aspect-ratio TRUE;display
"USEDEF";valid_file "F";width 6.0226in;height 1.6561in;depth
0in;original-width 42.0558in;original-height 11.6663in;cropleft "0";croptop
"1";cropright "1.0212";cropbottom "0";filename '1-a.png';file-properties
"XNPEU";}} \\ 
\FRAME{itbpF}{6.0226in}{1.6561in}{0in}{}{}{1-b.png}{\special{language
"Scientific Word";type "GRAPHIC";maintain-aspect-ratio TRUE;display
"USEDEF";valid_file "F";width 6.0226in;height 1.6561in;depth
0in;original-width 41.9719in;original-height 11.6663in;cropleft "0";croptop
"1";cropright "1.0233";cropbottom "0";filename '1-b.png';file-properties
"XNPEU";}} \\ 
\FRAME{itbpF}{6.0226in}{1.6561in}{0in}{}{}{1-c.png}{\special{language
"Scientific Word";type "GRAPHIC";maintain-aspect-ratio TRUE;display
"USEDEF";valid_file "F";width 6.0226in;height 1.6561in;depth
0in;original-width 41.9996in;original-height 11.6663in;cropleft "0";croptop
"1";cropright "1.0225";cropbottom "0";filename '1-c.png';file-properties
"XNPEU";}} \\ 
\FRAME{itbpF}{6.0226in}{1.6561in}{0in}{}{}{1-d.png}{\special{language
"Scientific Word";type "GRAPHIC";maintain-aspect-ratio TRUE;display
"USEDEF";valid_file "F";width 6.0226in;height 1.6561in;depth
0in;original-width 42.0282in;original-height 11.6663in;cropleft "0";croptop
"1";cropright "1.0219";cropbottom "0";filename '1-d.png';file-properties
"XNPEU";}}%
\end{tabular}%
\label{smoothed 2-2}%
\end{center}%
{\footnotesize Notes: Subfigures 1-4 correspond to regime pattern 1-4,
respectively. The x-axis is time and the y-axis is the probability. The
shaded regions correspond to regime 2. }${\footnotesize (N,T)=(100,300)}$%
{\footnotesize \ and }${\footnotesize (\rho ,}${\footnotesize $\zeta $}$%
{\footnotesize ,}${\footnotesize $\xi $}${\footnotesize )=(0,0,0)}$%
{\footnotesize .}%
\end{figure}%

\begin{figure}[tbp]%
\begin{center}%
\caption{Unsmoothed Probabilities of Regime 2 for DGP 1}%
\begin{tabular}{l}
\FRAME{itbpF}{6.0226in}{1.6561in}{0in}{}{}{2-a.png}{\special{language
"Scientific Word";type "GRAPHIC";maintain-aspect-ratio TRUE;display
"USEDEF";valid_file "F";width 6.0226in;height 1.6561in;depth
0in;original-width 41.9996in;original-height 11.6663in;cropleft "0";croptop
"1";cropright "1.0225";cropbottom "0";filename '2-a.png';file-properties
"XNPEU";}} \\ 
\FRAME{itbpF}{6.0226in}{1.6561in}{0in}{}{}{2-b.png}{\special{language
"Scientific Word";type "GRAPHIC";maintain-aspect-ratio TRUE;display
"USEDEF";valid_file "F";width 6.0226in;height 1.6561in;depth
0in;original-width 42.0282in;original-height 11.6663in;cropleft "0";croptop
"1";cropright "1.0219";cropbottom "0";filename '2-b.png';file-properties
"XNPEU";}} \\ 
\FRAME{itbpF}{6.0226in}{1.6561in}{0in}{}{}{2-c.png}{\special{language
"Scientific Word";type "GRAPHIC";maintain-aspect-ratio TRUE;display
"USEDEF";valid_file "F";width 6.0226in;height 1.6561in;depth
0in;original-width 42.0558in;original-height 11.6663in;cropleft "0";croptop
"1";cropright "1.0212";cropbottom "0";filename '2-c.png';file-properties
"XNPEU";}} \\ 
\FRAME{itbpF}{6.0226in}{1.6561in}{0in}{}{}{2-d.png}{\special{language
"Scientific Word";type "GRAPHIC";maintain-aspect-ratio TRUE;display
"USEDEF";valid_file "F";width 6.0226in;height 1.6561in;depth
0in;original-width 42.0282in;original-height 11.6663in;cropleft "0";croptop
"1";cropright "1.0219";cropbottom "0";filename '2-d.png';file-properties
"XNPEU";}}%
\end{tabular}%
\label{unsmoothed 2-2}%
\end{center}%
{\footnotesize Notes: Subfigures 1-4 correspond to regime pattern 1-4,
respectively. The x-axis is time and the y-axis is the probability. The
shaded regions correspond to regime 2. }${\footnotesize (N,T)=(100,300)}$%
{\footnotesize \ and }${\footnotesize (\rho ,}${\footnotesize $\zeta $}$%
{\footnotesize ,}${\footnotesize $\xi $}${\footnotesize )=(0,0,0)}$%
{\footnotesize .}%
\end{figure}%

\begin{figure}[tbp]%
\begin{center}%
\caption{Smoothed Probabilities of Regime 2 for DGP 2}%
\begin{tabular}{l}
\FRAME{itbpF}{6.0226in}{1.6561in}{0in}{}{}{3a.png}{\special{language
"Scientific Word";type "GRAPHIC";maintain-aspect-ratio TRUE;display
"USEDEF";valid_file "F";width 6.0226in;height 1.6561in;depth
0in;original-width 42.2504in;original-height 11.6663in;cropleft "0";croptop
"1";cropright "1.0165";cropbottom "0";filename '3A.png';file-properties
"XNPEU";}} \\ 
\FRAME{itbpF}{6.0226in}{1.6561in}{0in}{}{}{3-b.png}{\special{language
"Scientific Word";type "GRAPHIC";maintain-aspect-ratio TRUE;display
"USEDEF";valid_file "F";width 6.0226in;height 1.6561in;depth
0in;original-width 42.0558in;original-height 11.6663in;cropleft "0";croptop
"1";cropright "1.0212";cropbottom "0";filename '3-b.png';file-properties
"XNPEU";}} \\ 
\FRAME{itbpF}{6.0226in}{1.6561in}{0in}{}{}{3-c.png}{\special{language
"Scientific Word";type "GRAPHIC";maintain-aspect-ratio TRUE;display
"USEDEF";valid_file "F";width 6.0226in;height 1.6561in;depth
0in;original-width 42.0558in;original-height 11.6663in;cropleft "0";croptop
"1";cropright "1.0212";cropbottom "0";filename '3-c.png';file-properties
"XNPEU";}} \\ 
\FRAME{itbpF}{6.0226in}{1.6561in}{0in}{}{}{3-d.png}{\special{language
"Scientific Word";type "GRAPHIC";maintain-aspect-ratio TRUE;display
"USEDEF";valid_file "F";width 6.0226in;height 1.6561in;depth
0in;original-width 42.0282in;original-height 11.6663in;cropleft "0";croptop
"1";cropright "1.0219";cropbottom "0";filename '3-d.png';file-properties
"XNPEU";}}%
\end{tabular}%
\label{smoothed 2-1}%
\end{center}%
{\footnotesize Notes: Subfigures 1-4 correspond to regime pattern 1-4,
respectively. The x-axis is time and the y-axis is the probability. The
shaded regions correspond to regime 2. }${\footnotesize (N,T)=(100,300)}$%
{\footnotesize \ and }${\footnotesize (\rho ,}${\footnotesize $\zeta $}$%
{\footnotesize ,}${\footnotesize $\xi $}${\footnotesize )=(0,0,0)}$%
{\footnotesize .}%
\end{figure}%

\begin{figure}[tbp]%
\begin{center}%
\caption{Unsmoothed Probabilities of Regime 2 for DGP 2}%
\begin{tabular}{l}
\FRAME{itbpF}{6.0226in}{1.6561in}{0in}{}{}{4-a.png}{\special{language
"Scientific Word";type "GRAPHIC";maintain-aspect-ratio TRUE;display
"USEDEF";valid_file "F";width 6.0226in;height 1.6561in;depth
0in;original-width 42.1112in;original-height 11.6663in;cropleft "0";croptop
"1";cropright "1.0199";cropbottom "0";filename '4-a.png';file-properties
"XNPEU";}} \\ 
\FRAME{itbpF}{6.0226in}{1.6561in}{0in}{}{}{4-b.png}{\special{language
"Scientific Word";type "GRAPHIC";maintain-aspect-ratio TRUE;display
"USEDEF";valid_file "F";width 6.0226in;height 1.6561in;depth
0in;original-width 42.0835in;original-height 11.6663in;cropleft "0";croptop
"1";cropright "1.0205";cropbottom "0";filename '4-b.png';file-properties
"XNPEU";}} \\ 
\FRAME{itbpF}{6.0226in}{1.6561in}{0in}{}{}{4c.png}{\special{language
"Scientific Word";type "GRAPHIC";maintain-aspect-ratio TRUE;display
"USEDEF";valid_file "F";width 6.0226in;height 1.6561in;depth
0in;original-width 42.2504in;original-height 11.6663in;cropleft "0";croptop
"1";cropright "1.0165";cropbottom "0";filename '4C.png';file-properties
"XNPEU";}} \\ 
\FRAME{itbpF}{6.0226in}{1.6561in}{0in}{}{}{4-d.png}{\special{language
"Scientific Word";type "GRAPHIC";maintain-aspect-ratio TRUE;display
"USEDEF";valid_file "F";width 6.0226in;height 1.6561in;depth
0in;original-width 42.0835in;original-height 11.6663in;cropleft "0";croptop
"1";cropright "1.0205";cropbottom "0";filename '4-d.png';file-properties
"XNPEU";}}%
\end{tabular}%
\label{unsmoothed 2-1}%
\end{center}%
{\footnotesize Notes: Subfigures 1-4 correspond to regime pattern 1-4,
respectively. The x-axis is time and the y-axis is the probability. The
shaded regions correspond to regime 2. }${\footnotesize (N,T)=(100,300)}$%
{\footnotesize \ and }${\footnotesize (\rho ,}${\footnotesize $\zeta $}$%
{\footnotesize ,}${\footnotesize $\xi $}${\footnotesize )=(0,0,0)}$%
{\footnotesize .}%
\end{figure}%

\begin{figure}[tbp]%
\begin{center}%
\caption{Smoothed and Unsmoothed Probabilities of Regime 2 for Regime
Pattern 1, $(N,T)=(200,300)$ and $(\rho,\alpha,\beta)=(0,0,0)$}%
\begin{tabular}{l}
\FRAME{itbpF}{6.0226in}{1.6561in}{0in}{}{}{5-a.png}{\special{language
"Scientific Word";type "GRAPHIC";maintain-aspect-ratio TRUE;display
"USEDEF";valid_file "F";width 6.0226in;height 1.6561in;depth
0in;original-width 42.2781in;original-height 11.6663in;cropleft "0";croptop
"1";cropright "1.0154";cropbottom "0";filename '5-a.png';file-properties
"XNPEU";}} \\ 
\FRAME{itbpF}{6.0226in}{1.6561in}{0in}{}{}{5-b.png}{\special{language
"Scientific Word";type "GRAPHIC";maintain-aspect-ratio TRUE;display
"USEDEF";valid_file "F";width 6.0226in;height 1.6561in;depth
0in;original-width 42.2781in;original-height 11.6663in;cropleft "0";croptop
"1";cropright "1.0154";cropbottom "0";filename '5-b.png';file-properties
"XNPEU";}} \\ 
\FRAME{itbpF}{6.0226in}{1.6561in}{0in}{}{}{5-c.png}{\special{language
"Scientific Word";type "GRAPHIC";maintain-aspect-ratio TRUE;display
"USEDEF";valid_file "F";width 6.0226in;height 1.6561in;depth
0in;original-width 42.2219in;original-height 11.6663in;cropleft "0";croptop
"1";cropright "1.0172";cropbottom "0";filename '5-c.png';file-properties
"XNPEU";}} \\ 
\FRAME{itbpF}{6.0226in}{1.6561in}{0in}{}{}{5-d.png}{\special{language
"Scientific Word";type "GRAPHIC";maintain-aspect-ratio TRUE;display
"USEDEF";valid_file "F";width 6.0226in;height 1.6561in;depth
0in;original-width 42.2219in;original-height 11.6663in;cropleft "0";croptop
"1";cropright "1.0172";cropbottom "0";filename '5-d.png';file-properties
"XNPEU";}}%
\end{tabular}%
\label{NT}%
\end{center}%
{\footnotesize Notes: Subfigures 1-4 correspond to smoothed probabilities
for DGP1, unsmoothed probabilities for DGP1, smoothed probabilities for DGP2
and unsmoothed probabilities for DGP2, respectively. The x-axis is time and
the y-axis is the probability. The shaded regions correspond to regime 2.}%
\end{figure}%

\begin{figure}[tbp]%
\begin{center}%
\caption{Smoothed Probabilities of Regime 2 for Regime Pattern 1, $(N,T)=(100,300)$ and
$(\rho,\alpha,\beta)=(0.5,0,0)$ or $(\rho,\alpha,\beta)=(0,0.5,0.5)$}%
\begin{tabular}{l}
\FRAME{itbpF}{6.0226in}{1.6561in}{0in}{}{}{6-a.png}{\special{language
"Scientific Word";type "GRAPHIC";maintain-aspect-ratio TRUE;display
"USEDEF";valid_file "F";width 6.0226in;height 1.6561in;depth
0in;original-width 42.3058in;original-height 11.6663in;cropleft "0";croptop
"1";cropright "1.0152";cropbottom "0";filename '6-a.png';file-properties
"XNPEU";}} \\ 
\FRAME{itbpF}{6.0226in}{1.6561in}{0in}{}{}{6-b.png}{\special{language
"Scientific Word";type "GRAPHIC";maintain-aspect-ratio TRUE;display
"USEDEF";valid_file "F";width 6.0226in;height 1.6561in;depth
0in;original-width 42.2504in;original-height 11.6663in;cropleft "0";croptop
"1";cropright "1.0165";cropbottom "0";filename '6-b.png';file-properties
"XNPEU";}} \\ 
\FRAME{itbpF}{6.0226in}{1.6561in}{0in}{}{}{6-c.png}{\special{language
"Scientific Word";type "GRAPHIC";maintain-aspect-ratio TRUE;display
"USEDEF";valid_file "F";width 6.0226in;height 1.6561in;depth
0in;original-width 42.3058in;original-height 11.6663in;cropleft "0";croptop
"1";cropright "1.0152";cropbottom "0";filename '6-c.png';file-properties
"XNPEU";}} \\ 
\FRAME{itbpF}{6.0226in}{1.6561in}{0in}{}{}{6-d.png}{\special{language
"Scientific Word";type "GRAPHIC";maintain-aspect-ratio TRUE;display
"USEDEF";valid_file "F";width 6.0226in;height 1.6561in;depth
0in;original-width 42.3058in;original-height 11.6663in;cropleft "0";croptop
"1";cropright "1.0152";cropbottom "0";filename '6-d.png';file-properties
"XNPEU";}}%
\end{tabular}%
\label{rho-alpha-beta}%
\end{center}%
{\footnotesize Notes: Subfigures 1-4 correspond to smoothed probabilities
for DGP1 with }${\footnotesize (\rho ,}${\footnotesize $\zeta $}$%
{\footnotesize ,}${\footnotesize $\xi $}${\footnotesize )=(0.5,0,0)}$%
{\footnotesize , DGP1 with }${\footnotesize (\rho ,}${\footnotesize $\zeta $}%
${\footnotesize ,}${\footnotesize $\xi $}${\footnotesize )=(0,0.5,0.5)}$%
{\footnotesize , DGP2 with }${\footnotesize (\rho ,}${\footnotesize $\zeta $}%
${\footnotesize ,}${\footnotesize $\xi $}${\footnotesize )=(0.5,0,0)}$%
{\footnotesize \ and DGP2 with }${\footnotesize (\rho ,}${\footnotesize $%
\zeta $}${\footnotesize ,}${\footnotesize $\xi $}${\footnotesize %
)=(0,0.5,0.5)}${\footnotesize , respectively. The x-axis is time and the
y-axis is the probability. The shaded regions correspond to regime 2.}%
\end{figure}%

\begin{figure}[tbp]%
\begin{center}%
\caption{Smoothed and Unsmoothed Probabilities of Regime 2 for Regime
Pattern 1, $(N,T)=(50,500)$ and $(N,T)=(500,50)$}%
\begin{tabular}{l}
\FRAME{itbpF}{6.0224in}{1.6571in}{0in}{}{}{p5-smn50-dgp1.png}{\special%
{language "Scientific Word";type "GRAPHIC";maintain-aspect-ratio
TRUE;display "USEDEF";valid_file "F";width 6.0224in;height 1.6571in;depth
0in;original-width 42.1335in;original-height 11.1996in;cropleft "0";croptop
"1";cropright "0.9782";cropbottom "0";filename
'p5-smN50-dgp1.png';file-properties "XNPEU";}} \\ 
\FRAME{itbpF}{6.0224in}{1.6571in}{0in}{}{}{p5-smt50-dgp1.png}{\special%
{language "Scientific Word";type "GRAPHIC";maintain-aspect-ratio
TRUE;display "USEDEF";valid_file "F";width 6.0224in;height 1.6571in;depth
0in;original-width 42.16in;original-height 11.1996in;cropleft "0";croptop
"1";cropright "0.9770";cropbottom "0";filename
'p5-smT50-dgp1.png';file-properties "XNPEU";}} \\ 
\FRAME{itbpF}{6.0224in}{1.6571in}{0in}{}{}{p5-usmn50-dgp1.png}{\special%
{language "Scientific Word";type "GRAPHIC";maintain-aspect-ratio
TRUE;display "USEDEF";valid_file "F";width 6.0224in;height 1.6571in;depth
0in;original-width 42.1866in;original-height 11.1996in;cropleft "0";croptop
"1";cropright "0.9765";cropbottom "0";filename
'p5-usmN50-dgp1.png';file-properties "XNPEU";}} \\ 
\FRAME{itbpF}{6.0224in}{1.6571in}{0in}{}{}{p5-usmt50-dgp1.png}{\special%
{language "Scientific Word";type "GRAPHIC";maintain-aspect-ratio
TRUE;display "USEDEF";valid_file "F";width 6.0224in;height 1.6571in;depth
0in;original-width 42.1335in;original-height 11.1996in;cropleft "0";croptop
"1";cropright "0.9782";cropbottom "0";filename
'p5-usmT50-dgp1.png';file-properties "XNPEU";}}%
\end{tabular}%
\label{N50T50}%
\end{center}%
{\footnotesize Notes: Subfigures 1-2 correspond to smoothed probabilities
under regime pattern 4 and DGP1 with }${\footnotesize (N,T)=(50,500)}$ 
{\footnotesize or }${\footnotesize (500,50)}${\footnotesize , respectively.
Subfigures 3-4 correspond to unsmoothed probabilities under regime pattern 4
and DGP1 with }${\footnotesize (N,T)=(500,50)}$ {\footnotesize or }$%
{\footnotesize (500,50)}${\footnotesize , respectively. The x-axis is time
and the y-axis is the probability. The shaded regions correspond to regime 2.%
}%
\end{figure}%

\begin{figure}[tbp]%
\begin{center}%
\caption{Smoothed and Unsmoothed Probabilities of Regime 2 for DGP 4}%
\begin{tabular}{l}
\FRAME{itbpF}{6.0224in}{1.6571in}{0in}{}{}{p1-sm-dgp4.png}{\special{language
"Scientific Word";type "GRAPHIC";maintain-aspect-ratio TRUE;display
"USEDEF";valid_file "F";width 6.0224in;height 1.6571in;depth
0in;original-width 42.16in;original-height 11.1996in;cropleft "0";croptop
"1";cropright "0.9770";cropbottom "0";filename
'p1-sm-dgp4.png';file-properties "XNPEU";}} \\ 
\FRAME{itbpF}{6.0224in}{1.6571in}{0in}{}{}{p2-sm-dgp4.png}{\special{language
"Scientific Word";type "GRAPHIC";maintain-aspect-ratio TRUE;display
"USEDEF";valid_file "F";width 6.0224in;height 1.6571in;depth
0in;original-width 42.16in;original-height 11.1996in;cropleft "0";croptop
"1";cropright "0.9770";cropbottom "0";filename
'p2-sm-dgp4.png';file-properties "XNPEU";}} \\ 
\FRAME{itbpF}{6.0224in}{1.6571in}{0in}{}{}{p1-usm-dgp4.png}{\special%
{language "Scientific Word";type "GRAPHIC";maintain-aspect-ratio
TRUE;display "USEDEF";valid_file "F";width 6.0224in;height 1.6571in;depth
0in;original-width 42.16in;original-height 11.1996in;cropleft "0";croptop
"1";cropright "0.9770";cropbottom "0";filename
'p1-usm-dgp4.png';file-properties "XNPEU";}} \\ 
\FRAME{itbpF}{6.0224in}{1.6571in}{0in}{}{}{p2-usm-dgp4.png}{\special%
{language "Scientific Word";type "GRAPHIC";maintain-aspect-ratio
TRUE;display "USEDEF";valid_file "F";width 6.0224in;height 1.6571in;depth
0in;original-width 42.1335in;original-height 11.1996in;cropleft "0";croptop
"1";cropright "0.9782";cropbottom "0";filename
'p2-usm-dgp4.png';file-properties "XNPEU";}}%
\end{tabular}%
\label{dgp4}%
\end{center}%
{\footnotesize Notes: Subfigures 1-2 correspond to smoothed probabilities
for DGP4 under regime patterns 1 and 2, respectively. Subfigures 3-4
correspond to unsmoothed probabilities for DGP4 under regime patterns 1 and
2, respectively. The x-axis is time and the y-axis is the probability. The
shaded regions correspond to regime 2.}%
\end{figure}%

\begin{figure}[tbp]%
\begin{center}%
\caption{Histograms of the Estimated Loadings and the Estimated Factors  for
Regime Pattern 1}%
\begin{tabular}{lll}
\FRAME{itbpF}{1.8066in}{1.4053in}{0in}{}{}{p1-sm-l1.png}{\special{language
"Scientific Word";type "GRAPHIC";maintain-aspect-ratio TRUE;display
"USEDEF";valid_file "F";width 1.8066in;height 1.4053in;depth
0in;original-width 15.5554in;original-height 11.6663in;cropleft "0";croptop
"1";cropright "0.9681";cropbottom "0";filename
'p1-sm-l1.png';file-properties "XNPEU";}} & \FRAME{itbpF}{1.8066in}{1.4053in%
}{0in}{}{}{p1-sm-l2.png}{\special{language "Scientific Word";type
"GRAPHIC";maintain-aspect-ratio TRUE;display "USEDEF";valid_file "F";width
1.8066in;height 1.4053in;depth 0in;original-width 15.5554in;original-height
11.6663in;cropleft "0";croptop "1";cropright "0.9685";cropbottom
"0";filename 'p1-sm-l2.png';file-properties "XNPEU";}} & \FRAME{itbpF}{%
1.8066in}{1.4053in}{0in}{}{}{p1-sm-f.png}{\special{language "Scientific
Word";type "GRAPHIC";maintain-aspect-ratio TRUE;display "USEDEF";valid_file
"F";width 1.8066in;height 1.4053in;depth 0in;original-width
15.5554in;original-height 11.6663in;cropleft "0";croptop "1";cropright
"0.9685";cropbottom "0";filename 'p1-sm-f.png';file-properties "XNPEU";}} \\ 
\FRAME{itbpF}{1.8066in}{1.4053in}{0in}{}{}{p1-usm-l1.png}{\special{language
"Scientific Word";type "GRAPHIC";maintain-aspect-ratio TRUE;display
"USEDEF";valid_file "F";width 1.8066in;height 1.4053in;depth
0in;original-width 15.5554in;original-height 11.6663in;cropleft "0";croptop
"1";cropright "0.9685";cropbottom "0";filename
'p1-usm-l1.png';file-properties "XNPEU";}} & \FRAME{itbpF}{1.8066in}{1.4053in%
}{0in}{}{}{p1-usm-l2.png}{\special{language "Scientific Word";type
"GRAPHIC";maintain-aspect-ratio TRUE;display "USEDEF";valid_file "F";width
1.8066in;height 1.4053in;depth 0in;original-width 15.5554in;original-height
11.6663in;cropleft "0";croptop "1";cropright "0.9685";cropbottom
"0";filename 'p1-usm-l2.png';file-properties "XNPEU";}} & \FRAME{itbpF}{%
1.8066in}{1.4053in}{0in}{}{}{p1-usm-f.png}{\special{language "Scientific
Word";type "GRAPHIC";maintain-aspect-ratio TRUE;display "USEDEF";valid_file
"F";width 1.8066in;height 1.4053in;depth 0in;original-width
15.5554in;original-height 11.6663in;cropleft "0";croptop "1";cropright
"0.9685";cropbottom "0";filename 'p1-usm-f.png';file-properties "XNPEU";}}
\\ 
\FRAME{itbpF}{1.8066in}{1.4053in}{0in}{}{}{p1-smrab-l1.png}{\special%
{language "Scientific Word";type "GRAPHIC";maintain-aspect-ratio
TRUE;display "USEDEF";valid_file "F";width 1.8066in;height 1.4053in;depth
0in;original-width 15.5554in;original-height 11.6663in;cropleft "0";croptop
"1";cropright "0.9685";cropbottom "0";filename
'p1-smRAB-l1.png';file-properties "XNPEU";}} & \FRAME{itbpF}{1.8066in}{%
1.4053in}{0in}{}{}{p1-smrab-l2.png}{\special{language "Scientific Word";type
"GRAPHIC";maintain-aspect-ratio TRUE;display "USEDEF";valid_file "F";width
1.8066in;height 1.4053in;depth 0in;original-width 15.5554in;original-height
11.6663in;cropleft "0";croptop "1";cropright "0.9685";cropbottom
"0";filename 'p1-smRAB-l2.png';file-properties "XNPEU";}} & \FRAME{itbpF}{%
1.8066in}{1.4053in}{0in}{}{}{p1-smrab-f.png}{\special{language "Scientific
Word";type "GRAPHIC";maintain-aspect-ratio TRUE;display "USEDEF";valid_file
"F";width 1.8066in;height 1.4053in;depth 0in;original-width
15.5554in;original-height 11.6663in;cropleft "0";croptop "1";cropright
"0.9685";cropbottom "0";filename 'p1-smRAB-f.png';file-properties "XNPEU";}}
\\ 
\FRAME{itbpF}{1.8066in}{1.4053in}{0in}{}{}{p1-smn200-l1.png}{\special%
{language "Scientific Word";type "GRAPHIC";maintain-aspect-ratio
TRUE;display "USEDEF";valid_file "F";width 1.8066in;height 1.4053in;depth
0in;original-width 15.5554in;original-height 11.6663in;cropleft "0";croptop
"1";cropright "0.9685";cropbottom "0";filename
'p1-smN200-l1.png';file-properties "XNPEU";}} & \FRAME{itbpF}{1.8066in}{%
1.4053in}{0in}{}{}{p1-smn200-l2.png}{\special{language "Scientific
Word";type "GRAPHIC";maintain-aspect-ratio TRUE;display "USEDEF";valid_file
"F";width 1.8066in;height 1.4053in;depth 0in;original-width
15.5554in;original-height 11.6663in;cropleft "0";croptop "1";cropright
"0.9685";cropbottom "0";filename 'p1-smN200-l2.png';file-properties "XNPEU";}%
} & \FRAME{itbpF}{1.8066in}{1.4053in}{0in}{}{}{p1-smn200-f.png}{\special%
{language "Scientific Word";type "GRAPHIC";maintain-aspect-ratio
TRUE;display "USEDEF";valid_file "F";width 1.8066in;height 1.4053in;depth
0in;original-width 15.5554in;original-height 11.6663in;cropleft "0";croptop
"1";cropright "0.9685";cropbottom "0";filename
'p1-smN200-f.png';file-properties "XNPEU";}}%
\end{tabular}%
\label{histo-BCDC}%
\end{center}%
{\footnotesize Notes: Subfigures in the first to the fourth row correspond
to the smoothed algorithm with }${\footnotesize \rho =}${\footnotesize $%
\zeta $}${\footnotesize =}${\footnotesize $\xi $}${\footnotesize =0}$ 
{\footnotesize and }${\footnotesize (N,T)=(100,300)}${\footnotesize , the
unsmoothed algorithm with }${\footnotesize \rho =}${\footnotesize $\zeta $}$%
{\footnotesize =}${\footnotesize $\xi $}${\footnotesize =0}$ {\footnotesize %
and }${\footnotesize (N,T)=(100,300)}${\footnotesize , the smoothed
algorithm with }${\footnotesize \rho =}${\footnotesize $\zeta $}$%
{\footnotesize =}${\footnotesize $\xi $}${\footnotesize =0.5}$ 
{\footnotesize and }${\footnotesize (N,T)=(100,300)}${\footnotesize , and
the smoothed algorithm with }${\footnotesize \rho =}${\footnotesize $\zeta $}%
${\footnotesize =}${\footnotesize $\xi $}${\footnotesize =0}$ {\footnotesize %
and }${\footnotesize (N,T)=(200,300)}${\footnotesize , respectively.
Subfigures in the first to the third column correspond to the estimated
loadings for regime 1, the estimated loadings for regime 2 and the estimated
factors, respectively. The curve overlaid on the histograms is the standard
normal density function.}%
\end{figure}%

\begin{figure}[tbp]%
\begin{center}%
\caption{Histograms of the Estimated Loadings and the Estimated Factors  for
Regime Pattern 2}%
\begin{tabular}{lll}
\FRAME{itbpF}{1.8066in}{1.4053in}{0in}{}{}{p2-sm-l1.png}{\special{language
"Scientific Word";type "GRAPHIC";maintain-aspect-ratio TRUE;display
"USEDEF";valid_file "F";width 1.8066in;height 1.4053in;depth
0in;original-width 15.5554in;original-height 11.6663in;cropleft "0";croptop
"1";cropright "0.9685";cropbottom "0";filename
'p2-sm-l1.png';file-properties "XNPEU";}} & \FRAME{itbpF}{1.8066in}{1.4053in%
}{0in}{}{}{p2-sm-l2.png}{\special{language "Scientific Word";type
"GRAPHIC";maintain-aspect-ratio TRUE;display "USEDEF";valid_file "F";width
1.8066in;height 1.4053in;depth 0in;original-width 15.5554in;original-height
11.6663in;cropleft "0";croptop "1";cropright "0.9685";cropbottom
"0";filename 'p2-sm-l2.png';file-properties "XNPEU";}} & \FRAME{itbpF}{%
1.8066in}{1.4053in}{0in}{}{}{p2-sm-f.png}{\special{language "Scientific
Word";type "GRAPHIC";maintain-aspect-ratio TRUE;display "USEDEF";valid_file
"F";width 1.8066in;height 1.4053in;depth 0in;original-width
15.5554in;original-height 11.6663in;cropleft "0";croptop "1";cropright
"0.9685";cropbottom "0";filename 'p2-sm-f.png';file-properties "XNPEU";}} \\ 
\FRAME{itbpF}{1.8066in}{1.4053in}{0in}{}{}{p2-usm-l1.png}{\special{language
"Scientific Word";type "GRAPHIC";maintain-aspect-ratio TRUE;display
"USEDEF";valid_file "F";width 1.8066in;height 1.4053in;depth
0in;original-width 15.5554in;original-height 11.6663in;cropleft "0";croptop
"1";cropright "0.9685";cropbottom "0";filename
'p2-usm-l1.png';file-properties "XNPEU";}} & \FRAME{itbpF}{1.8066in}{1.4053in%
}{0in}{}{}{p2-usm-l2.png}{\special{language "Scientific Word";type
"GRAPHIC";maintain-aspect-ratio TRUE;display "USEDEF";valid_file "F";width
1.8066in;height 1.4053in;depth 0in;original-width 15.5554in;original-height
11.6663in;cropleft "0";croptop "1";cropright "0.9685";cropbottom
"0";filename 'p2-usm-l2.png';file-properties "XNPEU";}} & \FRAME{itbpF}{%
1.8066in}{1.4053in}{0in}{}{}{p2-usm-f.png}{\special{language "Scientific
Word";type "GRAPHIC";maintain-aspect-ratio TRUE;display "USEDEF";valid_file
"F";width 1.8066in;height 1.4053in;depth 0in;original-width
15.5554in;original-height 11.6663in;cropleft "0";croptop "1";cropright
"0.9685";cropbottom "0";filename 'p2-usm-f.png';file-properties "XNPEU";}}
\\ 
\FRAME{itbpF}{1.8066in}{1.4053in}{0in}{}{}{p2-smrab-l1.png}{\special%
{language "Scientific Word";type "GRAPHIC";maintain-aspect-ratio
TRUE;display "USEDEF";valid_file "F";width 1.8066in;height 1.4053in;depth
0in;original-width 15.5554in;original-height 11.6663in;cropleft "0";croptop
"1";cropright "0.9685";cropbottom "0";filename
'p2-smRAB-l1.png';file-properties "XNPEU";}} & \FRAME{itbpF}{1.8066in}{%
1.4053in}{0in}{}{}{p2-smrab-l2.png}{\special{language "Scientific Word";type
"GRAPHIC";maintain-aspect-ratio TRUE;display "USEDEF";valid_file "F";width
1.8066in;height 1.4053in;depth 0in;original-width 15.5554in;original-height
11.6663in;cropleft "0";croptop "1";cropright "0.9685";cropbottom
"0";filename 'p2-smRAB-l2.png';file-properties "XNPEU";}} & \FRAME{itbpF}{%
1.8066in}{1.4053in}{0in}{}{}{p2-smrab-f.png}{\special{language "Scientific
Word";type "GRAPHIC";maintain-aspect-ratio TRUE;display "USEDEF";valid_file
"F";width 1.8066in;height 1.4053in;depth 0in;original-width
15.5554in;original-height 11.6663in;cropleft "0";croptop "1";cropright
"0.9685";cropbottom "0";filename 'p2-smRAB-f.png';file-properties "XNPEU";}}
\\ 
\FRAME{itbpF}{1.8066in}{1.4053in}{0in}{}{}{p2-smn200-l1.png}{\special%
{language "Scientific Word";type "GRAPHIC";maintain-aspect-ratio
TRUE;display "USEDEF";valid_file "F";width 1.8066in;height 1.4053in;depth
0in;original-width 15.5554in;original-height 11.6663in;cropleft "0";croptop
"1";cropright "0.9685";cropbottom "0";filename
'p2-smN200-l1.png';file-properties "XNPEU";}} & \FRAME{itbpF}{1.8066in}{%
1.4053in}{0in}{}{}{p2-smn200-l2.png}{\special{language "Scientific
Word";type "GRAPHIC";maintain-aspect-ratio TRUE;display "USEDEF";valid_file
"F";width 1.8066in;height 1.4053in;depth 0in;original-width
15.5554in;original-height 11.6663in;cropleft "0";croptop "1";cropright
"0.9685";cropbottom "0";filename 'p2-smN200-l2.png';file-properties "XNPEU";}%
} & \FRAME{itbpF}{1.8066in}{1.4053in}{0in}{}{}{p2-smn200-f.png}{\special%
{language "Scientific Word";type "GRAPHIC";maintain-aspect-ratio
TRUE;display "USEDEF";valid_file "F";width 1.8066in;height 1.4053in;depth
0in;original-width 15.5554in;original-height 11.6663in;cropleft "0";croptop
"1";cropright "0.9685";cropbottom "0";filename
'p2-smN200-f.png';file-properties "XNPEU";}}%
\end{tabular}%
\label{histo-break}%
\end{center}%
{\footnotesize Notes: Subfigures in the first to the fourth row correspond
to the smoothed algorithm with }${\footnotesize \rho =}${\footnotesize $%
\zeta $}${\footnotesize =}${\footnotesize $\xi $}${\footnotesize =0}$ 
{\footnotesize and }${\footnotesize (N,T)=(100,300)}${\footnotesize , the
unsmoothed algorithm with }${\footnotesize \rho =}${\footnotesize $\zeta $}$%
{\footnotesize =}${\footnotesize $\xi $}${\footnotesize =0}$ {\footnotesize %
and }${\footnotesize (N,T)=(100,300)}${\footnotesize , the smoothed
algorithm with }${\footnotesize \rho =}${\footnotesize $\zeta $}$%
{\footnotesize =}${\footnotesize $\xi $}${\footnotesize =0.5}$ 
{\footnotesize and }${\footnotesize (N,T)=(100,300)}${\footnotesize , and
the smoothed algorithm with }${\footnotesize \rho =}${\footnotesize $\zeta $}%
${\footnotesize =}${\footnotesize $\xi $}${\footnotesize =0}$ {\footnotesize %
and }${\footnotesize (N,T)=(200,300)}${\footnotesize , respectively.
Subfigures in the first to the third column correspond to the estimated
loadings for regime 1, the estimated loadings for regime 2 and the estimated
factors, respectively. The curve overlaid on the histograms is the standard
normal density function.}%
\end{figure}%

\begin{figure}[tbp]%
\begin{center}%
\caption{Histograms of the Estimated Loadings and the Estimated Factors  for
Regime Pattern 3}%
\begin{tabular}{lll}
\FRAME{itbpF}{1.8066in}{1.4053in}{0in}{}{}{p3-sm-l1.png}{\special{language
"Scientific Word";type "GRAPHIC";maintain-aspect-ratio TRUE;display
"USEDEF";valid_file "F";width 1.8066in;height 1.4053in;depth
0in;original-width 15.5554in;original-height 11.6663in;cropleft "0";croptop
"1";cropright "0.9685";cropbottom "0";filename
'p3-sm-l1.png';file-properties "XNPEU";}} & \FRAME{itbpF}{1.8066in}{1.4053in%
}{0in}{}{}{p3-sm-l2.png}{\special{language "Scientific Word";type
"GRAPHIC";maintain-aspect-ratio TRUE;display "USEDEF";valid_file "F";width
1.8066in;height 1.4053in;depth 0in;original-width 15.5554in;original-height
11.6663in;cropleft "0";croptop "1";cropright "0.9685";cropbottom
"0";filename 'p3-sm-l2.png';file-properties "XNPEU";}} & \FRAME{itbpF}{%
1.8066in}{1.4053in}{0in}{}{}{p3-sm-f.png}{\special{language "Scientific
Word";type "GRAPHIC";maintain-aspect-ratio TRUE;display "USEDEF";valid_file
"F";width 1.8066in;height 1.4053in;depth 0in;original-width
15.5554in;original-height 11.6663in;cropleft "0";croptop "1";cropright
"0.9685";cropbottom "0";filename 'p3-sm-f.png';file-properties "XNPEU";}} \\ 
\FRAME{itbpF}{1.8066in}{1.4053in}{0in}{}{}{p3-usm-l1.png}{\special{language
"Scientific Word";type "GRAPHIC";maintain-aspect-ratio TRUE;display
"USEDEF";valid_file "F";width 1.8066in;height 1.4053in;depth
0in;original-width 15.5554in;original-height 11.6663in;cropleft "0";croptop
"1";cropright "0.9685";cropbottom "0";filename
'p3-usm-l1.png';file-properties "XNPEU";}} & \FRAME{itbpF}{1.8066in}{1.4053in%
}{0in}{}{}{p3-usm-l2.png}{\special{language "Scientific Word";type
"GRAPHIC";maintain-aspect-ratio TRUE;display "USEDEF";valid_file "F";width
1.8066in;height 1.4053in;depth 0in;original-width 15.5554in;original-height
11.6663in;cropleft "0";croptop "1";cropright "0.9685";cropbottom
"0";filename 'p3-usm-l2.png';file-properties "XNPEU";}} & \FRAME{itbpF}{%
1.8066in}{1.4053in}{0in}{}{}{p3-usm-f.png}{\special{language "Scientific
Word";type "GRAPHIC";maintain-aspect-ratio TRUE;display "USEDEF";valid_file
"F";width 1.8066in;height 1.4053in;depth 0in;original-width
15.5554in;original-height 11.6663in;cropleft "0";croptop "1";cropright
"0.9685";cropbottom "0";filename 'p3-usm-f.png';file-properties "XNPEU";}}
\\ 
\FRAME{itbpF}{1.8066in}{1.4053in}{0in}{}{}{p3-smrab-l1.png}{\special%
{language "Scientific Word";type "GRAPHIC";maintain-aspect-ratio
TRUE;display "USEDEF";valid_file "F";width 1.8066in;height 1.4053in;depth
0in;original-width 15.5554in;original-height 11.6663in;cropleft "0";croptop
"1";cropright "0.9685";cropbottom "0";filename
'p3-smRAB-l1.png';file-properties "XNPEU";}} & \FRAME{itbpF}{1.8066in}{%
1.4053in}{0in}{}{}{p3-smrab-l2.png}{\special{language "Scientific Word";type
"GRAPHIC";maintain-aspect-ratio TRUE;display "USEDEF";valid_file "F";width
1.8066in;height 1.4053in;depth 0in;original-width 15.5554in;original-height
11.6663in;cropleft "0";croptop "1";cropright "0.9685";cropbottom
"0";filename 'p3-smRAB-l2.png';file-properties "XNPEU";}} & \FRAME{itbpF}{%
1.8066in}{1.4053in}{0in}{}{}{p3-smrab-f.png}{\special{language "Scientific
Word";type "GRAPHIC";maintain-aspect-ratio TRUE;display "USEDEF";valid_file
"F";width 1.8066in;height 1.4053in;depth 0in;original-width
15.5554in;original-height 11.6663in;cropleft "0";croptop "1";cropright
"0.9685";cropbottom "0";filename 'p3-smRAB-f.png';file-properties "XNPEU";}}
\\ 
\FRAME{itbpF}{1.8066in}{1.4053in}{0in}{}{}{p3-smn200-l1.png}{\special%
{language "Scientific Word";type "GRAPHIC";maintain-aspect-ratio
TRUE;display "USEDEF";valid_file "F";width 1.8066in;height 1.4053in;depth
0in;original-width 15.5554in;original-height 11.6663in;cropleft "0";croptop
"1";cropright "0.9685";cropbottom "0";filename
'p3-smN200-l1.png';file-properties "XNPEU";}} & \FRAME{itbpF}{1.8066in}{%
1.4053in}{0in}{}{}{p3-smn200-l2.png}{\special{language "Scientific
Word";type "GRAPHIC";maintain-aspect-ratio TRUE;display "USEDEF";valid_file
"F";width 1.8066in;height 1.4053in;depth 0in;original-width
15.5554in;original-height 11.6663in;cropleft "0";croptop "1";cropright
"0.9685";cropbottom "0";filename 'p3-smN200-l2.png';file-properties "XNPEU";}%
} & \FRAME{itbpF}{1.8066in}{1.4053in}{0in}{}{}{p3-smn200-f.png}{\special%
{language "Scientific Word";type "GRAPHIC";maintain-aspect-ratio
TRUE;display "USEDEF";valid_file "F";width 1.8066in;height 1.4053in;depth
0in;original-width 15.5554in;original-height 11.6663in;cropleft "0";croptop
"1";cropright "0.9685";cropbottom "0";filename
'p3-smN200-f.png';file-properties "XNPEU";}}%
\end{tabular}%
\label{histo-breaks}%
\end{center}%
{\footnotesize Notes: Subfigures in the first to the fourth row correspond
to the smoothed algorithm with }${\footnotesize \rho =}${\footnotesize $%
\zeta $}${\footnotesize =}${\footnotesize $\xi $}${\footnotesize =0}$ 
{\footnotesize and }${\footnotesize (N,T)=(100,300)}${\footnotesize , the
unsmoothed algorithm with }${\footnotesize \rho =}${\footnotesize $\zeta $}$%
{\footnotesize =}${\footnotesize $\xi $}${\footnotesize =0}$ {\footnotesize %
and }${\footnotesize (N,T)=(100,300)}${\footnotesize , the smoothed
algorithm with }${\footnotesize \rho =}${\footnotesize $\zeta $}$%
{\footnotesize =}${\footnotesize $\xi $}${\footnotesize =0.5}$ 
{\footnotesize and }${\footnotesize (N,T)=(100,300)}${\footnotesize , and
the smoothed algorithm with }${\footnotesize \rho =}${\footnotesize $\zeta $}%
${\footnotesize =}${\footnotesize $\xi $}${\footnotesize =0}$ {\footnotesize %
and }${\footnotesize (N,T)=(200,300)}${\footnotesize , respectively.
Subfigures in the first to the third column correspond to the estimated
loadings for regime 1, the estimated loadings for regime 2 and the estimated
factors, respectively. The curve overlaid on the histograms is the standard
normal density function.}%
\end{figure}%

\begin{figure}[tbp]%
\begin{center}%
\caption{Histograms of the Estimated Loadings and the Estimated Factors for
Regime Pattern 4}%
\begin{tabular}{lll}
\FRAME{itbpF}{1.8066in}{1.4053in}{0in}{}{}{p4-sm-l1.png}{\special{language
"Scientific Word";type "GRAPHIC";maintain-aspect-ratio TRUE;display
"USEDEF";valid_file "F";width 1.8066in;height 1.4053in;depth
0in;original-width 15.5554in;original-height 11.6663in;cropleft "0";croptop
"1";cropright "0.9685";cropbottom "0";filename
'p4-sm-l1.png';file-properties "XNPEU";}} & \FRAME{itbpF}{1.8066in}{1.4053in%
}{0in}{}{}{p4-sm-l2.png}{\special{language "Scientific Word";type
"GRAPHIC";maintain-aspect-ratio TRUE;display "USEDEF";valid_file "F";width
1.8066in;height 1.4053in;depth 0in;original-width 15.5554in;original-height
11.6663in;cropleft "0";croptop "1";cropright "0.9685";cropbottom
"0";filename 'p4-sm-l2.png';file-properties "XNPEU";}} & \FRAME{itbpF}{%
1.8066in}{1.4053in}{0in}{}{}{p4-sm-f.png}{\special{language "Scientific
Word";type "GRAPHIC";maintain-aspect-ratio TRUE;display "USEDEF";valid_file
"F";width 1.8066in;height 1.4053in;depth 0in;original-width
15.5554in;original-height 11.6663in;cropleft "0";croptop "1";cropright
"0.9685";cropbottom "0";filename 'p4-sm-f.png';file-properties "XNPEU";}} \\ 
\FRAME{itbpF}{1.8066in}{1.4053in}{0in}{}{}{p4-usm-l1.png}{\special{language
"Scientific Word";type "GRAPHIC";maintain-aspect-ratio TRUE;display
"USEDEF";valid_file "F";width 1.8066in;height 1.4053in;depth
0in;original-width 15.5554in;original-height 11.6663in;cropleft "0";croptop
"1";cropright "0.9685";cropbottom "0";filename
'p4-usm-l1.png';file-properties "XNPEU";}} & \FRAME{itbpF}{1.8066in}{1.4053in%
}{0in}{}{}{p4-usm-l2.png}{\special{language "Scientific Word";type
"GRAPHIC";maintain-aspect-ratio TRUE;display "USEDEF";valid_file "F";width
1.8066in;height 1.4053in;depth 0in;original-width 15.5554in;original-height
11.6663in;cropleft "0";croptop "1";cropright "0.9685";cropbottom
"0";filename 'p4-usm-l2.png';file-properties "XNPEU";}} & \FRAME{itbpF}{%
1.8066in}{1.4053in}{0in}{}{}{p4-usm-f.png}{\special{language "Scientific
Word";type "GRAPHIC";maintain-aspect-ratio TRUE;display "USEDEF";valid_file
"F";width 1.8066in;height 1.4053in;depth 0in;original-width
15.5554in;original-height 11.6663in;cropleft "0";croptop "1";cropright
"0.9685";cropbottom "0";filename 'p4-usm-f.png';file-properties "XNPEU";}}
\\ 
\FRAME{itbpF}{1.8066in}{1.4053in}{0in}{}{}{p4-smrab-l1.png}{\special%
{language "Scientific Word";type "GRAPHIC";maintain-aspect-ratio
TRUE;display "USEDEF";valid_file "F";width 1.8066in;height 1.4053in;depth
0in;original-width 15.5554in;original-height 11.6663in;cropleft "0";croptop
"1";cropright "0.9685";cropbottom "0";filename
'p4-smRAB-l1.png';file-properties "XNPEU";}} & \FRAME{itbpF}{1.8066in}{%
1.4053in}{0in}{}{}{p4-smrab-l2.png}{\special{language "Scientific Word";type
"GRAPHIC";maintain-aspect-ratio TRUE;display "USEDEF";valid_file "F";width
1.8066in;height 1.4053in;depth 0in;original-width 15.5554in;original-height
11.6663in;cropleft "0";croptop "1";cropright "0.9685";cropbottom
"0";filename 'p4-smRAB-l2.png';file-properties "XNPEU";}} & \FRAME{itbpF}{%
1.8066in}{1.4053in}{0in}{}{}{p4-smrab-f.png}{\special{language "Scientific
Word";type "GRAPHIC";maintain-aspect-ratio TRUE;display "USEDEF";valid_file
"F";width 1.8066in;height 1.4053in;depth 0in;original-width
15.5554in;original-height 11.6663in;cropleft "0";croptop "1";cropright
"0.9685";cropbottom "0";filename 'p4-smRAB-f.png';file-properties "XNPEU";}}
\\ 
\FRAME{itbpF}{1.8066in}{1.4053in}{0in}{}{}{p4-smn200-l1.png}{\special%
{language "Scientific Word";type "GRAPHIC";maintain-aspect-ratio
TRUE;display "USEDEF";valid_file "F";width 1.8066in;height 1.4053in;depth
0in;original-width 15.5554in;original-height 11.6663in;cropleft "0";croptop
"1";cropright "0.9685";cropbottom "0";filename
'p4-smN200-l1.png';file-properties "XNPEU";}} & \FRAME{itbpF}{1.8066in}{%
1.4053in}{0in}{}{}{p4-smn200-l2.png}{\special{language "Scientific
Word";type "GRAPHIC";maintain-aspect-ratio TRUE;display "USEDEF";valid_file
"F";width 1.8066in;height 1.4053in;depth 0in;original-width
15.5554in;original-height 11.6663in;cropleft "0";croptop "1";cropright
"0.9685";cropbottom "0";filename 'p4-smN200-l2.png';file-properties "XNPEU";}%
} & \FRAME{itbpF}{1.8066in}{1.4053in}{0in}{}{}{p4-smn200-f.png}{\special%
{language "Scientific Word";type "GRAPHIC";maintain-aspect-ratio
TRUE;display "USEDEF";valid_file "F";width 1.8066in;height 1.4053in;depth
0in;original-width 15.5554in;original-height 11.6663in;cropleft "0";croptop
"1";cropright "0.9685";cropbottom "0";filename
'p4-smN200-f.png';file-properties "XNPEU";}}%
\end{tabular}%
\label{histo-Markov}%
\end{center}%
{\footnotesize Notes: Subfigures in the first to the fourth row correspond
to the smoothed algorithm with }${\footnotesize \rho =}${\footnotesize $%
\zeta $}${\footnotesize =}${\footnotesize $\xi $}${\footnotesize =0}$ 
{\footnotesize and }${\footnotesize (N,T)=(100,300)}${\footnotesize , the
unsmoothed algorithm with }${\footnotesize \rho =}${\footnotesize $\zeta $}$%
{\footnotesize =}${\footnotesize $\xi $}${\footnotesize =0}$ {\footnotesize %
and }${\footnotesize (N,T)=(100,300)}${\footnotesize , the smoothed
algorithm with }${\footnotesize \rho =}${\footnotesize $\zeta $}$%
{\footnotesize =}${\footnotesize $\xi $}${\footnotesize =0.5}$ 
{\footnotesize and }${\footnotesize (N,T)=(100,300)}${\footnotesize , and
the smoothed algorithm with }${\footnotesize \rho =}${\footnotesize $\zeta $}%
${\footnotesize =}${\footnotesize $\xi $}${\footnotesize =0}$ {\footnotesize %
and }${\footnotesize (N,T)=(200,300)}${\footnotesize , respectively.
Subfigures in the first to the third column correspond to the estimated
loadings for regime 1, the estimated loadings for regime 2 and the estimated
factors, respectively. The curve overlaid on the histograms is the standard
normal density function.}%
\end{figure}%

\begin{table}[tbp]%
\begin{center}%
\caption{Average $R^2$ of the Estimated Loading Space, Average $R^2$ of the
Estimated Factor Space, and Average Absolute Error of the Estimated
Transition Probabilities}%
\begin{tabular}{lllllll}
\hline\hline
& $R_{l1}^{2}$ & $R_{l2}^{2}$ & $R_{f}^{2}$ & $R_{Hf}^{2}$ & $\tilde{Q}_{11}$
& $\tilde{Q}_{22}$ \\ \cline{1-2}\cline{1-3}\cline{3-7}
\multicolumn{7}{c}{Smoothed with $(\rho ,\zeta ,\xi )=(0,0,0)$ and $%
(N,T)=(100,300)$} \\ \hline
Pattern 1 & 0.996 & 0.9762 & 0.7337 & 0.9889 & 0.0028 & 0.013 \\ 
Pattern 2 & 0.9931 & 0.9932 & 0.5155 & 0.9896 & N.A. & N.A. \\ 
Pattern 3 & 0.9949 & 0.9895 & 0.541 & 0.9894 & N.A. & N.A. \\ 
Pattern 4 & 0.9955 & 0.9854 & 0.6256 & 0.9892 & 0.0216 & 0.0378 \\ \hline
\multicolumn{7}{c}{Unsmoothed with $(\rho ,\zeta ,\xi )=(0,0,0)$ and $%
(N,T)=(100,300)$} \\ \hline
Pattern 1 & 0.9959 & 0.9678 & 0.6786 & 0.9782 & N.A. & N.A. \\ 
Pattern 2 & 0.9931 & 0.9932 & 0.4855 & 0.9885 & N.A. & N.A. \\ 
Pattern 3 & 0.9949 & 0.9892 & 0.5157 & 0.988 & N.A. & N.A. \\ 
Pattern 4 & 0.9955 & 0.9853 & 0.6127 & 0.9875 & N.A. & N.A. \\ \hline
\multicolumn{7}{c}{Smoothed with $(\rho ,\zeta ,\xi )=(0.5,0.5,0.5)$ and $%
(N,T)=(100,300)$} \\ \hline
Pattern 1 & 0.9933 & 0.9631 & 0.7255 & 0.9849 & 0.0053 & 0.0137 \\ 
Pattern 2 & 0.9928 & 0.9929 & 0.4797 & 0.9891 & N.A. & N.A. \\ 
Pattern 3 & 0.9915 & 0.9827 & 0.5458 & 0.9889 & N.A. & N.A. \\ 
Pattern 4 & 0.9927 & 0.9782 & 0.6285 & 0.9886 & 0.0239 & 0.0328 \\ \hline
\multicolumn{7}{c}{Smoothed with $(\rho ,\zeta ,\xi )=(0,0,0)$ and $%
(N,T)=(200,300)$} \\ \hline
Pattern 1 & 0.996 & 0.9756 & 0.7408 & 0.9936 & 0.0017 & 0.0151 \\ 
Pattern 2 & 0.9933 & 0.9933 & 0.5183 & 0.9949 & N.A. & N.A. \\ 
Pattern 3 & 0.995 & 0.9898 & 0.5611 & 0.9949 & N.A. & N.A. \\ 
Pattern 4 & 0.9956 & 0.9856 & 0.6227 & 0.9947 & 0.019 & 0.024 \\ \hline\hline
\end{tabular}%
\label{Table}%
\end{center}%
{\footnotesize Notes: The column under }$R_{l1}^{2}${\footnotesize \ shows
the average }$R^{2}${\footnotesize \ of the estimated loadings of regime 1
projecting on the true loadings of regime 1. The column under }$R_{l2}^{2}$%
{\footnotesize \ shows the average }$R^{2}${\footnotesize \ of the estimated
loadings of regime 2 projecting on the true loadings of regime 2. The column
under }$R_{f}^{2}${\footnotesize \ shows the average }$R^{2}${\footnotesize %
\ of the estimated factors projecting on the true factors. The column under }%
$R_{Hf}^{2}${\footnotesize \ shows the average }$R^{2}${\footnotesize \ of
the estimated factors projecting on the factors rotated by the regime
dependent rotation matrix }$H_{z_{t}}${\footnotesize . "N.A." means not
available.}%
\end{table}%

\begin{table}[tbp]%
\begin{center}%
\caption{Out of Sample Turning Points Detection}%
\begin{tabular}{llllll}
\hline\hline
& Recession & Expansion & Recession & Expansion & Recession \\ 
& 02/1980 & 08/1980 & 08/1981 & 12/1982 & 08/1990 \\ 
NBER BCDC & 4 & 11 & 5 & 7 & 9 \\ 
Chauvet Piger & 6 & 5 & 7 & 6 & 7 \\ 
This paper & 3 & 2 & 3 & 7 & N.A. \\ \hline
& Expansion & Recession & Expansion & Recession & Expansion \\ 
& 04/1991 & 04/2001 & 12/2001 & 01/2008 & 07/2009 \\ 
NBER BCDC & 21 & 8 & 20 & 11 & 15 \\ 
Chauvet Piger & 6 & 10 & 7 & 13 & 7 \\ 
This paper & 1 & 8 & 7 & 11 & 10 \\ \hline\hline
\end{tabular}%
\label{Table 2}%
\end{center}%
{\footnotesize Notes: mm/yyyy in the second and the seventh row indicate the
starting month of each recession and expansion. The row corresponds to "NBER
BCDC", "Chauvet Piger" and "This paper" shows the number of months it takes
the NBER BCDC, Chauvet and Piger (2008) and this paper to detect each
recession and expansion, respectively. "N.A." means not available.}%
\end{table}%

\pagebreak 
\setcounter{page}{1}%

\begin{center}
\appendix\textbf{APPENDIX}
\end{center}

\renewcommand{\theenumi}{\arabic{enumi}}

\section{Details for Theorem \protect\ref{consis}}

\begin{lem}
\label{E norm}Under Assumption \ref{error}(2) and \ref{error}(4), $%
\left\Vert E\right\Vert =O_{p}(N^{\frac{1}{4}}T^{\frac{1}{2}}+N^{\frac{1}{2}%
}T^{\frac{1}{4}})$.
\end{lem}

\begin{proof}
We shall show $\mathbb{E}\left\Vert E\right\Vert ^{4}=O(NT^{2}+N^{2}T)$.
First note that 
\begin{equation*}
\left\Vert E\right\Vert ^{4}=\left\Vert E^{\prime }E\right\Vert ^{2}\leq
\left\Vert E^{\prime }E\right\Vert
_{F}^{2}=\sum\nolimits_{i=1}^{N}\sum\nolimits_{k=1}^{N}(\sum%
\nolimits_{t=1}^{T}e_{it}e_{kt})^{2}\text{.}
\end{equation*}%
$\mathbb{E}(\sum\nolimits_{t=1}^{T}e_{it}e_{kt})^{2}$ is not larger than the
sum of $2\mathbb{E}(\sum\nolimits_{t=1}^{T}e_{it}e_{kt}-\sum%
\nolimits_{t=1}^{T}\mathbb{E(}e_{it}e_{kt}))^{2}$ and $2(\sum%
\nolimits_{t=1}^{T}\mathbb{E(}e_{it}e_{kt}))^{2}$. The sum of the former
over $i$ and $k$ is not larger than $N^{2}TM$ since by Assumption \ref{error}%
(4), $\mathbb{E}(\left\Vert \frac{1}{\sqrt{T}}\sum%
\nolimits_{t=1}^{T}(e_{it}e_{kt}-\mathbb{E(}e_{it}e_{kt}))\right\Vert
^{2})\leq M$. The sum of the latter over $i$ and $k$ is not larger than $%
NT^{2}M$ under Assumption \ref{error}(2).
\end{proof}

\begin{description}
\item[Proof of Theorem \protect\ref{consis}] 
\end{description}

\begin{proof}
Step (1): Since $z_{t}$ follows a Markov process,%
\begin{equation*}
l(\Lambda ,\sigma ^{2},Q,\phi )=\log
[\sum\nolimits_{z_{T}=1}^{J^{0}}...\sum\nolimits_{z_{1}=1}^{J^{0}}\prod%
\nolimits_{t=1}^{T}L(x_{t}\left\vert z_{t};\Lambda ,\sigma ^{2}\right. )\Pr
(z_{1}\left\vert \phi \right. )\prod\nolimits_{t=2}^{T}\Pr (z_{t}\left\vert
z_{t-1};Q\right. )].
\end{equation*}%
For $(\tilde{\Lambda},\tilde{\sigma}^{2},Q,\phi )$, let $m_{t}=\arg
\max\nolimits_{j}\{(2\pi )^{-\frac{N}{2}}\left\vert \tilde{\Lambda}_{j}%
\tilde{\Lambda}_{j}^{\prime }+\tilde{\sigma}^{2}I_{N}\right\vert ^{-\frac{1}{%
2}}e^{-\frac{1}{2}x_{t}^{\prime }(\tilde{\Lambda}_{j}\tilde{\Lambda}%
_{j}^{\prime }+\tilde{\sigma}^{2}I_{N})^{-1}x_{t}}\}$, i.e, $%
L(x_{t}\left\vert z_{t}=j;\tilde{\Lambda},\tilde{\sigma}^{2}\right. )$ takes
maximum when $j=m_{t}$. Since $\sum\nolimits_{z_{t}=1}^{J^{0}}\Pr
(z_{t}\left\vert z_{t-1};Q\right. )=1$ for any $z_{t-1}$, $%
\sum\nolimits_{z_{t}=1}^{J^{0}}L(x_{t}\left\vert z_{t};\tilde{\Lambda},%
\tilde{\sigma}^{2}\right. )\Pr (z_{t}\left\vert z_{t-1};Q\right. )\leq
L(x_{t}\left\vert z_{t}=m_{t};\tilde{\Lambda},\tilde{\sigma}^{2}\right. )$,
thus%
\begin{eqnarray}
&&l(\tilde{\Lambda},\tilde{\sigma}^{2},Q,\phi )  \notag \\
&=&\log
\{\sum\nolimits_{z_{T-1}=1}^{J^{0}}...\sum\nolimits_{z_{1}=1}^{J^{0}}\prod%
\nolimits_{t=1}^{T-1}L(x_{t}\left\vert z_{t};\tilde{\Lambda},\tilde{\sigma}%
^{2}\right. )\Pr (z_{1}\left\vert \phi \right.
)\prod\nolimits_{t=2}^{T-1}\Pr (z_{t}\left\vert z_{t-1};Q\right. )  \notag \\
&&[\sum\nolimits_{z_{T}=1}^{J^{0}}L(x_{T}\left\vert z_{T};\tilde{\Lambda},%
\tilde{\sigma}^{2}\right. )\Pr (z_{T}\left\vert z_{T-1};Q\right. )]\}  \notag
\\
&\leq &\log
\{\sum\nolimits_{z_{T-1}=1}^{J^{0}}...\sum\nolimits_{z_{1}=1}^{J^{0}}\prod%
\nolimits_{t=1}^{T-1}L(x_{t}\left\vert z_{t};\tilde{\Lambda},\tilde{\sigma}%
^{2}\right. )\Pr (z_{1}\left\vert \phi \right.
)\prod\nolimits_{t=2}^{T-1}\Pr (z_{t}\left\vert z_{t-1};Q\right. )  \notag \\
&&L(x_{T}\left\vert z_{T}=m_{T};\tilde{\Lambda},\tilde{\sigma}^{2}\right. )\}
\notag \\
&\leq &...\leq \sum\nolimits_{t=1}^{T}\log L(x_{t}\left\vert z_{t}=m_{t};%
\tilde{\Lambda},\tilde{\sigma}^{2}\right. ).  \label{j}
\end{eqnarray}%
It follows that%
\begin{eqnarray}
l(\tilde{\Lambda},\tilde{\sigma}^{2},Q,\phi ) &\leq
&\sum\nolimits_{t=1}^{T}\log [(2\pi )^{-\frac{N}{2}}\left\vert \tilde{\Lambda%
}_{m_{t}}\tilde{\Lambda}_{m_{t}}^{\prime }+\tilde{\sigma}^{2}I_{N}\right%
\vert ^{-\frac{1}{2}}e^{-\frac{1}{2}x_{t}^{\prime }(\tilde{\Lambda}_{m_{t}}%
\tilde{\Lambda}_{m_{t}}^{\prime }+\tilde{\sigma}^{2}I_{N})^{-1}x_{t}}] 
\notag \\
&=&-\frac{NT}{2}\log 2\pi -\frac{1}{2}\sum\nolimits_{t=1}^{T}\log \left\vert 
\tilde{\Lambda}_{m_{t}}\tilde{\Lambda}_{m_{t}}^{\prime }+\tilde{\sigma}%
^{2}I_{N}\right\vert  \notag \\
&&-\frac{1}{2}\sum\nolimits_{t=1}^{T}x_{t}^{\prime }(\tilde{\Lambda}_{m_{t}}%
\tilde{\Lambda}_{m_{t}}^{\prime }+\tilde{\sigma}^{2}I_{N})^{-1}x_{t}.
\label{b}
\end{eqnarray}%
Consider the last term on the right hand side of equation (\ref{b}). By
Woodbury identity, $(\tilde{\Lambda}_{m_{t}}\tilde{\Lambda}_{m_{t}}^{\prime
}+\tilde{\sigma}^{2}I_{N})^{-1}=\tilde{\sigma}^{-2}I_{N}-\tilde{\sigma}^{-2}%
\tilde{\Lambda}_{m_{t}}(\tilde{\sigma}^{2}I_{r_{m_{t}}^{0}}+\tilde{\Lambda}%
_{m_{t}}^{\prime }\tilde{\Lambda}_{m_{t}})^{-1}\tilde{\Lambda}%
_{m_{t}}^{\prime }$. Thus%
\begin{eqnarray*}
&&\sum\nolimits_{t=1}^{T}x_{t}^{\prime }(\tilde{\Lambda}_{m_{t}}\tilde{%
\Lambda}_{m_{t}}^{\prime }+\tilde{\sigma}^{2}I_{N})^{-1}x_{t} \\
&=&\tilde{\sigma}^{-2}\sum\nolimits_{t=1}^{T}x_{t}^{\prime }x_{t}-\tilde{%
\sigma}^{-2}\sum\nolimits_{t=1}^{T}x_{t}^{\prime }\tilde{\Lambda}_{m_{t}}(%
\tilde{\sigma}^{2}I_{r_{m_{t}}^{0}}+\tilde{\Lambda}_{m_{t}}^{\prime }\tilde{%
\Lambda}_{m_{t}})^{-1}\tilde{\Lambda}_{m_{t}}^{\prime }x_{t}.
\end{eqnarray*}%
Since $(\tilde{\Lambda}_{m_{t}}^{\prime }\tilde{\Lambda}_{m_{t}})^{-1}-(%
\tilde{\sigma}^{2}I_{r_{m_{t}}^{0}}+\tilde{\Lambda}_{m_{t}}^{\prime }\tilde{%
\Lambda}_{m_{t}})^{-1}=(\tilde{\sigma}^{2}I_{r_{m_{t}}^{0}}+\tilde{\Lambda}%
_{m_{t}}^{\prime }\tilde{\Lambda}_{m_{t}})^{-1}\tilde{\sigma}^{2}(\tilde{%
\Lambda}_{m_{t}}^{\prime }\tilde{\Lambda}_{m_{t}})^{-1}$,%
\begin{eqnarray}
&&\sum\nolimits_{t=1}^{T}x_{t}^{\prime }(\tilde{\Lambda}_{m_{t}}\tilde{%
\Lambda}_{m_{t}}^{\prime }+\tilde{\sigma}^{2}I_{N})^{-1}x_{t}  \notag \\
&=&\tilde{\sigma}^{-2}\sum\nolimits_{t=1}^{T}\left\Vert M_{\tilde{\Lambda}%
_{m_{t}}}x_{t}\right\Vert ^{2}+\sum\nolimits_{t=1}^{T}x_{t}^{\prime }\tilde{%
\Lambda}_{m_{t}}(\tilde{\sigma}^{2}I_{r_{m_{t}}^{0}}+\tilde{\Lambda}%
_{m_{t}}^{\prime }\tilde{\Lambda}_{m_{t}})^{-1}(\tilde{\Lambda}%
_{m_{t}}^{\prime }\tilde{\Lambda}_{m_{t}})^{-1}\tilde{\Lambda}%
_{m_{t}}^{\prime }x_{t}.  \label{c}
\end{eqnarray}

Step (2): Now consider $l(\Lambda ^{0},\tilde{\sigma}^{2},Q,\phi )$. Since $%
\Pr (z_{t}\left\vert z_{t-1};Q\right. )\geq \min_{j,k}Q_{jk}$,%
\begin{equation*}
\sum\nolimits_{z_{t}=1}^{J^{0}}L(x_{t}\left\vert z_{t};\Lambda ^{0},\tilde{%
\sigma}^{2}\right. )\Pr (z_{t}\left\vert z_{t-1};Q\right. )\geq
L(x_{t}\left\vert z_{t};\Lambda ^{0},\tilde{\sigma}^{2}\right.
)\min_{j,k}Q_{jk}.
\end{equation*}%
Note that $z_{t}$ denotes the true state on the right hand side. The left
hand side has $J^{0}$ terms in the summation, and the inequality follows
from throwing away all the other $J^{0}-1$ terms. Thus similar to inequality
(\ref{j}), 
\begin{eqnarray}
l(\Lambda ^{0},\tilde{\sigma}^{2},Q,\phi ) &\geq
&\sum\nolimits_{t=1}^{T}\log L(x_{t}\left\vert z_{t};\Lambda ^{0},\tilde{%
\sigma}^{2}\right. )\min_{j,k}Q_{jk}  \notag \\
&=&T\log \min_{j,k}Q_{jk}-\frac{NT}{2}\log 2\pi -\frac{1}{2}%
\sum\nolimits_{t=1}^{T}\log \left\vert \Lambda _{z_{t}}^{0}\Lambda
_{z_{t}}^{0\prime }+\tilde{\sigma}^{2}I_{N}\right\vert  \notag \\
&&-\frac{1}{2}\sum\nolimits_{t=1}^{T}x_{t}^{\prime }(\Lambda
_{z_{t}}^{0}\Lambda _{z_{t}}^{0\prime }+\tilde{\sigma}^{2}I_{N})^{-1}x_{t},
\label{d}
\end{eqnarray}%
and similar to equation (\ref{c}), 
\begin{eqnarray}
&&\sum\nolimits_{t=1}^{T}x_{t}^{\prime }(\Lambda _{z_{t}}^{0}\Lambda
_{z_{t}}^{0\prime }+\tilde{\sigma}^{2}I_{N})^{-1}x_{t}  \notag \\
&=&\tilde{\sigma}^{-2}\sum\nolimits_{t=1}^{T}\left\Vert M_{\Lambda
_{z_{t}}^{0}}x_{t}\right\Vert ^{2}+\sum\nolimits_{t=1}^{T}x_{t}^{\prime
}\Lambda _{z_{t}}^{0}(\tilde{\sigma}^{2}I_{r_{z_{t}}^{0}}+\Lambda
_{z_{t}}^{0\prime }\Lambda _{z_{t}}^{0})^{-1}(\Lambda _{z_{t}}^{0\prime
}\Lambda _{z_{t}}^{0})^{-1}\Lambda _{z_{t}}^{0\prime }x_{t}.  \label{e}
\end{eqnarray}

Step (3): $l(\tilde{\Lambda},\tilde{\sigma}^{2},Q,\phi )-l(\Lambda ^{0},%
\tilde{\sigma}^{2},Q,\phi )\geq 0$. Thus by equations (\ref{b})-(\ref{e}),
we have 
\begin{eqnarray}
&&\frac{1}{2}[\tilde{\sigma}^{-2}\sum\nolimits_{t=1}^{T}\left\Vert M_{\tilde{%
\Lambda}_{m_{t}}}x_{t}\right\Vert ^{2}-\tilde{\sigma}^{-2}\sum%
\nolimits_{t=1}^{T}\left\Vert M_{\Lambda _{z_{t}}^{0}}x_{t}\right\Vert ^{2}]
\notag \\
&\leq &-T\log \min_{j,k}Q_{jk}-\frac{1}{2}\sum\nolimits_{t=1}^{T}\log \frac{%
\left\vert \tilde{\Lambda}_{m_{t}}\tilde{\Lambda}_{m_{t}}^{\prime }+\tilde{%
\sigma}^{2}I_{N}\right\vert }{\left\vert \Lambda _{z_{t}}^{0}\Lambda
_{z_{t}}^{0\prime }+\tilde{\sigma}^{2}I_{N}\right\vert }  \notag \\
&&-\frac{1}{2}\sum\nolimits_{t=1}^{T}x_{t}^{\prime }\tilde{\Lambda}_{m_{t}}(%
\tilde{\sigma}^{2}I_{r_{m_{t}}^{0}}+\tilde{\Lambda}_{m_{t}}^{\prime }\tilde{%
\Lambda}_{m_{t}})^{-1}(\tilde{\Lambda}_{m_{t}}^{\prime }\tilde{\Lambda}%
_{m_{t}})^{-1}\tilde{\Lambda}_{m_{t}}^{\prime }x_{t}  \notag \\
&&+\frac{1}{2}\sum\nolimits_{t=1}^{T}x_{t}^{\prime }\Lambda _{z_{t}}^{0}(%
\tilde{\sigma}^{2}I_{r_{z_{t}}^{0}}+\Lambda _{z_{t}}^{0\prime }\Lambda
_{z_{t}}^{0})^{-1}(\Lambda _{z_{t}}^{0\prime }\Lambda
_{z_{t}}^{0})^{-1}\Lambda _{z_{t}}^{0\prime }x_{t}.  \label{f}
\end{eqnarray}

(3.1) The first term on the right hand side is $O(T)$ since $%
\min_{j,k}Q_{jk}>0$.

(3.2) The second term on the right hand side equals $-\frac{1}{2}%
\sum\nolimits_{t=1}^{T}\log \left\vert \frac{1}{\tilde{\sigma}^{2}}\tilde{%
\Lambda}_{m_{t}}^{\prime }\tilde{\Lambda}_{m_{t}}+I_{r_{m_{t}}^{0}}\right%
\vert +\frac{1}{2}\sum\nolimits_{t=1}^{T}\log \left\vert \frac{1}{\tilde{%
\sigma}^{2}}\Lambda _{z_{t}}^{0\prime }\Lambda
_{z_{t}}^{0}+I_{r_{z_{t}}^{0}}\right\vert $ since%
\begin{eqnarray}
\left\vert \Lambda _{z_{t}}^{0}\Lambda _{z_{t}}^{0\prime }+\tilde{\sigma}%
^{2}I_{N}\right\vert &=&\tilde{\sigma}^{2N}\left\vert \frac{1}{\tilde{\sigma}%
^{2}}\Lambda _{z_{t}}^{0}\Lambda _{z_{t}}^{0\prime }+I_{N}\right\vert =%
\tilde{\sigma}^{2N}\left\vert \frac{1}{\tilde{\sigma}^{2}}\Lambda
_{z_{t}}^{0\prime }\Lambda _{z_{t}}^{0}+I_{r_{z_{t}}^{0}}\right\vert ,
\label{s} \\
\text{and }\left\vert \tilde{\Lambda}_{m_{t}}\tilde{\Lambda}_{m_{t}}^{\prime
}+\tilde{\sigma}^{2}I_{N}\right\vert &=&\tilde{\sigma}^{2N}\left\vert \frac{1%
}{\tilde{\sigma}^{2}}\tilde{\Lambda}_{m_{t}}^{\prime }\tilde{\Lambda}%
_{m_{t}}+I_{r_{m_{t}}^{0}}\right\vert .  \label{t}
\end{eqnarray}%
$-\frac{1}{2}\sum\nolimits_{t=1}^{T}\log \left\vert \frac{1}{\tilde{\sigma}%
^{2}}\tilde{\Lambda}_{m_{t}}^{\prime }\tilde{\Lambda}%
_{m_{t}}+I_{r_{m_{t}}^{0}}\right\vert $ is negative, thus inequality (\ref{f}%
) still holds when this term is thrown away. By Assumption \ref{loadings}%
(1), $\left\vert \frac{1}{\tilde{\sigma}^{2}}\Lambda _{z_{t}}^{0\prime
}\Lambda _{z_{t}}^{0}+I_{r_{z_{t}}^{0}}\right\vert $ $\leq c(\frac{N}{\tilde{%
\sigma}^{2}})^{r_{z_{t}}^{0}}$ for some $c>0$, thus $\frac{1}{2}%
\sum\nolimits_{t=1}^{T}\log \left\vert \frac{1}{\tilde{\sigma}^{2}}\Lambda
_{z_{t}}^{0\prime }\Lambda _{z_{t}}^{0}+I_{r_{z_{t}}^{0}}\right\vert $ is $%
O_{p}(T\log N)$.

(3.3) The third term on the right hand side is negative, thus inequality (%
\ref{f}) still holds when this term is thrown away.

(3.4) The fourth term is bounded by $\frac{1}{2}\sum\nolimits_{t=1}^{T}\left%
\Vert x_{t}\right\Vert ^{2}\left\Vert (\tilde{\sigma}^{2}I_{r_{z_{t}}^{0}}+%
\Lambda _{z_{t}}^{0\prime }\Lambda _{z_{t}}^{0})^{-1}\right\Vert $ since $%
\left\Vert (\Lambda _{z_{t}}^{0\prime }\Lambda _{z_{t}}^{0})^{-\frac{1}{2}%
}\Lambda _{z_{t}}^{0\prime }x_{t}\right\Vert =\left\Vert P_{\Lambda
_{z_{t}}^{0}}x_{t}\right\Vert \leq \left\Vert x_{t}\right\Vert $ and $%
(\Lambda _{z_{t}}^{0\prime }\Lambda _{z_{t}}^{0})^{\frac{1}{2}}(\tilde{\sigma%
}^{2}I_{r_{z_{t}}^{0}}+\Lambda _{z_{t}}^{0\prime }\Lambda
_{z_{t}}^{0})^{-1}(\Lambda _{z_{t}}^{0\prime }\Lambda _{z_{t}}^{0})^{-\frac{1%
}{2}}=(\tilde{\sigma}^{2}I_{r_{z_{t}}^{0}}+\Lambda _{z_{t}}^{0\prime
}\Lambda _{z_{t}}^{0})^{-1}$. The latter is because $\tilde{\sigma}%
^{2}I_{r_{z_{t}}^{0}}+\Lambda _{z_{t}}^{0\prime }\Lambda _{z_{t}}^{0}$ and $%
\Lambda _{z_{t}}^{0\prime }\Lambda _{z_{t}}^{0}$ have the same eigenvectors.
By Assumption \ref{loadings}(1), $\left\Vert (\tilde{\sigma}%
^{2}I_{r_{z_{t}}^{0}}+\Lambda _{z_{t}}^{0\prime }\Lambda
_{z_{t}}^{0})^{-1}\right\Vert \leq \sup_{j}\left\Vert (\Lambda _{j}^{0\prime
}\Lambda _{j}^{0})^{-1}\right\Vert =O_{p}(\frac{1}{N})$. By Assumptions \ref%
{factors}(2), \ref{loadings}(1) and \ref{error}(1), $\sum\nolimits_{t=1}^{T}%
\left\Vert x_{t}\right\Vert ^{2}=O_{p}(NT)$. Thus the fourth term is $%
O_{p}(T)$.

(3.5) Now consider the left hand side of expression (\ref{f}). Since $%
x_{t}=\Lambda _{z_{t}}^{0}f_{t}^{0}+e_{t}$ and $M_{\Lambda
_{z_{t}}^{0}}\Lambda _{z_{t}}^{0}f_{t}^{0}=0$, it is easy to verify that the
left hand side equals%
\begin{equation}
\frac{1}{2}\tilde{\sigma}^{-2}[\sum\nolimits_{t=1}^{T}\left\Vert M_{\tilde{%
\Lambda}_{m_{t}}}\Lambda _{z_{t}}^{0}f_{t}^{0}\right\Vert
^{2}+2\sum\nolimits_{t=1}^{T}e_{t}^{\prime }M_{\tilde{\Lambda}%
_{m_{t}}}\Lambda _{z_{t}}^{0}f_{t}^{0}+\sum\nolimits_{t=1}^{T}\left\Vert
P_{\Lambda _{z_{t}}^{0}}e_{t}\right\Vert
^{2}-\sum\nolimits_{t=1}^{T}\left\Vert P_{\tilde{\Lambda}_{m_{t}}}e_{t}%
\right\Vert ^{2}].  \label{h}
\end{equation}%
For the fourth term of expression (\ref{h}), we have%
\begin{eqnarray}
\sum\nolimits_{t=1}^{T}\left\Vert P_{\tilde{\Lambda}_{m_{t}}}e_{t}\right%
\Vert ^{2} &\leq
&\sum\nolimits_{j=1}^{J^{0}}\sum\nolimits_{t=1}^{T}\left\Vert P_{\tilde{%
\Lambda}_{j}}e_{t}\right\Vert
^{2}=\sum\nolimits_{j=1}^{J^{0}}\sum\nolimits_{t=1}^{T}e_{t}^{\prime }\tilde{%
\Lambda}_{j}(\tilde{\Lambda}_{j}^{\prime }\tilde{\Lambda}_{j})^{-1}\tilde{%
\Lambda}_{j}^{\prime }e_{t}  \notag \\
&=&\sum\nolimits_{j=1}^{J^{0}}tr[(\tilde{\Lambda}_{j}^{\prime }\tilde{\Lambda%
}_{j})^{-\frac{1}{2}}\tilde{\Lambda}_{j}^{\prime
}(\sum\nolimits_{t=1}^{T}e_{t}e_{t}^{\prime })\tilde{\Lambda}_{j}(\tilde{%
\Lambda}_{j}^{\prime }\tilde{\Lambda}_{j})^{-\frac{1}{2}}]  \notag \\
&\leq &\sum\nolimits_{j=1}^{J^{0}}r_{j}^{0}\rho _{\max
}(\sum\nolimits_{t=1}^{T}e_{t}e_{t}^{\prime
})=\sum\nolimits_{j=1}^{J^{0}}r_{j}^{0}\left\Vert E^{\prime }E\right\Vert 
\notag \\
&=&O_{p}(N^{\frac{1}{2}}T+NT^{\frac{1}{2}}).  \label{g}
\end{eqnarray}%
The last equality follows from Lemma \ref{E norm}. Similarly, the third term
of expression (\ref{h}) is $O_{p}(N^{\frac{1}{2}}T+NT^{\frac{1}{2}})$. The
second term of expression (\ref{h}) equals $2\sum\nolimits_{t=1}^{T}e_{t}^{%
\prime }\Lambda _{z_{t}}^{0}f_{t}^{0}-2\sum\nolimits_{t=1}^{T}e_{t}^{\prime
}P_{\tilde{\Lambda}_{m_{t}}}\Lambda _{z_{t}}^{0}f_{t}^{0}$. Since $\mathbb{E}%
(\left\Vert \frac{1}{\sqrt{N}}\sum\nolimits_{i=1}^{N}\lambda
_{ji}^{0}e_{it}\right\Vert ^{2})\leq M$ for all $j$ and $t$ by Assumptions %
\ref{loadings}(1), \ref{error}(1) and \ref{error}(2), and $%
\sum\nolimits_{t=1}^{T}\left\Vert f_{t}^{0}\right\Vert ^{2}=O_{p}(T)$ by
Assumption \ref{factors}(2), we have%
\begin{equation*}
\left\Vert \sum\nolimits_{t=1}^{T}e_{t}^{\prime }\Lambda
_{z_{t}}^{0}f_{t}^{0}\right\Vert \leq (\sum\nolimits_{t=1}^{T}\left\Vert
e_{t}^{\prime }\Lambda _{z_{t}}^{0}\right\Vert ^{2})^{\frac{1}{2}%
}(\sum\nolimits_{t=1}^{T}\left\Vert f_{t}^{0}\right\Vert ^{2})^{\frac{1}{2}%
}=O_{p}(N^{\frac{1}{2}}T).
\end{equation*}%
By expression (\ref{g}), Assumption \ref{factors}(2) and Assumption \ref%
{loadings}(1), we have%
\begin{equation*}
\left\Vert \sum\nolimits_{t=1}^{T}e_{t}^{\prime }P_{\tilde{\Lambda}%
_{m_{t}}}\Lambda _{z_{t}}^{0}f_{t}^{0}\right\Vert \leq
(\sum\nolimits_{t=1}^{T}\left\Vert P_{\tilde{\Lambda}_{m_{t}}}e_{t}\right%
\Vert ^{2}\sum\nolimits_{t=1}^{T}\left\Vert f_{t}^{0}\right\Vert ^{2})^{%
\frac{1}{2}}\sup_{j}\left\Vert \Lambda _{j}^{0}\right\Vert =O_{p}(N^{\frac{3%
}{4}}T+NT^{\frac{3}{4}}).
\end{equation*}%
Thus the second term of expression (\ref{h}) is $O_{p}(N^{\frac{3}{4}}T+NT^{%
\frac{3}{4}})$.

(3.6) Move the second to the fourth term of expression (\ref{h}) to the
right hand side of equation (\ref{f}), and take the results (3.1)-(3.5)
together, we have%
\begin{eqnarray*}
0 &\leq &\frac{1}{2}\tilde{\sigma}^{-2}\sum\nolimits_{t=1}^{T}\left\Vert M_{%
\tilde{\Lambda}_{m_{t}}}\Lambda _{z_{t}}^{0}f_{t}^{0}\right\Vert ^{2}\leq
O(T)+O_{p}(T\log N) \\
&&+O(T)+O_{p}(N^{\frac{1}{2}}T+NT^{\frac{1}{2}})+O_{p}(N^{\frac{3}{4}}T+NT^{%
\frac{3}{4}}).
\end{eqnarray*}%
Thus $\sum\nolimits_{t=1}^{T}\left\Vert M_{\tilde{\Lambda}_{m_{t}}}\Lambda
_{z_{t}}^{0}f_{t}^{0}\right\Vert ^{2}$ is $O_{p}(N^{\frac{3}{4}}T+NT^{\frac{3%
}{4}})$. In the summation, there are $q_{1}^{0}T$ terms\footnote{%
Rigorously speaking, there are $\sum\nolimits_{t=1}^{T}1_{z_{t}=1}$ terms,
but $\frac{1}{T}\sum\nolimits_{t=1}^{T}1_{z_{t}=1}\overset{p}{\rightarrow }%
q_{1}^{0}$ as $T\rightarrow \infty $.} with $\Lambda _{z_{t}}^{0}=\Lambda
_{1}^{0}$, since $q_{1}^{0}$ is the unconditional probability of $z_{t}=1$.
For each $t$ with $z_{t}=1$, $\Lambda _{1}^{0}f_{t}^{0}$ are projected on
one of $\tilde{\Lambda}_{j},j=1,...,J^{0}$, thus there exists one certain $%
\tilde{\Lambda}_{j}$ such that $\Lambda _{1}^{0}f_{t}^{0}$ is projected on $%
\tilde{\Lambda}_{j}$ at least $\frac{q_{1}^{0}T}{J^{0}}$ times. Define this $%
\tilde{\Lambda}_{j}$ as $\tilde{\Lambda}_{1}$, then $\sum%
\nolimits_{t=1}^{T}1_{m_{t}=1}1_{z_{t}=1}\geq \frac{q_{1}^{0}T}{J^{0}}$.
Thus by Assumption \ref{factors}(1),%
\begin{equation*}
\rho _{\min }(\frac{1}{\sum\nolimits_{t=1}^{T}1_{m_{t}=1}1_{z_{t}=1}}%
\sum\nolimits_{t=1}^{T}f_{t}^{0}f_{t}^{0\prime }1_{m_{t}=1}1_{z_{t}=1})\geq c
\end{equation*}
for some $c>0$ w.p.a.1. Since $\left\Vert M_{\tilde{\Lambda}_{m_{t}}}\Lambda
_{z_{t}}^{0}f_{t}^{0}\right\Vert $ is positive for any $z_{t}$ and $m_{t}$,
we have%
\begin{eqnarray*}
O_{p}(N^{\frac{3}{4}}T+NT^{\frac{3}{4}})
&=&\sum\nolimits_{t=1}^{T}\left\Vert M_{\tilde{\Lambda}_{m_{t}}}\Lambda
_{z_{t}}^{0}f_{t}^{0}\right\Vert ^{2}\geq \sum\nolimits_{t=1}^{T}\left\Vert
M_{\tilde{\Lambda}_{1}}\Lambda _{1}^{0}f_{t}^{0}\right\Vert
^{2}1_{m_{t}=1}1_{z_{t}=1} \\
&=&tr(\Lambda _{1}^{0\prime }M_{\tilde{\Lambda}_{1}}\Lambda
_{1}^{0}\sum\nolimits_{t=1}^{T}f_{t}^{0}f_{t}^{0\prime
}1_{m_{t}=1}1_{z_{t}=1}) \\
&\geq &tr(\Lambda _{1}^{0\prime }M_{\tilde{\Lambda}_{1}}\Lambda
_{1}^{0})\rho _{\min }(\sum\nolimits_{t=1}^{T}f_{t}^{0}f_{t}^{0\prime
}1_{m_{t}=1}1_{z_{t}=1}) \\
&\geq &tr(\Lambda _{1}^{0\prime }M_{\tilde{\Lambda}_{1}}\Lambda _{1}^{0})%
\frac{Tq_{1}^{0}}{J^{0}}c\text{ w.p.a.1, }
\end{eqnarray*}%
thus $\frac{1}{N}\left\Vert M_{\tilde{\Lambda}_{1}}\Lambda
_{1}^{0}\right\Vert _{F}^{2}=\frac{1}{N}tr(\Lambda _{1}^{0\prime }M_{\tilde{%
\Lambda}_{1}}\Lambda _{1}^{0})=O_{p}(\frac{1}{\sqrt{\delta _{NT}}})$.
Similarly, for $j=2,...,J^{0}$, we also have $\frac{1}{N}\left\Vert M_{%
\tilde{\Lambda}_{j}}\Lambda _{j}^{0}\right\Vert _{F}^{2}=O_{p}(\frac{1}{%
\sqrt{\delta _{NT}}})$.
\end{proof}

\section{Details for Theorem \protect\ref{ptjT}}

\begin{lem}
\label{lamhat norm}Under the assumptions of Theorem \ref{consis}, $\frac{1}{%
\tilde{\sigma}^{2}+\tilde{\Lambda}_{jl}^{\prime }\tilde{\Lambda}_{jl}}=O_{p}(%
\frac{1}{\sqrt{\delta _{NT}}})$ for each $j=1,...,J^{0}$ and each $%
l=1,...,r_{j}^{0}$, where $\tilde{\Lambda}_{jl}$ denotes the $l$-th column
of $\tilde{\Lambda}_{j}$.
\end{lem}

\begin{proof}
(1) Consider expression (\ref{f}). In step (3.1), (3.2) and (3.4) of proof
of Theorem \ref{consis}, we have shown that the first, the second, and the
fourth term on the right hand side of expression (\ref{f}) is $O_{p}(T)$, $%
O_{p}(T\log N)$ and $O_{p}(T)$ respectively. In step (3.5) we have shown
that the left hand side of expression (\ref{f}) equals expression (\ref{h}),
and the last three terms of expression (\ref{h}) together is $O_{p}(N^{\frac{%
3}{4}}T+NT^{\frac{3}{4}})$. Move the last three terms of expression (\ref{h}%
) to the right hand side of expression (\ref{f}), and move the third term on
the right hand side of expression (\ref{f}) to the left hand side, then we
have 
\begin{eqnarray}
&&\frac{1}{2}\frac{1}{\tilde{\sigma}^{2}}[\sum\nolimits_{t=1}^{T}\left\Vert
M_{\tilde{\Lambda}_{m_{t}}}\Lambda _{z_{t}}^{0}f_{t}^{0}\right\Vert ^{2}+%
\frac{1}{2}\sum\nolimits_{t=1}^{T}x_{t}^{\prime }\tilde{\Lambda}_{m_{t}}(%
\tilde{\sigma}^{2}I_{r_{m_{t}}^{0}}+\tilde{\Lambda}_{m_{t}}^{\prime }\tilde{%
\Lambda}_{m_{t}})^{-1}(\tilde{\Lambda}_{m_{t}}^{\prime }\tilde{\Lambda}%
_{m_{t}})^{-1}\tilde{\Lambda}_{m_{t}}^{\prime }x_{t}  \notag \\
&=&O_{p}(N^{\frac{3}{4}}T+NT^{\frac{3}{4}}).  \label{k}
\end{eqnarray}%
The two terms on the left hand side of (\ref{k}) are nonnegative, thus $%
\sum\nolimits_{t=1}^{T}x_{t}^{\prime }\tilde{\Lambda}_{m_{t}}(\tilde{\sigma}%
^{2}I_{r_{m_{t}}^{0}}+\tilde{\Lambda}_{m_{t}}^{\prime }\tilde{\Lambda}%
_{m_{t}})^{-1}(\tilde{\Lambda}_{m_{t}}^{\prime }\tilde{\Lambda}_{m_{t}})^{-1}%
\tilde{\Lambda}_{m_{t}}^{\prime }x_{t}=O_{p}(N^{\frac{3}{4}}T+NT^{\frac{3}{4}%
})$.

(2)%
\begin{eqnarray*}
&&\left\Vert \sum\nolimits_{t=1}^{T}e_{t}^{\prime }\tilde{\Lambda}_{m_{t}}(%
\tilde{\sigma}^{2}I_{r_{m_{t}}^{0}}+\tilde{\Lambda}_{m_{t}}^{\prime }\tilde{%
\Lambda}_{m_{t}})^{-1}(\tilde{\Lambda}_{m_{t}}^{\prime }\tilde{\Lambda}%
_{m_{t}})^{-1}\tilde{\Lambda}_{m_{t}}^{\prime }\Lambda
_{z_{t}}^{0}f_{t}^{0}\right\Vert \\
&\leq &(\sum\nolimits_{t=1}^{T}e_{t}^{\prime }\tilde{\Lambda}_{m_{t}}(\tilde{%
\sigma}^{2}I_{r_{m_{t}}^{0}}+\tilde{\Lambda}_{m_{t}}^{\prime }\tilde{\Lambda}%
_{m_{t}})^{-2}(\tilde{\Lambda}_{m_{t}}^{\prime }\tilde{\Lambda}_{m_{t}})^{-1}%
\tilde{\Lambda}_{m_{t}}^{\prime }e_{t})^{\frac{1}{2}}(\sum%
\nolimits_{t=1}^{T}\left\Vert \Lambda _{z_{t}}^{0}f_{t}^{0}\right\Vert
^{2})^{\frac{1}{2}} \\
&\leq &\frac{1}{\tilde{\sigma}^{4}}(\sum\nolimits_{t=1}^{T}e_{t}^{\prime }P_{%
\tilde{\Lambda}_{m_{t}}}e_{t})^{\frac{1}{2}}(\sum\nolimits_{t=1}^{T}\left%
\Vert f_{t}^{0}\right\Vert ^{2})^{\frac{1}{2}}\sup_{j}\left\Vert \Lambda
_{j}^{0}\right\Vert =O_{p}(N^{\frac{3}{4}}T+NT^{\frac{3}{4}}),
\end{eqnarray*}%
where the first inequality follows from Cauchy-Schwarz inequality, the
second inequality follows from the fact that $\tilde{\Lambda}%
_{m_{t}}^{\prime }\tilde{\Lambda}_{m_{t}}$ is diagonal and all diagonal
elements of $\tilde{\sigma}^{2}I_{r_{m_{t}}^{0}}+\tilde{\Lambda}%
_{m_{t}}^{\prime }\tilde{\Lambda}_{m_{t}}$ are larger than $\tilde{\sigma}%
^{2}$, and the equality follows from Assumption \ref{factors}(2), Assumption %
\ref{loadings}(1) and expression (\ref{g}) in step (3.5) of proof of Theorem %
\ref{consis}.

(3) It follows from (1) and (2) that%
\begin{eqnarray}
&&\sum\nolimits_{t=1}^{T}f_{t}^{0\prime }\Lambda _{z_{t}}^{0\prime }\tilde{%
\Lambda}_{m_{t}}(\tilde{\sigma}^{2}I_{r_{m_{t}}^{0}}+\tilde{\Lambda}%
_{m_{t}}^{\prime }\tilde{\Lambda}_{m_{t}})^{-1}(\tilde{\Lambda}%
_{m_{t}}^{\prime }\tilde{\Lambda}_{m_{t}})^{-1}\tilde{\Lambda}%
_{m_{t}}^{\prime }\Lambda _{z_{t}}^{0}f_{t}^{0}  \notag \\
&&+\sum\nolimits_{t=1}^{T}e_{t}^{\prime }\tilde{\Lambda}_{m_{t}}(\tilde{%
\sigma}^{2}I_{r_{m_{t}}^{0}}+\tilde{\Lambda}_{m_{t}}^{\prime }\tilde{\Lambda}%
_{m_{t}})^{-1}(\tilde{\Lambda}_{m_{t}}^{\prime }\tilde{\Lambda}_{m_{t}})^{-1}%
\tilde{\Lambda}_{m_{t}}^{\prime }e_{t}  \notag \\
&=&\sum\nolimits_{t=1}^{T}x_{t}^{\prime }\tilde{\Lambda}_{m_{t}}(\tilde{%
\sigma}^{2}I_{r_{m_{t}}^{0}}+\tilde{\Lambda}_{m_{t}}^{\prime }\tilde{\Lambda}%
_{m_{t}})^{-1}(\tilde{\Lambda}_{m_{t}}^{\prime }\tilde{\Lambda}_{m_{t}})^{-1}%
\tilde{\Lambda}_{m_{t}}^{\prime }x_{t}  \notag \\
&&-2\sum\nolimits_{t=1}^{T}e_{t}^{\prime }\tilde{\Lambda}_{m_{t}}(\tilde{%
\sigma}^{2}I_{r_{m_{t}}^{0}}+\tilde{\Lambda}_{m_{t}}^{\prime }\tilde{\Lambda}%
_{m_{t}})^{-1}(\tilde{\Lambda}_{m_{t}}^{\prime }\tilde{\Lambda}_{m_{t}})^{-1}%
\tilde{\Lambda}_{m_{t}}^{\prime }\Lambda _{z_{t}}^{0}f_{t}^{0}  \notag \\
&=&O_{p}(N^{\frac{3}{4}}T+NT^{\frac{3}{4}}).  \label{n}
\end{eqnarray}%
The two terms on the left hand side of (\ref{n}) are nonnegative, thus $%
\sum\nolimits_{t=1}^{T}f_{t}^{0\prime }\Lambda _{z_{t}}^{0\prime }\tilde{%
\Lambda}_{m_{t}}(\tilde{\sigma}^{2}I_{r_{m_{t}}^{0}}+\tilde{\Lambda}%
_{m_{t}}^{\prime }\tilde{\Lambda}_{m_{t}})^{-1}(\tilde{\Lambda}%
_{m_{t}}^{\prime }\tilde{\Lambda}_{m_{t}})^{-1}\tilde{\Lambda}%
_{m_{t}}^{\prime }\Lambda _{z_{t}}^{0}f_{t}^{0}=O_{p}(N^{\frac{3}{4}}T+NT^{%
\frac{3}{4}})$. Since each term in the summation is nonnegative, we have $%
\sum\nolimits_{t=1}^{T}f_{t}^{0\prime }\Lambda _{j}^{0\prime }\tilde{\Lambda}%
_{j}(\tilde{\sigma}^{2}I_{r_{j}^{0}}+\tilde{\Lambda}_{j}^{\prime }\tilde{%
\Lambda}_{j})^{-1}(\tilde{\Lambda}_{j}^{\prime }\tilde{\Lambda}_{j})^{-1}%
\tilde{\Lambda}_{j}^{\prime }\Lambda
_{j}^{0}f_{t}^{0}1_{z_{t}=j}1_{m_{t}=j}=O_{p}(N^{\frac{3}{4}}T+NT^{\frac{3}{4%
}})$ for each $j$.

As explained in step (3.6) of proof of Theorem \ref{consis}, $%
\sum\nolimits_{t=1}^{T}1_{m_{t}=j}1_{z_{t}=j}\geq \frac{q_{j}^{0}T}{J^{0}}$,
and by Assumption \ref{factors}(1), $\rho _{\min }(\frac{1}{%
\sum\nolimits_{t=1}^{T}1_{m_{t}=j}1_{z_{t}=j}}\sum%
\nolimits_{t=1}^{T}f_{t}^{0}f_{t}^{0\prime }1_{m_{t}=j}1_{z_{t}=j})\geq c$
for some $c>0$ w.p.a.1. Thus we have%
\begin{eqnarray*}
O_{p}(N^{\frac{3}{4}}T+NT^{\frac{3}{4}})
&=&\sum\nolimits_{t=1}^{T}f_{t}^{0\prime }\Lambda _{j}^{0\prime }\tilde{%
\Lambda}_{j}(\tilde{\sigma}^{2}I_{r_{j}^{0}}+\tilde{\Lambda}_{j}^{\prime }%
\tilde{\Lambda}_{j})^{-1}(\tilde{\Lambda}_{j}^{\prime }\tilde{\Lambda}%
_{j})^{-1}\tilde{\Lambda}_{j}^{\prime }\Lambda
_{j}^{0}f_{t}^{0}1_{z_{t}=j}1_{m_{t}=j} \\
&=&tr(\Lambda _{j}^{0\prime }\tilde{\Lambda}_{j}(\tilde{\sigma}%
^{2}I_{r_{j}^{0}}+\tilde{\Lambda}_{j}^{\prime }\tilde{\Lambda}_{j})^{-1}(%
\tilde{\Lambda}_{j}^{\prime }\tilde{\Lambda}_{j})^{-1}\tilde{\Lambda}%
_{j}^{\prime }\Lambda _{j}^{0}\sum\nolimits_{t=1}^{T}f_{t}^{0}f_{t}^{0\prime
}1_{z_{t}=j}1_{m_{t}=j}) \\
&\geq &tr(\Lambda _{j}^{0\prime }\tilde{\Lambda}_{j}(\tilde{\sigma}%
^{2}I_{r_{j}^{0}}+\tilde{\Lambda}_{j}^{\prime }\tilde{\Lambda}_{j})^{-1}(%
\tilde{\Lambda}_{j}^{\prime }\tilde{\Lambda}_{j})^{-1}\tilde{\Lambda}%
_{j}^{\prime }\Lambda _{j}^{0})\rho _{\min
}(\sum\nolimits_{t=1}^{T}f_{t}^{0}f_{t}^{0\prime }1_{m_{t}=j}1_{z_{t}=j}) \\
&\geq &tr(\Lambda _{j}^{0\prime }\tilde{\Lambda}_{j}(\tilde{\sigma}%
^{2}I_{r_{j}^{0}}+\tilde{\Lambda}_{j}^{\prime }\tilde{\Lambda}_{j})^{-1}(%
\tilde{\Lambda}_{j}^{\prime }\tilde{\Lambda}_{j})^{-1}\tilde{\Lambda}%
_{j}^{\prime }\Lambda _{j}^{0})\frac{Tq_{j}^{0}}{J^{0}}c\text{ w.p.a.1.}
\end{eqnarray*}%
Thus $tr(\Lambda _{j}^{0\prime }\tilde{\Lambda}_{j}(\tilde{\sigma}%
^{2}I_{r_{j}^{0}}+\tilde{\Lambda}_{j}^{\prime }\tilde{\Lambda}_{j})^{-1}(%
\tilde{\Lambda}_{j}^{\prime }\tilde{\Lambda}_{j})^{-1}\tilde{\Lambda}%
_{j}^{\prime }\Lambda _{j}^{0})=O_{p}(\frac{N}{\sqrt{\delta _{NT}}})$ for
each $j$.

(4) Noting that $\tilde{\Lambda}_{jl}$ is orthogonal to $\tilde{\Lambda}%
_{jl^{\prime }}$ for $l\neq l^{\prime }$, we have%
\begin{eqnarray}
\sum\nolimits_{l=1}^{r_{j}^{0}}\left\Vert P_{\tilde{\Lambda}_{jl}}\Lambda
_{j}^{0}\right\Vert _{F}^{2} &=&\left\Vert P_{\tilde{\Lambda}_{j}}\Lambda
_{j}^{0}\right\Vert _{F}^{2}=\left\Vert \Lambda _{j}^{0}\right\Vert
_{F}^{2}-\left\Vert M_{\tilde{\Lambda}_{j}}\Lambda _{j}^{0}\right\Vert
_{F}^{2},  \label{p} \\
\sum\nolimits_{l=1}^{r_{j}^{0}}\frac{1}{\tilde{\sigma}^{2}+\tilde{\Lambda}%
_{jl}^{\prime }\tilde{\Lambda}_{jl}}\left\Vert P_{\tilde{\Lambda}%
_{jl}}\Lambda _{j}^{0}\right\Vert _{F}^{2} &=&tr(\Lambda _{j}^{0\prime }%
\tilde{\Lambda}_{j}(\tilde{\sigma}^{2}I_{r_{j}^{0}}+\tilde{\Lambda}%
_{j}^{\prime }\tilde{\Lambda}_{j})^{-1}(\tilde{\Lambda}_{j}^{\prime }\tilde{%
\Lambda}_{j})^{-1}\tilde{\Lambda}_{j}^{\prime }\Lambda _{j}^{0}).  \label{o}
\end{eqnarray}%
Each term in the summation on the left hand side of equation (\ref{o}) is
nonnegative, thus%
\begin{equation}
\frac{1}{\tilde{\sigma}^{2}+\tilde{\Lambda}_{jl}^{\prime }\tilde{\Lambda}%
_{jl}}\left\Vert P_{\tilde{\Lambda}_{jl}}\Lambda _{j}^{0}\right\Vert
_{F}^{2}=O_{p}(\frac{N}{\sqrt{\delta _{NT}}})\text{ for each }j\text{ and }l%
\text{.}  \label{r}
\end{equation}%
Now consider $\left\Vert P_{\tilde{\Lambda}_{j1}}\Lambda _{j}^{0}\right\Vert
_{F}^{2}$. Let $\tilde{\Lambda}_{j,-1}=(\tilde{\Lambda}_{j2},...,\tilde{%
\Lambda}_{jr_{j}^{0}})$, we have 
\begin{eqnarray}
\sum\nolimits_{l\neq 1}\left\Vert P_{\tilde{\Lambda}_{jl}}\Lambda
_{j}^{0}\right\Vert _{F}^{2} &=&\left\Vert P_{\tilde{\Lambda}_{j,-1}}\Lambda
_{j}^{0}\right\Vert _{F}^{2}=tr(\Lambda _{j}^{0\prime }P_{\tilde{\Lambda}%
_{j,-1}}\Lambda _{j}^{0})  \notag \\
&=&tr[(\tilde{\Lambda}_{j,-1}^{\prime }\tilde{\Lambda}_{j,-1})^{-\frac{1}{2}}%
\tilde{\Lambda}_{j,-1}^{\prime }\Lambda _{j}^{0}\Lambda _{j}^{0\prime }%
\tilde{\Lambda}_{j,-1}(\tilde{\Lambda}_{j,-1}^{\prime }\tilde{\Lambda}%
_{j,-1})^{-\frac{1}{2}}]  \notag \\
&\leq &\left\Vert \Lambda _{j}^{0}\right\Vert _{F}^{2}-\rho _{\min }(\Lambda
_{j}^{0}\Lambda _{j}^{0\prime }).  \label{q}
\end{eqnarray}%
The inequality in expression (\ref{q}) becomes equality when $\tilde{\Lambda}%
_{j,-1}(\tilde{\Lambda}_{j,-1}^{\prime }\tilde{\Lambda}_{j,-1})^{-\frac{1}{2}%
}$ are eigenvectors of $\Lambda _{j}^{0}\Lambda _{j}^{0\prime }$
corresponding to the largest $r_{j}^{0}-1$ eigenvalues. Expressions (\ref{p}%
) and (\ref{q}) together implies that $\left\Vert P_{\tilde{\Lambda}%
_{j1}}\Lambda _{j}^{0}\right\Vert _{F}^{2}\geq \rho _{\min }(\Lambda
_{j}^{0}\Lambda _{j}^{0\prime })-\left\Vert M_{\tilde{\Lambda}_{j}}\Lambda
_{j}^{0}\right\Vert _{F}^{2}$, thus by Assumption \ref{loadings}(1) and
Theorem \ref{consis}, $\frac{1}{N}\left\Vert P_{\tilde{\Lambda}_{j1}}\Lambda
_{j}^{0}\right\Vert _{F}^{2}$ is bounded away from zero in probability. This
together with expression (\ref{r}) implies that $\frac{1}{\tilde{\sigma}^{2}+%
\tilde{\Lambda}_{j1}^{\prime }\tilde{\Lambda}_{j1}}=O_{p}(\frac{1}{\sqrt{%
\delta _{NT}}})$. Similarly, $\frac{1}{\tilde{\sigma}^{2}+\tilde{\Lambda}%
_{jl}^{\prime }\tilde{\Lambda}_{jl}}=O_{p}(\frac{1}{\sqrt{\delta _{NT}}})$
for $l=2,...,r_{j}^{0}$.
\end{proof}

\begin{description}
\item[Proof of Theorem \protect\ref{ptjT}] 
\end{description}

\begin{proof}
Part (1):

Step (1): We first show $\left\vert \tilde{p}_{tj\left\vert t\right.
}-1_{z_{t}=j}\right\vert =o_{p}(\frac{1}{N^{\eta }})$.

When $z_{t}=j$, since $\tilde{p}_{tj\left\vert t\right. }=\frac{\tilde{p}%
_{tj\left\vert t-1\right. }L(x_{t}\left\vert z_{t}=j;\tilde{\Lambda}_{j},%
\tilde{\sigma}^{2}\right. )}{\sum\nolimits_{k=1}^{J^{0}}\tilde{p}%
_{tk\left\vert t-1\right. }L(x_{t}\left\vert z_{t}=k;\tilde{\Lambda}_{k},%
\tilde{\sigma}^{2}\right. )}$, we have%
\begin{eqnarray*}
\left\vert \tilde{p}_{tj\left\vert t\right. }-1_{z_{t}=j}\right\vert &=&%
\frac{\sum\nolimits_{k\neq j}\tilde{p}_{tk\left\vert t-1\right.
}L(x_{t}\left\vert z_{t}=k;\tilde{\Lambda}_{k},\tilde{\sigma}^{2}\right. )}{%
\sum\nolimits_{k=1}^{J^{0}}\tilde{p}_{tk\left\vert t-1\right.
}L(x_{t}\left\vert z_{t}=k;\tilde{\Lambda}_{k},\tilde{\sigma}^{2}\right. )}
\\
&\leq &\sum\nolimits_{k\neq j}\frac{\tilde{p}_{tk\left\vert t-1\right. }}{%
\tilde{p}_{tj\left\vert t-1\right. }}e^{\log L(x_{t}\left\vert z_{t}=k;%
\tilde{\Lambda}_{k},\tilde{\sigma}^{2}\right. )-\log L(x_{t}\left\vert
z_{t}=j;\tilde{\Lambda}_{j},\tilde{\sigma}^{2}\right. )}.
\end{eqnarray*}%
When $z_{t}=h\neq j$, since $\sum\nolimits_{k=1}^{J^{0}}\tilde{p}%
_{tk\left\vert t\right. }=1$, we have $\tilde{p}_{tj\left\vert t\right.
}-1_{z_{t}=j}=\tilde{p}_{tj\left\vert t\right. }\leq 1-\tilde{p}%
_{th\left\vert t\right. }$, thus it suffices to show $\left\vert \tilde{p}%
_{tj\left\vert t\right. }-1_{z_{t}=j}\right\vert =o_{p}(\frac{1}{N^{\eta }})$
when $z_{t}=j$. Since $\tilde{p}_{tj\left\vert t-1\right. }=Q_{j\cdot }%
\tilde{p}_{t-1\left\vert t-1\right. }\geq \min_{k}Q_{jk}>0$ for all $j$ ($%
Q_{j\cdot }$ denotes the $j$-th row of $Q$), it suffices to show $%
\sup_{t}e^{\log L(x_{t}\left\vert z_{t}=k;\tilde{\Lambda}_{k},\tilde{\sigma}%
^{2}\right. )-\log L(x_{t}\left\vert z_{t}=j;\tilde{\Lambda}_{j},\tilde{%
\sigma}^{2}\right. )}=o_{p}(\frac{1}{N^{\eta }})$ for any $k\neq j$, i.e.,
it suffices to show for any fixed $M>0$,%
\begin{eqnarray}
\Pr (\sup_{t}[\log L(x_{t}\left\vert z_{t}=k;\tilde{\Lambda}_{k},\tilde{%
\sigma}^{2}\right. )-\log L(x_{t}\left\vert z_{t}=j;\tilde{\Lambda}_{j},%
\tilde{\sigma}^{2}\right. )] &\geq &\log \frac{M}{N^{\eta }})\rightarrow 0,%
\text{ or}  \notag \\
\Pr (\min_{t}[\log L(x_{t}\left\vert z_{t}=j;\tilde{\Lambda}_{j},\tilde{%
\sigma}^{2}\right. )-\log L(x_{t}\left\vert z_{t}=k;\tilde{\Lambda}_{k},%
\tilde{\sigma}^{2}\right. )] &\leq &\eta \log N-\log M)\rightarrow 0\text{. }
\label{u}
\end{eqnarray}%
Similar to equation (\ref{e}),%
\begin{eqnarray*}
\log L(x_{t}\left\vert z_{t}=j;\tilde{\Lambda}_{j},\tilde{\sigma}^{2}\right.
) &=&-\frac{N}{2}\log 2\pi -\frac{1}{2}\log \left\vert \tilde{\Lambda}_{j}%
\tilde{\Lambda}_{j}^{\prime }+\tilde{\sigma}^{2}I_{N}\right\vert -\frac{1}{2}%
\tilde{\sigma}^{-2}\left\Vert M_{\tilde{\Lambda}_{j}}x_{t}\right\Vert ^{2} \\
&&-\frac{1}{2}x_{t}^{\prime }\tilde{\Lambda}_{j}(\tilde{\sigma}%
^{2}I_{r_{j}^{0}}+\tilde{\Lambda}_{j}^{\prime }\tilde{\Lambda}_{j})^{-1}(%
\tilde{\Lambda}_{j}^{\prime }\tilde{\Lambda}_{j})^{-1}\tilde{\Lambda}%
_{j}^{\prime }x_{t}, \\
\log L(x_{t}\left\vert z_{t}=k;\tilde{\Lambda}_{k},\tilde{\sigma}^{2}\right.
) &=&-\frac{N}{2}\log 2\pi -\frac{1}{2}\log \left\vert \tilde{\Lambda}_{k}%
\tilde{\Lambda}_{k}^{\prime }+\tilde{\sigma}^{2}I_{N}\right\vert -\frac{1}{2}%
\tilde{\sigma}^{-2}\left\Vert M_{\tilde{\Lambda}_{k}}x_{t}\right\Vert ^{2} \\
&&-\frac{1}{2}x_{t}^{\prime }\tilde{\Lambda}_{k}(\tilde{\sigma}%
^{2}I_{r_{k}^{0}}+\tilde{\Lambda}_{k}^{\prime }\tilde{\Lambda}_{k})^{-1}(%
\tilde{\Lambda}_{k}^{^{\prime }}\tilde{\Lambda}_{k})^{-1}\tilde{\Lambda}%
_{k}^{\prime }x_{t},
\end{eqnarray*}%
and similar to equations (\ref{s}) and (\ref{t}), $\frac{\left\vert \tilde{%
\Lambda}_{j}\tilde{\Lambda}_{j}^{\prime }+\tilde{\sigma}^{2}I_{N}\right\vert 
}{\left\vert \tilde{\Lambda}_{k}\tilde{\Lambda}_{k}^{\prime }+\tilde{\sigma}%
^{2}I_{N}\right\vert }=\frac{\left\vert \frac{1}{\tilde{\sigma}^{2}}\tilde{%
\Lambda}_{j}^{\prime }\tilde{\Lambda}_{j}+I_{r_{j}^{0}}\right\vert }{%
\left\vert \frac{1}{\tilde{\sigma}^{2}}\tilde{\Lambda}_{k}^{\prime }\tilde{%
\Lambda}_{k}+I_{r_{k}^{0}}\right\vert }$. Thus%
\begin{eqnarray}
&&\log L(x_{t}\left\vert z_{t}=j;\tilde{\Lambda}_{j},\tilde{\sigma}%
^{2}\right. )-\log L(x_{t}\left\vert z_{t}=k;\tilde{\Lambda}_{k},\tilde{%
\sigma}^{2}\right. )  \notag \\
&=&-\frac{1}{2}\log \left\vert \frac{1}{\tilde{\sigma}^{2}}\tilde{\Lambda}%
_{j}^{\prime }\tilde{\Lambda}_{j}+I_{r_{j}^{0}}\right\vert +\frac{1}{2}\log
\left\vert \frac{1}{\tilde{\sigma}^{2}}\tilde{\Lambda}_{k}^{\prime }\tilde{%
\Lambda}_{k}+I_{r_{k}^{0}}\right\vert  \notag \\
&&+\frac{1}{2}\tilde{\sigma}^{-2}(\left\Vert M_{\tilde{\Lambda}%
_{k}}x_{t}\right\Vert ^{2}-\left\Vert M_{\Lambda _{k}^{0}}x_{t}\right\Vert
^{2}+\left\Vert M_{\Lambda _{j}^{0}}x_{t}\right\Vert ^{2}-\left\Vert M_{%
\tilde{\Lambda}_{j}}x_{t}\right\Vert ^{2}+\left\Vert M_{\Lambda
_{k}^{0}}x_{t}\right\Vert ^{2}-\left\Vert M_{\Lambda
_{j}^{0}}x_{t}\right\Vert ^{2})  \notag \\
&&-\frac{1}{2}x_{t}^{\prime }\tilde{\Lambda}_{j}(\tilde{\sigma}%
^{2}I_{r_{j}^{0}}+\tilde{\Lambda}_{j}^{\prime }\tilde{\Lambda}_{j})^{-1}(%
\tilde{\Lambda}_{j}^{\prime }\tilde{\Lambda}_{j})^{-1}\tilde{\Lambda}%
_{j}^{\prime }x_{t}+\frac{1}{2}x_{t}^{\prime }\tilde{\Lambda}_{k}(\tilde{%
\sigma}^{2}I_{r_{k}^{0}}+\tilde{\Lambda}_{k}^{\prime }\tilde{\Lambda}%
_{k})^{-1}(\tilde{\Lambda}_{k}^{^{\prime }}\tilde{\Lambda}_{k})^{-1}\tilde{%
\Lambda}_{k}^{\prime }x_{t}  \notag \\
&\geq &-\frac{1}{2}\log \left\vert \frac{1}{\tilde{\sigma}^{2}}\tilde{\Lambda%
}_{j}^{\prime }\tilde{\Lambda}_{j}+I_{r_{j}^{0}}\right\vert -\frac{1}{2}%
x_{t}^{\prime }\tilde{\Lambda}_{j}(\tilde{\sigma}^{2}I_{r_{j}^{0}}+\tilde{%
\Lambda}_{j}^{\prime }\tilde{\Lambda}_{j})^{-1}(\tilde{\Lambda}_{j}^{\prime }%
\tilde{\Lambda}_{j})^{-1}\tilde{\Lambda}_{j}^{\prime }x_{t}  \notag \\
&&+\frac{1}{2}\tilde{\sigma}^{-2}(\left\Vert M_{\tilde{\Lambda}%
_{k}}x_{t}\right\Vert ^{2}-\left\Vert M_{\Lambda _{k}^{0}}x_{t}\right\Vert
^{2})+\frac{1}{2}\tilde{\sigma}^{-2}(\left\Vert M_{\Lambda
_{j}^{0}}x_{t}\right\Vert ^{2}-\left\Vert M_{\tilde{\Lambda}%
_{j}}x_{t}\right\Vert ^{2})  \notag \\
&&-\frac{1}{2}\tilde{\sigma}^{-2}e_{t}^{\prime }P_{\Lambda _{k}^{0}}e_{t}+%
\tilde{\sigma}^{-2}e_{t}^{\prime }M_{\Lambda _{k}^{0}}\Lambda
_{j}^{0}f_{t}^{0}+\frac{1}{2}\tilde{\sigma}^{-2}f_{t}^{0\prime }\Lambda
_{j}^{0\prime }M_{\Lambda _{k}^{0}}\Lambda _{j}^{0}f_{t}^{0},
\end{eqnarray}%
where the inequality follows from throwing away $\frac{1}{2}\log \left\vert 
\frac{1}{\tilde{\sigma}^{2}}\tilde{\Lambda}_{k}^{\prime }\tilde{\Lambda}%
_{k}+I_{r_{k}^{0}}\right\vert $, $\frac{1}{2}x_{t}^{\prime }\tilde{\Lambda}%
_{k}(\tilde{\sigma}^{2}I_{r_{k}^{0}}+\tilde{\Lambda}_{k}^{\prime }\tilde{%
\Lambda}_{k})^{-1}(\tilde{\Lambda}_{k}^{^{\prime }}\tilde{\Lambda}_{k})^{-1}%
\tilde{\Lambda}_{k}^{\prime }x_{t}$ and $e_{t}^{\prime }P_{\Lambda
_{j}^{0}}e_{t}$. It follows that%
\begin{eqnarray}
&&\min_{t}[\log L(x_{t}\left\vert z_{t}=j;\tilde{\Lambda}_{j},\tilde{\sigma}%
^{2}\right. )-\log L(x_{t}\left\vert z_{t}=k;\tilde{\Lambda}_{k},\tilde{%
\sigma}^{2}\right. )]  \notag \\
&\geq &-\frac{1}{2}\log \left\vert \frac{1}{\tilde{\sigma}^{2}}\tilde{\Lambda%
}_{j}^{\prime }\tilde{\Lambda}_{j}+I_{r_{j}^{0}}\right\vert -\frac{1}{2}%
\sup_{t}x_{t}^{\prime }\tilde{\Lambda}_{j}(\tilde{\sigma}^{2}I_{r_{j}^{0}}+%
\tilde{\Lambda}_{j}^{\prime }\tilde{\Lambda}_{j})^{-1}(\tilde{\Lambda}%
_{j}^{\prime }\tilde{\Lambda}_{j})^{-1}\tilde{\Lambda}_{j}^{\prime }x_{t} 
\notag \\
&&-\frac{1}{2}\tilde{\sigma}^{-2}\sup_{t}\left\vert \left\Vert M_{\tilde{%
\Lambda}_{k}}x_{t}\right\Vert ^{2}-\left\Vert M_{\Lambda
_{k}^{0}}x_{t}\right\Vert ^{2}\right\vert -\frac{1}{2}\tilde{\sigma}%
^{-2}\sup_{t}\left\vert \left\Vert M_{\Lambda _{j}^{0}}x_{t}\right\Vert
^{2}-\left\Vert M_{\tilde{\Lambda}_{j}}x_{t}\right\Vert ^{2}\right\vert 
\notag \\
&&-\frac{1}{2}\tilde{\sigma}^{-2}\sup_{t}e_{t}^{\prime }P_{\Lambda
_{k}^{0}}e_{t}-\tilde{\sigma}^{-2}\sup_{t}\left\vert e_{t}^{\prime
}M_{\Lambda _{k}^{0}}\Lambda _{j}^{0}f_{t}^{0}\right\vert +\frac{1}{2}\tilde{%
\sigma}^{-2}\min_{t}f_{t}^{0\prime }\Lambda _{j}^{0\prime }M_{\Lambda
_{k}^{0}}\Lambda _{j}^{0}f_{t}^{0}  \notag \\
&\equiv &-(A_{1}+A_{2}+A_{3}+A_{4}+A_{5}+A_{6})+\frac{1}{2}\tilde{\sigma}%
^{-2}\min_{t}f_{t}^{0\prime }\Lambda _{j}^{0\prime }M_{\Lambda
_{k}^{0}}\Lambda _{j}^{0}f_{t}^{0}.
\end{eqnarray}%
Thus for expression (\ref{u}), it suffices to show%
\begin{equation}
\Pr (\frac{1}{2}\tilde{\sigma}^{-2}\min_{t}f_{t}^{0\prime }\Lambda
_{j}^{0\prime }M_{\Lambda _{k}^{0}}\Lambda _{j}^{0}f_{t}^{0}\leq
A_{1}+A_{2}+A_{3}+A_{4}+A_{5}+A_{6}+\eta \log N)\rightarrow 0.  \label{v}
\end{equation}

By Assumption \ref{loadings}(2), $\min_{t}f_{t}^{0\prime }\Lambda
_{j}^{0\prime }M_{\Lambda _{k}^{0}}\Lambda _{j}^{0}f_{t}^{0}\geq NC$ for
some $C>0$. Thus it suffices to show that $A_{1},...,A_{6}$ are all $%
o_{p}(N) $ when $T^{\frac{16}{\alpha }}/N\rightarrow 0$ and $T^{\frac{2}{%
\alpha }+\frac{2}{\beta }}/N\rightarrow 0$.

Term $A_{1}$: As shown in equation (\ref{foc1}), $\tilde{\Lambda}%
_{jl}^{\prime }\tilde{\Lambda}_{jl}+\tilde{\sigma}^{2}$ is an eigenvalue of $%
\tilde{S}_{j}=\sum\nolimits_{t=1}^{T}\tilde{p}_{tj\left\vert T\right.
}x_{t}x_{t}^{\prime }/\sum\nolimits_{t=1}^{T}\tilde{p}_{tj\left\vert
T\right. }$, which is bounded by $\sup_{t}\left\Vert x_{t}\right\Vert ^{2}$.
We next show that $\sup_{t}\left\Vert x_{t}\right\Vert =O_{p}(N^{\frac{1}{2}%
}T^{\frac{1}{\alpha }})$. By Assumption \ref{loadings}(1) and \ref{factors}%
(2), $\sup_{t}\left\Vert \Lambda _{z_{t}}^{0}f_{t}^{0}\right\Vert ^{\alpha
}\leq \sup_{j}\left\Vert \Lambda _{j}^{0}\right\Vert ^{\alpha
}\sum\nolimits_{t=1}^{T}\left\Vert f_{t}^{0}\right\Vert ^{\alpha }=N^{\frac{%
\alpha }{2}}T$. By Holder inequality, $\left\Vert e_{t}\right\Vert
^{2}=\sum\nolimits_{i=1}^{N}e_{it}^{2}\leq
(\sum\nolimits_{i=1}^{N}e_{it}^{\alpha })^{\frac{2}{\alpha }}N^{1-\frac{2}{%
\alpha }}$, thus $\sup_{t}\left\Vert e_{t}\right\Vert ^{\alpha }\leq N^{%
\frac{\alpha }{2}-1}\sup_{t}(\sum\nolimits_{i=1}^{N}e_{it}^{\alpha })\leq N^{%
\frac{\alpha }{2}-1}\sum\nolimits_{t=1}^{T}\sum\nolimits_{i=1}^{N}e_{it}^{%
\alpha }=O_{p}(N^{\frac{\alpha }{2}}T)$ by Assumption \ref{error}(1). It
follows that $\sup_{t}\left\Vert x_{t}\right\Vert \leq \sup_{t}\left\Vert
\Lambda _{z_{t}}^{0}f_{t}^{0}\right\Vert +\sup_{t}\left\Vert
e_{t}\right\Vert =O_{p}(N^{\frac{1}{2}}T^{\frac{1}{\alpha }})$. Thus%
\begin{eqnarray*}
A_{1} &=&\frac{1}{2}\log \left\vert \frac{1}{\tilde{\sigma}^{2}}\tilde{%
\Lambda}_{j}^{\prime }\tilde{\Lambda}_{j}+I_{r_{j}^{0}}\right\vert =\frac{1}{%
2}\sum\nolimits_{l=1}^{r_{j}^{0}}\log \frac{\tilde{\Lambda}_{jl}^{\prime }%
\tilde{\Lambda}_{jl}+\tilde{\sigma}^{2}}{\tilde{\sigma}^{2}} \\
&\leq &\frac{1}{2}r_{j}^{0}\log \frac{\sup_{t}\left\Vert x_{t}\right\Vert
^{2}}{\tilde{\sigma}^{2}}=O_{p}(\log NT^{\frac{2}{\alpha }})=o_{p}(N)\text{
when }\frac{\log T}{N}\rightarrow 0\text{. }
\end{eqnarray*}

Term $A_{2}$: By Lemma \ref{lamhat norm}, $\frac{1}{\tilde{\sigma}^{2}+%
\tilde{\Lambda}_{jl}^{\prime }\tilde{\Lambda}_{jl}}=O_{p}(\frac{1}{\sqrt{%
\delta _{NT}}})$ for each $j$ and each $l$. We have shown for term $A_{1}$
that $\sup_{t}\left\Vert x_{t}\right\Vert =O_{p}(N^{\frac{1}{2}}T^{\frac{1}{%
\alpha }})$. Thus%
\begin{eqnarray*}
A_{2} &\leq &\frac{1}{2}\sup_{t}(x_{t}^{\prime }P_{\tilde{\Lambda}%
_{j}}x_{t}\sup_{l}\frac{1}{\tilde{\sigma}^{2}+\tilde{\Lambda}_{jl}^{\prime }%
\tilde{\Lambda}_{jl}})\leq \frac{1}{2}\sup_{t}\left\Vert x_{t}\right\Vert
^{2}\sup_{l}\frac{1}{\tilde{\sigma}^{2}+\tilde{\Lambda}_{jl}^{\prime }\tilde{%
\Lambda}_{jl}} \\
&=&O_{p}(NT^{\frac{2}{\alpha }})O_{p}(\frac{1}{\sqrt{\delta _{NT}}})=o_{p}(N)%
\text{ when }T^{\frac{8}{\alpha }}/N\rightarrow 0\text{ and }\alpha >8\text{.%
}
\end{eqnarray*}

Term $A_{3}$:%
\begin{eqnarray}
\left\Vert P_{\Lambda _{k}^{0}}-P_{\tilde{\Lambda}_{k}}\right\Vert ^{2}
&\leq &\left\Vert P_{\Lambda _{k}^{0}}-P_{\tilde{\Lambda}_{k}}\right\Vert
_{F}^{2}=tr[(P_{\Lambda _{k}^{0}}-P_{\tilde{\Lambda}_{k}})^{2}]  \notag \\
&=&2tr(I_{r_{k}^{0}}-P_{\Lambda _{k}^{0}}P_{\tilde{\Lambda}%
_{k}})=2\left\Vert M_{\tilde{\Lambda}_{k}}\Lambda _{k}^{0}(\Lambda
_{k}^{0\prime }\Lambda _{k}^{0})^{-\frac{1}{2}}\right\Vert _{F}^{2}  \notag
\\
&\leq &2\frac{1}{N}\left\Vert M_{\tilde{\Lambda}_{k}}\Lambda
_{k}^{0}\right\Vert _{F}^{2}\left\Vert (\frac{1}{N}\Lambda _{k}^{0\prime
}\Lambda _{k}^{0})^{-\frac{1}{2}}\right\Vert _{F}^{2}=O_{p}(\frac{1}{\sqrt{%
\delta _{NT}}}),  \label{bo}
\end{eqnarray}%
where the last equality follows from Theorem \ref{consis} and Assumption \ref%
{loadings}(1). We have shown for term $A_{1}$ that $\sup_{t}\left\Vert
x_{t}\right\Vert =O_{p}(N^{\frac{1}{2}}T^{\frac{1}{\alpha }})$. Thus 
\begin{eqnarray*}
A_{3} &=&\frac{1}{2}\tilde{\sigma}^{-2}\sup_{t}\left\vert x_{t}^{\prime
}(P_{\Lambda _{k}^{0}}-P_{\tilde{\Lambda}_{k}})x_{t}\right\vert \leq \frac{1%
}{2}\tilde{\sigma}^{-2}\left\Vert P_{\Lambda _{k}^{0}}-P_{\tilde{\Lambda}%
_{k}}\right\Vert \sup_{t}\left\Vert x_{t}\right\Vert ^{2} \\
&=&O_{p}(\delta _{NT}^{-\frac{1}{4}})NT^{\frac{2}{\alpha }}=o_{p}(N)\text{
when }T^{\frac{16}{\alpha }}/N\rightarrow 0\text{ and }\alpha >16.
\end{eqnarray*}

Similar to term $A_{3}$, Term $A_{4}$ is also $o_{p}(N)$ when $T^{\frac{16}{%
\alpha }}/N\rightarrow 0$ and $\alpha >16$.

Term $A_{5}$: By Assumption \ref{moments}(1), $\sup_{t}\left\Vert \frac{%
\Lambda _{k}^{0\prime }e_{t}}{\sqrt{N}}\right\Vert ^{\beta }\leq
\sum\nolimits_{t=1}^{T}\left\Vert \frac{\Lambda _{k}^{0\prime }e_{t}}{\sqrt{N%
}}\right\Vert ^{\beta }=O_{p}(T)$. Thus%
\begin{equation*}
A_{5}\leq \frac{1}{2}\tilde{\sigma}^{-2}\left\Vert (\frac{1}{N}\Lambda
_{k}^{0\prime }\Lambda _{k}^{0})^{-1}\right\Vert \sup_{t}\left\Vert \frac{%
\Lambda _{k}^{0\prime }e_{t}}{\sqrt{N}}\right\Vert ^{2}=O_{p}(T^{\frac{2}{%
\beta }})=o_{p}(N)\text{ when }T^{\frac{2}{\beta }}/N\rightarrow 0.
\end{equation*}

Term $A_{6}$: By Assumption \ref{factors}(2), $\sup_{t}\left\Vert
f_{t}^{0}\right\Vert ^{\alpha }\leq \sum\nolimits_{t=1}^{T}\left\Vert
f_{t}^{0}\right\Vert ^{\alpha }=O_{p}(T)$, thus $\sup_{t}\left\Vert
f_{t}^{0}\right\Vert =O_{p}(T^{\frac{1}{\alpha }})$. We have shown for term $%
A_{5}$ that $\sup_{t}\left\Vert \frac{\Lambda _{j}^{0\prime }e_{t}}{\sqrt{N}}%
\right\Vert =O_{p}(T^{\frac{1}{\beta }})$. Thus%
\begin{eqnarray*}
A_{6} &\leq &\tilde{\sigma}^{-2}\sup_{t}\left\vert e_{t}^{\prime }\Lambda
_{j}^{0}f_{t}^{0}\right\vert +\tilde{\sigma}^{-2}\sup_{t}\left\vert
e_{t}^{\prime }\Lambda _{k}^{0}(\frac{\Lambda _{k}^{0\prime }\Lambda _{k}^{0}%
}{N})^{-1}\frac{\Lambda _{k}^{0\prime }\Lambda _{j}^{0}}{N}%
f_{t}^{0}\right\vert \\
&\leq &\tilde{\sigma}^{-2}\sup_{t}\left\Vert e_{t}^{\prime }\Lambda
_{j}^{0}\right\Vert \sup_{t}\left\Vert f_{t}^{0}\right\Vert (1+\left\Vert (%
\frac{\Lambda _{k}^{0\prime }\Lambda _{k}^{0}}{N})^{-1}\right\Vert
\left\Vert \frac{\Lambda _{k}^{0\prime }\Lambda _{j}^{0}}{N}\right\Vert ) \\
&=&O_{p}(N^{\frac{1}{2}}T^{\frac{1}{\alpha }+\frac{1}{\beta }})=o_{p}(N)%
\text{ when }T^{\frac{2}{\alpha }+\frac{2}{\beta }}/N\rightarrow 0.
\end{eqnarray*}

Step (2): We next prove $\tilde{p}_{tk\left\vert T\right. }=o_{p}(\frac{1}{%
N^{\eta }})$ for $k\neq j$ when the true state is $z_{t}=j$. Let $Q_{\cdot
k} $ denote the $k$-th column of $Q$ and "$\div $" denotes element-wise
division for two vectors.%
\begin{eqnarray*}
\tilde{p}_{tk\left\vert T\right. } &=&\tilde{p}_{tk\left\vert t\right.
}\times Q_{\cdot k}^{\prime }(\tilde{p}_{t+1\left\vert T\right. }\div \tilde{%
p}_{t+1\left\vert t\right. })=\tilde{p}_{tk\left\vert t\right. }\times 
\tilde{p}_{t+1\left\vert T\right. }^{\prime }(Q_{\cdot k}\div \tilde{p}%
_{t+1\left\vert t\right. }) \\
&\leq &\tilde{p}_{tk\left\vert t\right. }\max_{l}\frac{Q_{lk}}{Q_{lj}}\frac{1%
}{\tilde{p}_{tj\left\vert t\right. }}=o_{p}(\frac{1}{N^{\eta }}),
\end{eqnarray*}%
where the inequality is due to the fact that each element of $\tilde{p}%
_{t+1\left\vert t\right. }=Q\tilde{p}_{t\left\vert t\right. }$ is not
smaller than $Q_{\cdot j}\tilde{p}_{tj\left\vert t\right. }$ and the last
equality follows from step (1) and $\min_{l}Q_{lj}>0$.

Part (2):

Similar to expression (\ref{v}), it suffices to show%
\begin{equation}
\Pr (\frac{1}{2}\tilde{\sigma}^{-2}f_{t}^{0\prime }\Lambda _{j}^{0\prime
}M_{\Lambda _{k}^{0}}\Lambda _{j}^{0}f_{t}^{0}\leq A_{1}^{\prime
}+A_{2}^{\prime }+A_{3}^{\prime }+A_{4}^{\prime }+A_{5}^{\prime
}+A_{6}^{\prime }+\eta \log N)\rightarrow 0,
\end{equation}%
where $A_{1}^{\prime },...,A_{6}^{\prime }$ equals $A_{1},...,A_{6}$ without
taking supremum with respect to $t$. Given the calculation of terms $%
A_{1},...,A_{6}$, it is not difficult to see that without taking supremum, $%
A_{1}^{\prime },...,A_{6}^{\prime }$ becomes $O_{p}(\log N)$, $O_{p}(\frac{N%
}{\sqrt{\delta _{NT}}})$, $O_{p}(\frac{N}{\delta _{NT}^{\frac{1}{4}}})$, $%
O_{p}(\frac{N}{\delta _{NT}^{\frac{1}{4}}})$, $O_{p}(1)$ and $O_{p}(N^{\frac{%
1}{2}})$ respectively. Since $f_{t}^{0\prime }\Lambda _{j}^{0\prime
}M_{\Lambda _{k}^{0}}\Lambda _{j}^{0}f_{t}^{0}\geq NC$ for some $C>0$, $%
A_{1}^{\prime },...,A_{6}^{\prime }$ are all dominated by this term.
\end{proof}

\section{Details for Theorem \protect\ref{rate} and Theorem \protect\ref{ld}}

\begin{description}
\item[Proof of Proposition \protect\ref{VjHj}] 
\end{description}

\begin{proof}
(1) Let $V_{jNT}$ be an $r_{j}^{0}\times r_{j}^{0}$ diagonal matrix
consisting of eigenvalues of $\frac{(\Lambda _{j}^{0\prime }\Lambda
_{j}^{0})^{\frac{1}{2}}(\sum\nolimits_{t=1}^{T}f_{t}^{0}f_{t}^{0\prime
}1_{z_{t}=j})(\Lambda _{j}^{0\prime }\Lambda _{j}^{0})^{\frac{1}{2}}}{%
NTq_{j}^{0}}$ in descending order and $\Upsilon _{jNT}$ be the corresponding
eigenvectors. Let $\bar{\Lambda}_{j}^{0}=\Lambda _{j}^{0}(\Lambda
_{j}^{0\prime }\Lambda _{j}^{0})^{-\frac{1}{2}}\Upsilon _{jNT}$ be the
normalized version of $\Lambda _{j}^{0}$, then $\bar{\Lambda}_{j}^{0\prime }%
\bar{\Lambda}_{j}^{0}=I_{r_{j}^{0}}$. Let $\check{\Lambda}_{j}=\tilde{\Lambda%
}_{j}(\tilde{\Lambda}_{j}^{\prime }\tilde{\Lambda}_{j})^{-\frac{1}{2}}$ be
the normalized version of $\tilde{\Lambda}_{j}$, then $\check{\Lambda}%
_{j}^{\prime }\check{\Lambda}_{j}=I_{r_{j}^{0}}$.

From equation (\ref{foc1}), we have $\check{\Lambda}_{j}W_{jNT}=(\frac{1}{NT}%
\sum\nolimits_{t=1}^{T}\tilde{p}_{tj\left\vert T\right. }x_{t}x_{t}^{\prime
})\check{\Lambda}_{j}$. The left hand side equals $P_{\bar{\Lambda}_{j}^{0}}%
\check{\Lambda}_{j}W_{jNT}+M_{\bar{\Lambda}_{j}^{0}}\check{\Lambda}%
_{j}W_{jNT}=\bar{\Lambda}_{j}^{0}\bar{\Lambda}_{j}^{0\prime }\check{\Lambda}%
_{j}W_{jNT}+M_{\bar{\Lambda}_{j}^{0}}\check{\Lambda}_{j}W_{jNT}$. The right
hand side equals%
\begin{eqnarray}
&&\Lambda _{j}^{0}\frac{(\sum\nolimits_{t=1}^{T}f_{t}^{0}f_{t}^{0\prime
}1_{z_{t}=j})\Lambda _{j}^{0\prime }\check{\Lambda}_{j}}{NT}+\frac{%
\sum\nolimits_{t=1}^{T}\mathbb{E(}e_{t}e_{t}^{\prime })1_{z_{t}=j}\check{%
\Lambda}_{j}}{NT}+\frac{\sum\nolimits_{t=1}^{T}(e_{t}e_{t}^{\prime }-\mathbb{%
E(}e_{t}e_{t}^{\prime }))1_{z_{t}=j}\check{\Lambda}_{j}}{NT}  \notag \\
&&+\frac{\sum\nolimits_{t=1}^{T}e_{t}f_{t}^{0\prime }1_{z_{t}=j}\Lambda
_{j}^{0\prime }\check{\Lambda}_{j}}{NT}+\frac{\Lambda
_{j}^{0}\sum\nolimits_{t=1}^{T}f_{t}^{0}e_{t}^{\prime }1_{z_{t}=j}\check{%
\Lambda}_{j}}{NT}+\frac{\sum\nolimits_{t=1}^{T}(\tilde{p}_{tj\left\vert
T\right. }-1_{z_{t}=j})x_{t}x_{t}^{\prime }}{NT}\check{\Lambda}_{j}  \notag
\\
&\equiv &\Lambda _{j}^{0}\frac{(\sum\nolimits_{t=1}^{T}f_{t}^{0}f_{t}^{0%
\prime }1_{z_{t}=j})\Lambda _{j}^{0\prime }\check{\Lambda}_{j}}{NT}%
+I+II+III+IV+D.  \label{z}
\end{eqnarray}%
Note that $\Lambda _{j}^{0}\frac{(\sum\nolimits_{t=1}^{T}f_{t}^{0}f_{t}^{0%
\prime }1_{z_{t}=j})\Lambda _{j}^{0\prime }\check{\Lambda}_{j}}{NT}=\bar{%
\Lambda}_{j}^{0}q_{j}^{0}V_{jNT}\bar{\Lambda}_{j}^{0\prime }\check{\Lambda}%
_{j}$, thus we have%
\begin{equation}
\bar{\Lambda}_{j}^{0}(\bar{\Lambda}_{j}^{0\prime }\check{\Lambda}%
_{j}W_{jNT}-q_{j}^{0}V_{jNT}\bar{\Lambda}_{j}^{0\prime }\check{\Lambda}%
_{j})+M_{\bar{\Lambda}_{j}^{0}}\check{\Lambda}_{j}W_{jNT}=I+II+III+IV+D
\label{w}
\end{equation}%
Terms $I,...,IV$ correspond to the right hand of equation (A.1) in Bai
(2003). By Assumption \ref{error}(2), $\left\Vert I\right\Vert
_{F}^{2}=O_{p}(\frac{1}{N})$. By Assumption \ref{error}(4), $\left\Vert
II\right\Vert _{F}^{2}=O_{p}(\frac{1}{T})$. By Assumptions \ref{moments}(2)
and \ref{loadings}(1), $\left\Vert III\right\Vert _{F}^{2}$ and $\left\Vert
IV\right\Vert _{F}^{2}$ are $O_{p}(\frac{1}{T})$. The detailed calculation
is similar to the proof of Theorem 1 in Bai and Ng (2002), hence omitted
here. Now consider term $D$. Since $\left\Vert \frac{\sum\nolimits_{t=1}^{T}(%
\tilde{p}_{tj\left\vert T\right. }-1_{z_{t}=j})x_{t}x_{t}^{\prime }}{NT}%
\right\Vert \leq \frac{\sum\nolimits_{t=1}^{T}\left\vert \tilde{p}%
_{tj\left\vert T\right. }-1_{z_{t}=j}\right\vert \left\Vert x_{t}\right\Vert
^{2}}{NT}\leq \sup_{t}\left\vert \tilde{p}_{tj\left\vert T\right.
}-1_{z_{t}=j}\right\vert \frac{\sum\nolimits_{t=1}^{T}\left\Vert
x_{t}\right\Vert ^{2}}{NT}$,%
\begin{eqnarray}
\left\Vert D\right\Vert _{F} &\leq &\left\Vert \frac{\sum\nolimits_{t=1}^{T}(%
\tilde{p}_{tj\left\vert T\right. }-1_{z_{t}=j})x_{t}x_{t}^{\prime }}{NT}%
\right\Vert \left\Vert \check{\Lambda}_{j}\right\Vert _{F}  \notag \\
&\leq &\sqrt{r_{j}^{0}}\sup_{t}\left\vert \tilde{p}_{tj\left\vert T\right.
}-1_{z_{t}=j}\right\vert \frac{\sum\nolimits_{t=1}^{T}\left\Vert
x_{t}\right\Vert ^{2}}{NT}=o_{p}(\frac{1}{N^{\eta }}).  \label{aw}
\end{eqnarray}%
The last equality follows from Theorem \ref{ptjT} and $\frac{%
\sum\nolimits_{t=1}^{T}\left\Vert x_{t}\right\Vert ^{2}}{NT}=O_{p}(1)$,
which can be easily shown using Assumptions \ref{factors}(2), \ref{loadings}%
(1) and \ref{error}(1). In summary, the right hand side of equation (\ref{w}%
) is $O_{p}(\frac{1}{\delta _{NT}})$. The two terms on the left hand side%
\footnote{%
The left hand side of equation (\ref{w}) corresponds to a further
decomposition of the left hand side of equation (A.1) in Bai (2003).} are
orthogonal to each other, thus both $\left\Vert M_{\bar{\Lambda}_{j}^{0}}%
\check{\Lambda}_{j}W_{jNT}\right\Vert _{F}$ and $\left\Vert \bar{\Lambda}%
_{j}^{0}(\bar{\Lambda}_{j}^{0\prime }\check{\Lambda}%
_{j}W_{jNT}-q_{j}^{0}V_{jNT}\bar{\Lambda}_{j}^{0\prime }\check{\Lambda}%
_{j})\right\Vert _{F}$ are $O_{p}(\frac{1}{\delta _{NT}})$. Since $%
\left\Vert A\right\Vert _{F}^{2}=tr(A^{\prime }A)$ for any matrix $A$ and $%
\bar{\Lambda}_{j}^{0}$ is orthonormal,%
\begin{equation}
\left\Vert \bar{\Lambda}_{j}^{0\prime }\check{\Lambda}%
_{j}W_{jNT}-q_{j}^{0}V_{jNT}\bar{\Lambda}_{j}^{0\prime }\check{\Lambda}%
_{j}\right\Vert _{F}=\left\Vert \bar{\Lambda}_{j}^{0}(\bar{\Lambda}%
_{j}^{0\prime }\check{\Lambda}_{j}W_{jNT}-q_{j}^{0}V_{jNT}\bar{\Lambda}%
_{j}^{0\prime }\check{\Lambda}_{j})\right\Vert _{F}=o_{p}(1).  \label{x}
\end{equation}%
We next show that equation (\ref{x}) implies that $\bar{\Lambda}%
_{j}^{0\prime }\check{\Lambda}_{j}\overset{p}{\rightarrow }I_{r_{j}^{0}}$
and $W_{jNT}\overset{p}{\rightarrow }q_{j}^{0}V_{j}$.

First, the Euclidean norm of each column of $\bar{\Lambda}_{j}^{0\prime }%
\check{\Lambda}_{j}$ converges in probability to 1 and the inner product of
different columns converges in probability to 0, because%
\begin{eqnarray}
\left\Vert I_{r_{j}^{0}}-\check{\Lambda}_{j}^{\prime }\bar{\Lambda}_{j}^{0}%
\bar{\Lambda}_{j}^{0\prime }\check{\Lambda}_{j}\right\Vert _{F} &\leq &\sqrt{%
r_{j}^{0}}\left\Vert I_{r_{j}^{0}}-\check{\Lambda}_{j}^{\prime }\bar{\Lambda}%
_{j}^{0}\bar{\Lambda}_{j}^{0\prime }\check{\Lambda}_{j}\right\Vert \leq 
\sqrt{r_{j}^{0}}tr(I_{r_{j}^{0}}-\check{\Lambda}_{j}^{\prime }\bar{\Lambda}%
_{j}^{0}\bar{\Lambda}_{j}^{0\prime }\check{\Lambda}_{j})  \notag \\
&=&\sqrt{r_{j}^{0}}\left\Vert M_{\bar{\Lambda}_{j}^{0}}\check{\Lambda}%
_{j}\right\Vert _{F}^{2}=\sqrt{r_{j}^{0}}\left\Vert M_{\check{\Lambda}_{j}}%
\bar{\Lambda}_{j}^{0}\right\Vert _{F}^{2}=o_{p}(1).  \label{y}
\end{eqnarray}%
The second inequality follows from the fact that $I_{r_{j}^{0}}-\check{%
\Lambda}_{j}^{\prime }\bar{\Lambda}_{j}^{0}\bar{\Lambda}_{j}^{0\prime }%
\check{\Lambda}_{j}$ is positive semi-definite. The second to last equality
follows from the fact that both $\bar{\Lambda}_{j}^{0}$ and $\check{\Lambda}%
_{j}$ are orthonormal. The last equality follows from Theorem \ref{consis}.

Let $V_{jNT,i}$, $W_{jNT,1}$ and $(\bar{\Lambda}_{j}^{0\prime }\check{\Lambda%
}_{j})_{i1}$ denote the $i$-th diagonal element of $V_{jNT}$, the 1st
diagonal element of $W_{jNT}$ and the $(i,1)$-th element of $\bar{\Lambda}%
_{j}^{0\prime }\check{\Lambda}_{j}$, then the first column of $\bar{\Lambda}%
_{j}^{0\prime }\check{\Lambda}_{j}W_{jNT}-q_{j}^{0}V_{jNT}\bar{\Lambda}%
_{j}^{0\prime }\check{\Lambda}_{j}$ is $(W_{jNT,1}-q_{j}^{0}V_{jNT,i})(\bar{%
\Lambda}_{j}^{0\prime }\check{\Lambda}_{j})_{i1}$, $i=1,...,r_{j}^{0}$.
Equation (\ref{x}) implies that for all $i=1,...,r_{j}^{0}$, $%
(W_{jNT,1}-q_{j}^{0}V_{jNT,i})(\bar{\Lambda}_{j}^{0\prime }\check{\Lambda}%
_{j})_{i1}$ is $o_{p}(1)$. We have shown through expression (\ref{y}) that $%
\sum\nolimits_{i=1}^{r_{j}^{0}}(\bar{\Lambda}_{j}^{0\prime }\check{\Lambda}%
_{j})_{i1}^{2}{}\overset{p}{\rightarrow }1$, thus there exists at least one
certain $i$ such that $(\bar{\Lambda}_{j}^{0\prime }\check{\Lambda}%
_{j})_{i1} $ is bounded away from zero in probability. Without loss of
generality, suppose $(\bar{\Lambda}_{j}^{0\prime }\check{\Lambda}_{j})_{11}$
is bounded away from zero in probability. Since $%
(W_{jNT,1}-q_{j}^{0}V_{jNT,1})(\bar{\Lambda}_{j}^{0\prime }\check{\Lambda}%
_{j})_{11}$ is $o_{p}(1)$, we must have $%
W_{jNT,1}-q_{j}^{0}V_{jNT,1}=o_{p}(1)$. This implies that $%
W_{jNT,1}-q_{j}^{0}V_{jNT,i}$ is bounded away from zero in probability for $%
i\neq 1$ because by Assumption \ref{diff-eig}, $V_{jNT,i}\neq V_{jNT,1}$
w.p.a.1. Since $(W_{jNT,1}-q_{j}^{0}V_{jNT,i})(\bar{\Lambda}_{j}^{0\prime }%
\check{\Lambda}_{j})_{i1}$ is $o_{p}(1)$ for all $i$, we must have $(\bar{%
\Lambda}_{j}^{0\prime }\check{\Lambda}_{j})_{i1}=o_{p}(1)$ for $i\neq 1$.
This together with $\sum\nolimits_{i=1}^{r_{j}^{0}}(\bar{\Lambda}%
_{j}^{0\prime }\check{\Lambda}_{j})_{i1}^{2}{}\overset{p}{\rightarrow }1$
implies that $(\bar{\Lambda}_{j}^{0\prime }\check{\Lambda}_{j})_{11}\overset{%
p}{\rightarrow }1$. In summary, we have shown that the first column of $\bar{%
\Lambda}_{j}^{0\prime }\check{\Lambda}_{j}$ converges in probability to $%
(1,0,...,0)$.

Similarly, for the second column of $\bar{\Lambda}_{j}^{0\prime }\check{%
\Lambda}_{j}W_{jNT}-q_{j}^{0}V_{jNT}\bar{\Lambda}_{j}^{0\prime }\check{%
\Lambda}_{j}$, we can also show that one element converges in probability to
1 and the other elements are $o_{p}(1)$. Since the inner product of the
first column and the second column of $\bar{\Lambda}_{j}^{0\prime }\check{%
\Lambda}_{j}$ is $o_{p}(1)$, $(\bar{\Lambda}_{j}^{0\prime }\check{\Lambda}%
_{j})_{12}$ must be $o_{p}(1)$. Thus $(\bar{\Lambda}_{j}^{0\prime }\check{%
\Lambda}_{j})_{i2}\overset{p}{\rightarrow }1$ for certain $i\neq 1$ and $(%
\bar{\Lambda}_{j}^{0\prime }\check{\Lambda}_{j})_{i2}=o_{p}(1)$ for all the
other $i$. Without loss of generality, suppose $(\bar{\Lambda}_{j}^{0\prime }%
\check{\Lambda}_{j})_{22}\overset{p}{\rightarrow }1$ and $(\bar{\Lambda}%
_{j}^{0\prime }\check{\Lambda}_{j})_{i2}=o_{p}(1)$ for $i\neq 2$. Since $%
(W_{jNT,2}-q_{j}^{0}V_{jNT,i})(\bar{\Lambda}_{j}^{0\prime }\check{\Lambda}%
_{j})_{i2}$ is $o_{p}(1)$ for all $i$, we must have $%
W_{jNT,2}-q_{j}^{0}V_{jNT,2}=o_{p}(1)$.

Similarly, the third column of $\bar{\Lambda}_{j}^{0\prime }\check{\Lambda}%
_{j}$ converges in probability to $(0,0,1,...,0)$ and $%
W_{jNT,3}-q_{j}^{0}V_{jNT,3}=o_{p}(1)$. Repeat the argument for all columns
of $\bar{\Lambda}_{j}^{0\prime }\check{\Lambda}_{j}$, we have $\bar{\Lambda}%
_{j}^{0\prime }\check{\Lambda}_{j}\overset{p}{\rightarrow }I_{r_{j}^{0}}$
and $W_{jNT}-q_{j}^{0}V_{jNT}=o_{p}(1)$. Since $V_{jNT}\overset{p}{%
\rightarrow }V_{j}$, we have $W_{jNT}\overset{p}{\rightarrow }q_{j}^{0}V_{j}$%
.

(2) By Theorem \ref{ptjT}(1), $\left\vert \frac{1}{T}\sum\nolimits_{t=1}^{T}(%
\tilde{p}_{tj\left\vert T\right. }-1_{z_{t}=j})\right\vert \leq
\sup_{t}\left\vert \tilde{p}_{tj\left\vert T\right. }-1_{z_{t}=j}\right\vert
=o_{p}(\frac{1}{N^{\eta }})$. By Assumption \ref{state}, $\frac{1}{T}%
\sum\nolimits_{t=1}^{T}1_{z_{t}=j}\overset{p}{\rightarrow }q_{j}^{0}$. Thus $%
\frac{1}{T}\sum\nolimits_{t=1}^{T}\tilde{p}_{tj\left\vert T\right. }\overset{%
p}{\rightarrow }q_{j}^{0}$. We have shown that $W_{jNT}\overset{p}{%
\rightarrow }q_{j}^{0}V_{j}$, thus $\frac{\tilde{\Lambda}_{j}^{\prime }%
\tilde{\Lambda}_{j}}{N}=W_{jNT}/\frac{1}{T}\sum\nolimits_{t=1}^{T}\tilde{p}%
_{tj\left\vert T\right. }-\frac{\tilde{\sigma}^{2}}{N}I_{r_{j}^{0}}\overset{p%
}{\rightarrow }V_{j}$. It follows that%
\begin{eqnarray}
H_{j} &=&\frac{\sum\nolimits_{t=1}^{T}f_{t}^{0}f_{t}^{0\prime }1_{z_{t}=j}}{T%
}\frac{\Lambda _{j}^{0\prime }\tilde{\Lambda}_{j}}{N}W_{jNT}^{-1}  \notag \\
&=&(\Lambda _{j}^{0\prime }\Lambda _{j}^{0})^{-\frac{1}{2}}\frac{(\Lambda
_{j}^{0\prime }\Lambda _{j}^{0})^{\frac{1}{2}}(\sum%
\nolimits_{t=1}^{T}f_{t}^{0}f_{t}^{0\prime }1_{z_{t}=j})(\Lambda
_{j}^{0\prime }\Lambda _{j}^{0})^{\frac{1}{2}}}{NT}(\Lambda _{j}^{0\prime
}\Lambda _{j}^{0})^{-\frac{1}{2}}\Lambda _{j}^{0\prime }\check{\Lambda}_{j}(%
\tilde{\Lambda}_{j}^{\prime }\tilde{\Lambda}_{j})^{\frac{1}{2}}W_{jNT}^{-1} 
\notag \\
&=&(\frac{\Lambda _{j}^{0\prime }\Lambda _{j}^{0}}{N})^{-\frac{1}{2}%
}\Upsilon _{jNT}V_{jNT}(\bar{\Lambda}_{j}^{0\prime }\check{\Lambda}_{j})(%
\frac{\tilde{\Lambda}_{j}^{\prime }\tilde{\Lambda}_{j}}{N})^{\frac{1}{2}%
}W_{jNT}^{-1}q_{j}^{0}\overset{p}{\rightarrow }\Sigma _{\Lambda _{j}}^{-%
\frac{1}{2}}\Upsilon _{j}V_{j}^{\frac{1}{2}}.
\end{eqnarray}
\end{proof}

\begin{description}
\item[Proof of Theorem \protect\ref{rate}] 
\end{description}

\begin{proof}
From equation (\ref{z}), we have $\tilde{\Lambda}_{j}W_{jNT}=\Lambda _{j}^{0}%
\frac{(\sum\nolimits_{t=1}^{T}f_{t}^{0}f_{t}^{0\prime }1_{z_{t}=j})\Lambda
_{j}^{0\prime }\tilde{\Lambda}_{j}}{NT}+(I+II+III+IV+D)(\tilde{\Lambda}%
_{j}^{\prime }\tilde{\Lambda}_{j})^{\frac{1}{2}}$, i.e.,%
\begin{equation}
\tilde{\Lambda}_{j}-\Lambda _{j}^{0}H_{j}=(I+II+III+IV+D)(\tilde{\Lambda}%
_{j}^{\prime }\tilde{\Lambda}_{j})^{\frac{1}{2}}W_{jNT}^{-1}.  \label{aa}
\end{equation}%
We have shown in Proposition \ref{VjHj} that $\left\Vert
I+II+III+IV+D\right\Vert _{F}^{2}=O_{p}(\frac{1}{\delta _{NT}^{2}})$, $%
W_{jNT}\overset{p}{\rightarrow }q_{j}^{0}V_{j}$ and $\frac{\tilde{\Lambda}%
_{j}^{\prime }\tilde{\Lambda}_{j}}{N}\overset{p}{\rightarrow }V_{j}$. Thus $%
\frac{1}{N}\left\Vert \tilde{\Lambda}_{j}-\Lambda _{j}^{0}H_{j}\right\Vert
_{F}^{2}=O_{p}(\frac{1}{\delta _{NT}^{2}})$.
\end{proof}

\begin{description}
\item[Proof of Theorem \protect\ref{ld}] 
\end{description}

\begin{proof}
Let $I_{i}$, $II_{i}$, $III_{i}$, $IV_{i}$ and $D_{i}$ denote the $i$-th row
of $I$, $II$, $III$, $IV$ and $D$ respectively. From equation (\ref{aa}), we
have%
\begin{equation*}
\tilde{\lambda}_{ji}^{\prime }-\lambda _{ji}^{0\prime
}H_{j}=(I_{i}+II_{i}+III_{i}+IV_{i}+D_{i})(\tilde{\Lambda}_{j}^{\prime }%
\tilde{\Lambda}_{j})^{\frac{1}{2}}W_{jNT}^{-1}.
\end{equation*}%
By Assumptions \ref{loadings}(1) and \ref{error}(2) and Theorem \ref{rate}, $%
I_{i}(\tilde{\Lambda}_{j}^{\prime }\tilde{\Lambda}_{j})^{\frac{1}{2}}=O_{p}(%
\frac{1}{\sqrt{N}\delta _{NT}})$. By Assumptions \ref{error}(4) and \ref%
{dist}(1) and Theorem \ref{rate}, $II_{i}(\tilde{\Lambda}_{j}^{\prime }%
\tilde{\Lambda}_{j})^{\frac{1}{2}}=O_{p}(\frac{1}{\sqrt{T}\delta _{NT}})$.
By Assumption \ref{error}(2) and Theorem \ref{rate}, $III_{i}(\tilde{\Lambda}%
_{j}^{\prime }\tilde{\Lambda}_{j})^{\frac{1}{2}}=\frac{\sum%
\nolimits_{t=1}^{T}e_{it}f_{t}^{0\prime }1_{z_{t}=j}\Lambda _{j}^{0\prime
}\Lambda _{j}^{0}H_{j}}{NT}+O_{p}(\frac{1}{\sqrt{T}\delta _{NT}})$. By
Assumptions \ref{error}(2) and \ref{dist}(2) and Theorem \ref{rate}, $IV_{i}(%
\tilde{\Lambda}_{j}^{\prime }\tilde{\Lambda}_{j})^{\frac{1}{2}}=O_{p}(\frac{1%
}{\sqrt{T}\delta _{NT}})$. The detailed calculation of these four terms is
similar to the proof of Lemma A.2 in Bai (2003), hence omitted here. For the
term $D_{i}(\tilde{\Lambda}_{j}^{\prime }\tilde{\Lambda}_{j})^{\frac{1}{2}}$%
, we have%
\begin{eqnarray*}
\left\Vert D_{i}(\tilde{\Lambda}_{j}^{\prime }\tilde{\Lambda}_{j})^{\frac{1}{%
2}}\right\Vert ^{2} &=&\left\Vert \frac{1}{NT}\sum\nolimits_{t=1}^{T}(\tilde{%
p}_{tj\left\vert T\right. }-1_{z_{t}=j})x_{it}x_{t}^{\prime }\tilde{\Lambda}%
_{j}\right\Vert ^{2} \\
&\leq &\frac{1}{N^{2}T^{2}}\sum\nolimits_{t=1}^{T}(\tilde{p}_{tj\left\vert
T\right. }-1_{z_{t}=j})^{2}x_{it}^{2}\sum\nolimits_{t=1}^{T}\left\Vert
x_{t}\right\Vert ^{2}\left\Vert \tilde{\Lambda}_{j}\right\Vert _{F}^{2} \\
&\leq &\sup_{t}\left\vert \tilde{p}_{tj\left\vert T\right.
}-1_{z_{t}=j}\right\vert ^{2}\frac{\sum\nolimits_{t=1}^{T}x_{it}^{2}}{T}%
\frac{\sum\nolimits_{t=1}^{T}\left\Vert x_{t}\right\Vert ^{2}}{NT}\frac{%
\left\Vert \tilde{\Lambda}_{j}\right\Vert _{F}^{2}}{N}=o_{p}(\frac{1}{%
N^{2\eta }}),
\end{eqnarray*}%
where the last equality follows from Theorem \ref{ptjT}. We have shown in
Proposition \ref{VjHj} that $W_{jNT}\overset{p}{\rightarrow }q_{j}^{0}V_{j}$%
, thus $W_{jNT}^{-1}=O_{p}(1)$. It follows that%
\begin{equation*}
\sqrt{Tq_{j}^{0}}(\tilde{\lambda}_{ji}-H_{j}^{\prime }\lambda
_{ji}^{0})=q_{j}^{0}W_{jNT}^{-1}H_{j}^{\prime }\frac{\Lambda _{j}^{0\prime
}\Lambda _{j}^{0}}{N}\frac{\sum\nolimits_{t=1}^{T}f_{t}^{0}e_{it}1_{z_{t}=j}%
}{\sqrt{Tq_{j}^{0}}}+O_{p}(\frac{\sqrt{T}}{N})+o_{p}(1).
\end{equation*}%
Thus by Proposition \ref{VjHj}\ and Assumption \ref{dist}(3),%
\begin{equation*}
\sqrt{Tq_{j}^{0}}(\tilde{\lambda}_{ji}-H_{j}^{\prime }\lambda _{ji}^{0})%
\overset{d}{\rightarrow }\mathcal{N}(0,V_{j}^{-\frac{1}{2}}\Upsilon
_{j}^{\prime }\Sigma _{\Lambda _{j}}^{\frac{1}{2}}\Phi _{ji}\Sigma _{\Lambda
_{j}}^{\frac{1}{2}}\Upsilon _{j}V_{j}^{-\frac{1}{2}})\text{ when }\sqrt{T}%
/N\rightarrow 0.
\end{equation*}
\end{proof}

\section{Details for Theorem \protect\ref{factor}}

\begin{lem}
\label{lambdahat}Under Assumptions \ref{factors}-\ref{dist}, and assume $%
T^{^{\frac{16}{\alpha }}}/N\rightarrow 0$ and $T^{\frac{2}{\alpha }+\frac{2}{%
\beta }}/N\rightarrow 0$,

(1) $\frac{1}{N}e_{t}^{\prime }(\tilde{\Lambda}_{j}-\Lambda
_{j}^{0}H_{j})=O_{p}(\frac{1}{\delta _{NT}^{2}})$ for each $j$ and $t$,

(2) $\frac{1}{N}\Lambda _{j}^{0\prime }(\tilde{\Lambda}_{j}-\Lambda
_{j}^{0}H_{j})=O_{p}(\frac{1}{\delta _{NT}^{2}})$ for each $j$.
\end{lem}

\begin{proof}
Part (1): From equation (\ref{aa}), we have $\frac{1}{N}e_{t}^{\prime }(%
\tilde{\Lambda}_{j}-\Lambda _{j}^{0}H_{j})=\frac{1}{N}e_{t}^{\prime
}(I+II+III+IV+D)(\tilde{\Lambda}_{j}^{\prime }\tilde{\Lambda}_{j})^{\frac{1}{%
2}}W_{jNT}^{-1}$. Consider each term one by one. 
\begin{equation*}
\frac{e_{t}^{\prime }I(\tilde{\Lambda}_{j}^{\prime }\tilde{\Lambda}_{j})^{%
\frac{1}{2}}}{N}=e_{t}^{\prime }\frac{\sum\nolimits_{t=1}^{T}\mathbb{E(}%
e_{t}e_{t}^{\prime })1_{z_{t}=j}}{N^{2}T}\Lambda _{j}^{0}H_{j}+e_{t}^{\prime
}\frac{\sum\nolimits_{t=1}^{T}\mathbb{E(}e_{t}e_{t}^{\prime })1_{z_{t}=j}}{%
N^{2}T}(\tilde{\Lambda}_{j}-\Lambda _{j}^{0}H_{j}).
\end{equation*}%
By Assumption \ref{error}(1) and \ref{error}(2), the first term is $O_{p}(%
\frac{1}{N})$. By Assumption \ref{error}(2) and Theorem \ref{rate}, the
second term is $O_{p}(\frac{1}{\sqrt{N}\delta _{NT}})$.%
\begin{equation*}
\frac{e_{t}^{\prime }II(\tilde{\Lambda}_{j}^{\prime }\tilde{\Lambda}_{j})^{%
\frac{1}{2}}}{N}=e_{t}^{\prime }\frac{\sum\nolimits_{t=1}^{T}(e_{t}e_{t}^{%
\prime }-\mathbb{E(}e_{t}e_{t}^{\prime }))1_{z_{t}=j}}{N^{2}T}\Lambda
_{j}^{0}H_{j}+e_{t}^{\prime }\frac{\sum\nolimits_{t=1}^{T}(e_{t}e_{t}^{%
\prime }-\mathbb{E(}e_{t}e_{t}^{\prime }))1_{z_{t}=j}}{N^{2}T}(\tilde{\Lambda%
}_{j}-\Lambda _{j}^{0}H_{j}).
\end{equation*}%
By Assumption \ref{dist}(1), the first term is $O_{p}(\frac{1}{\sqrt{NT}})$.
By Assumption \ref{error}(4) and Theorem \ref{rate}, the second term is $%
O_{p}(\frac{1}{\sqrt{T}\delta _{NT}})$.%
\begin{equation*}
\frac{e_{t}^{\prime }III(\tilde{\Lambda}_{j}^{\prime }\tilde{\Lambda}_{j})^{%
\frac{1}{2}}}{N}=e_{t}^{\prime }\frac{\sum\nolimits_{t=1}^{T}e_{t}f_{t}^{0%
\prime }1_{z_{t}=j}\Lambda _{j}^{0\prime }}{N^{2}T}\Lambda
_{j}^{0}H_{j}+e_{t}^{\prime }\frac{\sum\nolimits_{t=1}^{T}e_{t}f_{t}^{0%
\prime }1_{z_{t}=j}\Lambda _{j}^{0\prime }}{N^{2}T}(\tilde{\Lambda}%
_{j}-\Lambda _{j}^{0}H_{j}).
\end{equation*}%
The first term is $O_{p}(\frac{1}{\sqrt{NT}})+O_{p}(\frac{1}{T})$ since $%
\frac{1}{NT}e_{t}^{\prime }\sum\nolimits_{t=1}^{T}e_{t}f_{t}^{0\prime
}1_{z_{t}=j}=O_{p}(\frac{1}{\sqrt{NT}})+O_{p}(\frac{1}{T})$, which follows
from Assumptions \ref{error}(3), \ref{dist}(1) and $e_{it}e_{is}=\gamma
_{i,ts}+(e_{it}e_{is}-\gamma _{i,ts})$. By Theorem \ref{rate}, the second
term is $O_{p}(\frac{1}{\sqrt{T}\delta _{NT}})$.%
\begin{equation*}
\frac{e_{t}^{\prime }IV(\tilde{\Lambda}_{j}^{\prime }\tilde{\Lambda}_{j})^{%
\frac{1}{2}}}{N}=e_{t}^{\prime }\frac{\Lambda
_{j}^{0}\sum\nolimits_{t=1}^{T}f_{t}^{0}e_{t}^{\prime }1_{z_{t}=j}}{N^{2}T}%
\Lambda _{j}^{0}H_{j}+e_{t}^{\prime }\frac{\Lambda
_{j}^{0}\sum\nolimits_{t=1}^{T}f_{t}^{0}e_{t}^{\prime }1_{z_{t}=j}}{N^{2}T}(%
\tilde{\Lambda}_{j}-\Lambda _{j}^{0}H_{j}).
\end{equation*}%
The first term is $O_{p}(\frac{1}{\sqrt{NT}})$ since by Assumption \ref{dist}%
(2), $\frac{1}{NT}\sum\nolimits_{t=1}^{T}f_{t}^{0}e_{t}^{\prime
}1_{z_{t}=j}\Lambda _{j}^{0}=O_{p}(\frac{1}{\sqrt{NT}})$. By Theorem \ref%
{rate}, the second term is $O_{p}(\frac{1}{\sqrt{T}\delta _{NT}})$.%
\begin{equation*}
\left\Vert \frac{e_{t}^{\prime }D(\tilde{\Lambda}_{j}^{\prime }\tilde{\Lambda%
}_{j})^{\frac{1}{2}}}{N}\right\Vert \leq \left\Vert \frac{e_{t}^{\prime }}{%
\sqrt{N}}\right\Vert \left\Vert D\right\Vert _{F}\left\Vert (\frac{\tilde{%
\Lambda}_{j}^{\prime }\tilde{\Lambda}_{j}}{N})^{\frac{1}{2}}\right\Vert .
\end{equation*}%
Thus from equation (\ref{aw}), this term is $o_{p}(\frac{1}{N^{\eta }})$.
Finally, note that $W_{jNT}^{-1}\overset{d}{\rightarrow }\frac{1}{q_{j}^{0}}%
V_{j}^{-1}$, part (1) is proved.

Part (2) can be proved similarly.
\end{proof}

\begin{description}
\item[Proof of Theorem \protect\ref{factor}] 
\end{description}

\begin{proof}
First, by Woodbury identity,%
\begin{eqnarray*}
\tilde{f}_{t} &=&\sum\nolimits_{j=1}^{J^{0}}(\tilde{\sigma}^{2}I_{r_{j}^{0}}+%
\tilde{\Lambda}_{j}^{\prime }\tilde{\Lambda}_{j})^{-1}\tilde{\Lambda}%
_{j}^{\prime }x_{t}\tilde{p}_{tj\left\vert T\right. } \\
&=&\sum\nolimits_{k=1}^{J^{0}}\sum\nolimits_{j=1}^{J^{0}}\tilde{p}%
_{tj\left\vert T\right. }(\tilde{\sigma}^{2}I_{r_{j}^{0}}+\tilde{\Lambda}%
_{j}^{\prime }\tilde{\Lambda}_{j})^{-1}\tilde{\Lambda}_{j}^{\prime }\Lambda
_{k}^{0}f_{t}^{0}1_{z_{t}=k}+\sum\nolimits_{j=1}^{J^{0}}\tilde{p}%
_{tj\left\vert T\right. }(\tilde{\sigma}^{2}I_{r_{j}^{0}}+\tilde{\Lambda}%
_{j}^{\prime }\tilde{\Lambda}_{j})^{-1}\tilde{\Lambda}_{j}^{\prime }e_{t}.
\end{eqnarray*}%
When $z_{t}=k$, we have%
\begin{eqnarray*}
&&\sum\nolimits_{j=1}^{J^{0}}\tilde{p}_{tj\left\vert T\right. }(\tilde{\sigma%
}^{2}I_{r_{j}^{0}}+\tilde{\Lambda}_{j}^{\prime }\tilde{\Lambda}_{j})^{-1}%
\tilde{\Lambda}_{j}^{\prime }\Lambda _{k}^{0}f_{t}^{0} \\
&=&(\tilde{\sigma}^{2}I_{r_{k}^{0}}+\tilde{\Lambda}_{k}^{\prime }\tilde{%
\Lambda}_{k})^{-1}\tilde{\Lambda}_{k}^{\prime }\Lambda _{k}^{0}f_{t}^{0} \\
&&+(\tilde{p}_{tk\left\vert T\right. }-1)(\tilde{\sigma}^{2}I_{r_{k}^{0}}+%
\tilde{\Lambda}_{k}^{\prime }\tilde{\Lambda}_{k})^{-1}\tilde{\Lambda}%
_{k}^{\prime }\Lambda _{k}^{0}f_{t}^{0}+\sum\nolimits_{j\neq k}\tilde{p}%
_{tj\left\vert T\right. }(\tilde{\sigma}^{2}I_{r_{j}^{0}}+\tilde{\Lambda}%
_{j}^{\prime }\tilde{\Lambda}_{j})^{-1}\tilde{\Lambda}_{j}^{\prime }\Lambda
_{k}^{0}f_{t}^{0} \\
&=&H_{k}^{-1}f_{t}^{0}+(\tilde{\sigma}^{2}I_{r_{k}^{0}}+\tilde{\Lambda}%
_{k}^{\prime }\tilde{\Lambda}_{k})^{-1}\tilde{\Lambda}_{k}^{\prime }(\Lambda
_{k}^{0}H_{k}-\tilde{\Lambda}_{k})H_{k}^{-1}f_{t}^{0}-(\tilde{\sigma}%
^{2}I_{r_{k}^{0}}+\tilde{\Lambda}_{k}^{\prime }\tilde{\Lambda}_{k})^{-1}%
\tilde{\sigma}^{2}H_{k}^{-1}f_{t}^{0} \\
&&+(\tilde{p}_{tk\left\vert T\right. }-1)(\tilde{\sigma}^{2}I_{r_{k}^{0}}+%
\tilde{\Lambda}_{k}^{\prime }\tilde{\Lambda}_{k})^{-1}\tilde{\Lambda}%
_{k}^{\prime }\Lambda _{k}^{0}f_{t}^{0}+\sum\nolimits_{j\neq k}\tilde{p}%
_{tj\left\vert T\right. }(\tilde{\sigma}^{2}I_{r_{j}^{0}}+\tilde{\Lambda}%
_{j}^{\prime }\tilde{\Lambda}_{j})^{-1}\tilde{\Lambda}_{j}^{\prime }\Lambda
_{k}^{0}f_{t}^{0} \\
&\equiv &H_{k}^{-1}f_{t}^{0}+B_{k1t}-B_{k2t}+B_{k3t}+B_{k4t},
\end{eqnarray*}%
and%
\begin{eqnarray*}
&&\sum\nolimits_{j=1}^{J^{0}}\tilde{p}_{tj\left\vert T\right. }(\tilde{\sigma%
}^{2}I_{r_{j}^{0}}+\tilde{\Lambda}_{j}^{\prime }\tilde{\Lambda}_{j})^{-1}%
\tilde{\Lambda}_{j}^{\prime }e_{t} \\
&=&(\tilde{\sigma}^{2}I_{r_{k}^{0}}+\tilde{\Lambda}_{k}^{\prime }\tilde{%
\Lambda}_{k})^{-1}(\tilde{\Lambda}_{k}-\Lambda _{k}^{0}H_{k})^{\prime
}e_{t}+(\tilde{\sigma}^{2}I_{r_{k}^{0}}+\tilde{\Lambda}_{k}^{\prime }\tilde{%
\Lambda}_{k})^{-1}H_{k}^{\prime }\Lambda _{k}^{0\prime }e_{t} \\
&&+(\tilde{p}_{tk\left\vert T\right. }-1)(\tilde{\sigma}^{2}I_{r_{k}^{0}}+%
\tilde{\Lambda}_{k}^{\prime }\tilde{\Lambda}_{k})^{-1}\tilde{\Lambda}%
_{k}^{\prime }e_{t}+\sum\nolimits_{j\neq k}^{J^{0}}\tilde{p}_{tj\left\vert
T\right. }(\tilde{\sigma}^{2}I_{r_{j}^{0}}+\tilde{\Lambda}_{j}^{\prime }%
\tilde{\Lambda}_{j})^{-1}\tilde{\Lambda}_{j}^{\prime }e_{t} \\
&\equiv &C_{k1t}+C_{k2t}+C_{k3t}+C_{k4t}.
\end{eqnarray*}%
It follows that $\tilde{f}%
_{t}-H_{z_{t}}^{-1}f_{t}^{0}=B_{z_{t}1t}-B_{z_{t}2t}+B_{z_{t}3t}+B_{z_{t}4t}+C_{z_{t}1t}+C_{z_{t}2t}+C_{z_{t}3t}+C_{z_{t}4t} 
$.

Proof of part (1): First consider $B_{z_{t}1t}$.%
\begin{equation*}
\frac{\sum\nolimits_{t=1}^{T}\left\Vert B_{z_{t}1t}\right\Vert ^{2}}{T}\leq
\sum\nolimits_{j=1}^{J^{0}}\left\Vert (\tilde{\sigma}^{2}I_{r_{j}^{0}}+%
\tilde{\Lambda}_{j}^{\prime }\tilde{\Lambda}_{j})^{-1}\tilde{\Lambda}%
_{j}^{\prime }(\Lambda _{j}^{0}H_{j}-\tilde{\Lambda}_{j})H_{j}^{-1}\right%
\Vert ^{2}\frac{\sum\nolimits_{t=1}^{T}\left\Vert f_{t}^{0}\right\Vert ^{2}}{%
T}=O_{p}(\frac{1}{\delta _{NT}^{2}}),
\end{equation*}%
where the equality is due to the following facts:

(1) By Proposition \ref{VjHj}(1), $\frac{1}{N}(\tilde{\sigma}%
^{2}I_{r_{j}^{0}}+\tilde{\Lambda}_{j}^{\prime }\tilde{\Lambda}_{j})=W_{jNT}/%
\frac{1}{T}\sum\nolimits_{t=1}^{T}\tilde{p}_{tj\left\vert T\right. }\overset{%
p}{\rightarrow }V_{j}$ for all $j$.

(2) $\left\Vert \frac{1}{\sqrt{N}}\tilde{\Lambda}_{j}^{\prime }\right\Vert
_{F}=\sqrt{tr(\frac{1}{N}\tilde{\Lambda}_{j}^{\prime }\tilde{\Lambda}_{j})}%
\overset{p}{\rightarrow }\sqrt{tr(V_{j})}$ for all $j$.

(3) By Theorem \ref{rate}, $\left\Vert \frac{1}{\sqrt{N}}(\Lambda
_{j}^{0}H_{j}-\tilde{\Lambda}_{j})\right\Vert =O_{p}(\frac{1}{\delta _{NT}})$
for all $j$.

(4) By Proposition \ref{VjHj}(1), $\left\Vert H_{j}^{-1}\right\Vert
=O_{p}(1) $ for all $j$.

(5) $\frac{1}{T}\sum\nolimits_{t=1}^{T}\left\Vert f_{t}^{0}\right\Vert
^{2}=O_{p}(1)$ by Assumption \ref{factors}.

It is easy to see that $\frac{1}{T}\sum\nolimits_{t=1}^{T}\left\Vert
B_{z_{t}2t}\right\Vert ^{2}=O_{p}(\frac{1}{N^{2}})$. For $B_{z_{t}3t}$, we
have%
\begin{equation*}
\frac{\sum\nolimits_{t=1}^{T}\left\Vert B_{z_{t}3t}\right\Vert ^{2}}{T}\leq
\sup_{t}\left\Vert \tilde{p}_{tz_{t}\left\vert T\right. }-1\right\Vert
^{2}\left\Vert (\tilde{\sigma}^{2}I_{r_{z_{t}}^{0}}+\tilde{\Lambda}%
_{z_{t}}^{\prime }\tilde{\Lambda}_{z_{t}})^{-1}\tilde{\Lambda}%
_{z_{t}}^{\prime }\Lambda _{z_{t}}^{0}\right\Vert ^{2}\frac{%
\sum\nolimits_{t=1}^{T}\left\Vert f_{t}^{0}\right\Vert ^{2}}{T}=o_{p}(\frac{1%
}{N^{2\eta }}),
\end{equation*}%
where the equality is due to $\sup_{t}\left\Vert \tilde{p}_{tz_{t}\left\vert
T\right. }-1\right\Vert \leq \sup_{j}\sup_{t}\left\Vert \tilde{p}%
_{tj\left\vert T\right. }-1_{z_{t}=j}\right\Vert =o_{p}(\frac{1}{N^{\eta }})$
by Theorem \ref{ptjT}. It is easy to see that $\frac{1}{T}%
\sum\nolimits_{t=1}^{T}\left\Vert B_{z_{t}4t}\right\Vert ^{2}=o_{p}(\frac{1}{%
N^{2\eta }})$. Similarly, we can show that $\frac{\sum\nolimits_{t=1}^{T}%
\left\Vert C_{z_{t}1t}\right\Vert ^{2}}{T}$ is $O_{p}(\frac{1}{\delta
_{NT}^{2}})$, $\frac{\sum\nolimits_{t=1}^{T}\left\Vert
C_{z_{t}2t}\right\Vert ^{2}}{T}$ is $O_{p}(\frac{1}{N})$, and both $\frac{%
\sum\nolimits_{t=1}^{T}\left\Vert C_{z_{t}3t}\right\Vert ^{2}}{T}$ and $%
\frac{\sum\nolimits_{t=1}^{T}\left\Vert C_{z_{t}4t}\right\Vert ^{2}}{T}$ are 
$o_{p}(\frac{1}{N^{2\eta }})(O_{p}(\frac{1}{N})+O_{p}(\frac{1}{\delta
_{NT}^{2}}))$. Thus $\frac{1}{T}\sum\nolimits_{t=1}^{T}\left\Vert \tilde{f}%
_{t}-H_{z_{t}}^{-1}f_{t}^{0}\right\Vert ^{2}=O_{p}(\frac{1}{\delta _{NT}^{2}}%
)$.

Proof of part (2): By Theorem \ref{rate} and Lemma \ref{lambdahat}(2), $%
\tilde{\Lambda}_{k}^{\prime }(\Lambda _{k}^{0}H_{k}-\tilde{\Lambda}%
_{k})=O_{p}(\frac{1}{\delta _{NT}^{2}})$ for any $k$. This together with
fact (1) and fact (4) listed above implies $B_{z_{t}1t}=O_{p}(\frac{1}{%
\delta _{NT}^{2}})$. Similarly, it is easy to see that $B_{z_{t}2t}=O_{p}(%
\frac{1}{N})$, $B_{z_{t}3t}=o_{p}(\frac{1}{N^{\eta }})$, $B_{z_{t}4t}=o_{p}(%
\frac{1}{N^{\eta }})$, $C_{z_{t}1t}=O_{p}(\frac{1}{\delta _{NT}^{2}})$, $%
C_{z_{t}2t}=O_{p}(\frac{1}{\sqrt{N}})$, $C_{z_{t}3t}=o_{p}(\frac{1}{N^{\eta }%
})$ and $C_{z_{t}4t}=o_{p}(\frac{1}{N^{\eta }})$. The leading term is $%
C_{z_{t}2t}$. Since $\frac{\tilde{\Lambda}_{z_{t}}^{\prime }\tilde{\Lambda}%
_{z_{t}}}{N}\overset{p}{\rightarrow }V_{z_{t}}$, $H_{z_{t}}\overset{p}{%
\rightarrow }\Sigma _{\Lambda _{z_{t}}}^{-\frac{1}{2}}\Upsilon
_{z_{t}}V_{z_{t}}^{\frac{1}{2}}$ and $\frac{1}{\sqrt{N}}\Lambda
_{z_{t}}^{0\prime }e_{t}\overset{d}{\rightarrow }\mathcal{N}(0,\Gamma
_{z_{t}t})$ by Assumption \ref{dist}(4), we have $\sqrt{N}(\tilde{f}%
_{t}-H_{z_{t}}^{-1}f_{t}^{0})\overset{d}{\rightarrow }\mathcal{N}%
(0,V_{z_{t}}^{-\frac{1}{2}}\Upsilon _{z_{t}}^{\prime }\Sigma _{\Lambda
_{z_{t}}}^{-\frac{1}{2}}\Gamma _{z_{t}t}\Sigma _{\Lambda _{z_{t}}}^{-\frac{1%
}{2}}\Upsilon _{z_{t}}V_{z_{t}}^{-\frac{1}{2}})$.
\end{proof}

\section{Proof of Theorem \protect\ref{Q}}

\begin{proof}
First, $\tilde{Q}_{jk}=\sum\nolimits_{t=2}^{T}\tilde{p}_{tjk\left\vert
T\right. }/\sum\nolimits_{j=1}^{J^{0}}\sum\nolimits_{t=2}^{T}\tilde{p}%
_{tjk\left\vert T\right. }=\frac{1}{T-1}\sum\nolimits_{t=2}^{T}\tilde{p}%
_{tjk\left\vert T\right. }/\frac{1}{T-1}\sum\nolimits_{t=1}^{T-1}\tilde{p}%
_{tk\left\vert T\right. }$. For the denominator, by Theorem \ref{ptjT}, we
have 
\begin{equation}
\frac{1}{T-1}\sum\nolimits_{t=1}^{T-1}\tilde{p}_{tk\left\vert T\right. }=%
\frac{1}{T-1}\sum\nolimits_{t=1}^{T-1}1_{z_{t}=k}+o_{p}(\frac{1}{N^{\eta }})%
\overset{p}{\rightarrow }q_{k}^{0}.  \label{av}
\end{equation}%
For the numerator, we have 
\begin{eqnarray}
\frac{1}{T-1}\sum\nolimits_{t=2}^{T}\tilde{p}_{tjk\left\vert T\right. } &=&%
\frac{1}{T-1}\sum\nolimits_{t=2}^{T}\tilde{p}_{tj\left\vert T\right. }\Pr
(z_{t-1}=k\left\vert z_{t}=j,x_{1:T};\tilde{\Lambda},\tilde{\sigma}%
^{2},Q,\phi \right. )  \notag \\
&=&\frac{1}{T-1}\sum\nolimits_{t=2}^{T}[1_{z_{t}=j}+o_{p}(\frac{1}{N^{\eta }}%
)][1_{z_{t-1}=k}+o_{p}(\frac{1}{N^{\eta }})].  \label{at}
\end{eqnarray}%
The second equality of (\ref{at}) follows from: (1) $\tilde{p}_{tj\left\vert
T\right. }=1_{z_{t}=j}+o_{p}(\frac{1}{N^{\eta }})$ by Theorem \ref{ptjT},%
\begin{eqnarray*}
\text{(2) }\Pr (z_{t-1} &=&k\left\vert z_{t}=j,x_{1:T};\tilde{\Lambda},%
\tilde{\sigma}^{2},Q,\phi \right. )=\Pr (z_{t-1}=k\left\vert
z_{t}=j,x_{1:t-1};\tilde{\Lambda},\tilde{\sigma}^{2},Q,\phi \right. ) \\
&=&\frac{\Pr (z_{t-1}=k,z_{t}=j\left\vert x_{1:t-1};\tilde{\Lambda},\tilde{%
\sigma}^{2},Q,\phi \right. )}{\Pr (z_{t}=j\left\vert x_{1:t-1};\tilde{\Lambda%
},\tilde{\sigma}^{2},Q,\phi \right. )}=\frac{Q_{jk}\tilde{p}%
_{t-1,k\left\vert t-1\right. }}{\sum\nolimits_{h=1}^{J^{0}}Q_{jh}\tilde{p}%
_{t-1,h\left\vert t-1\right. }} \\
&=&1_{z_{t-1}=k}+o_{p}(\frac{1}{N^{\eta }}),
\end{eqnarray*}%
where the last equality follows from Theorem \ref{ptjT}. Since $z_{t}$
follows a Markov process,%
\begin{equation}
\frac{1}{T-1}\sum\nolimits_{t=2}^{T}1_{z_{t}=j}1_{z_{t-1}=k}\overset{p}{%
\rightarrow }\mathbb{E}(1_{z_{t}=j}1_{z_{t-1}=k})=\mathbb{E[E}%
(1_{z_{t}=j}1_{z_{t-1}=k}\left\vert 1_{z_{t-1}=k}\right.
)]=q_{k}^{0}Q_{jk}^{0}\text{.}  \label{au}
\end{equation}%
Take equations (\ref{av})-(\ref{au}) together, we have shown $\tilde{Q}_{jk}%
\overset{p}{\rightarrow }Q_{jk}^{0}$.
\end{proof}

\section{Proof of Theorem \textbf{\protect\ref{r0&J0}}}

\begin{proof}
First, for notational convenience, let the $\bar{J}$-dimensional vector $%
(r_{1},...,r_{\bar{J}})$ denote the numbers of factors. When the number of
regime is less than $\bar{J}$, some elements of this vector are zeros. Rank $%
(r_{1},...,r_{\bar{J}})$ in descending order and denote it as $%
(r_{(1)},...,r_{(\bar{J})})$. Similarly, let $%
(r_{(1)}^{0},...,r_{(J^{0})}^{0},0,...,0)$ denote the true numbers of
factors, then it suffices to show that $\Pr (\tilde{r}_{(j)}=r_{(j)}^{0}$
for $j=1,...,\bar{J})\rightarrow 1$.

Step (1):%
\begin{eqnarray}
\Pr (PC(r_{(1)},...,r_{(\bar{J})})
&>&PC(r_{(1)}^{0},...,r_{(J^{0})}^{0}))\rightarrow 0\text{ if }%
r_{(1)}>r_{(1)}^{0},  \label{bg} \\
\Pr (PC(r_{(1)},...,r_{(\bar{J})})
&>&PC(r_{(1)}^{0},...,r_{(J^{0})}^{0}))\rightarrow 0\text{ if }%
r_{(1)}<r_{(1)}^{0}.  \label{bh}
\end{eqnarray}

Proof of expression (\ref{bg}): From expressions (\ref{b}) and (\ref{c}),%
\begin{eqnarray*}
&&l(\tilde{\Lambda}_{1,r_{(1)}},...,\tilde{\Lambda}_{\bar{J},r_{(\bar{J}%
)}},\sigma ^{2},Q,\phi )\leq \\
&&-\frac{NT\log 2\pi }{2}-\frac{1}{2}\sum\nolimits_{t=1}^{T}\log \left\vert 
\tilde{\Lambda}_{m_{t},r_{(m_{t})}}\tilde{\Lambda}_{m_{t},r_{(m_{t})}}^{%
\prime }+\sigma ^{2}I_{N}\right\vert -\frac{\sum\nolimits_{t=1}^{T}\left%
\Vert M_{\tilde{\Lambda}_{m_{t},r_{(m_{t})}}}x_{t}\right\Vert ^{2}}{2\sigma
^{2}} \\
&&-\frac{1}{2}\sum\nolimits_{t=1}^{T}x_{t}^{\prime }\tilde{\Lambda}%
_{m_{t},r_{(m_{t})}}(\sigma ^{2}I_{r_{(m_{t})}}+\tilde{\Lambda}%
_{m_{t},r_{(m_{t})}}^{\prime }\tilde{\Lambda}_{m_{t},r_{(m_{t})}})^{-1}(%
\tilde{\Lambda}_{m_{t},r_{(m_{t})}}^{\prime }\tilde{\Lambda}%
_{m_{t},r_{(m_{t})}})^{-1}\tilde{\Lambda}_{m_{t},r_{(m_{t})}}^{\prime }x_{t},
\end{eqnarray*}%
where $\tilde{\Lambda}_{j,r_{(j)}}$ is $N\times r_{(j)}$ and the subscript $%
r_{(j)}$ is suppressed if $r_{(j)}=r_{(j)}^{0}$. If $r_{(j)}=0$, then $%
\tilde{\Lambda}_{j,r_{(j)}}=0$. From expressions (\ref{d}) and (\ref{e}), we
have%
\begin{eqnarray*}
&&l(\Lambda _{1}^{0},...,\Lambda _{J^{0}}^{0},\sigma ^{2},Q,\phi ) \\
&\geq &-\frac{NT}{2}\log 2\pi -\frac{1}{2}\sum\nolimits_{t=1}^{T}\log
\left\vert \Lambda _{z_{t}}^{0}\Lambda _{z_{t}}^{0\prime }+\sigma
^{2}I_{N}\right\vert -T\log \min_{j,k}Q_{jk}- \\
&&\frac{1}{2\sigma ^{2}}\sum\nolimits_{t=1}^{T}\left\Vert M_{\Lambda
_{z_{t}}^{0}}x_{t}\right\Vert ^{2}-\frac{1}{2}\sum\nolimits_{t=1}^{T}x_{t}^{%
\prime }\Lambda _{z_{t}}^{0}(\sigma ^{2}I_{r_{z_{t}}^{0}}+\Lambda
_{z_{t}}^{0\prime }\Lambda _{z_{t}}^{0})^{-1}(\Lambda _{z_{t}}^{0\prime
}\Lambda _{z_{t}}^{0})^{-1}\Lambda _{z_{t}}^{0\prime }x_{t}.
\end{eqnarray*}%
It follows that%
\begin{eqnarray}
&&l(\tilde{\Lambda}_{1,r_{(1)}},...,\tilde{\Lambda}_{\bar{J},r_{(\bar{J}%
)}},\sigma ^{2},Q,\phi )-l(\Lambda _{1}^{0},...,\Lambda _{J^{0}}^{0},\sigma
^{2},Q,\phi )\leq  \notag \\
&&-\frac{1}{2\sigma ^{2}}\sum\nolimits_{t=1}^{T}\left\Vert M_{\tilde{\Lambda}%
_{m_{t},r_{(m_{t})}}}\Lambda _{z_{t}}^{0}f_{t}^{0}\right\Vert ^{2}-\frac{1}{%
\sigma ^{2}}\sum\nolimits_{t=1}^{T}e_{t}^{\prime }M_{\tilde{\Lambda}%
_{m_{t},,r_{(m_{t})}}}\Lambda _{z_{t}}^{0}f_{t}^{0}  \notag \\
&&-\frac{1}{2\sigma ^{2}}\sum\nolimits_{t=1}^{T}\left\Vert P_{\Lambda
_{z_{t}}^{0}}e_{t}\right\Vert ^{2}+\frac{1}{2\sigma ^{2}}\sum%
\nolimits_{t=1}^{T}\left\Vert P_{\tilde{\Lambda}_{m_{t},r_{(m_{t})}}}e_{t}%
\right\Vert ^{2}  \notag \\
&&+T\log \min_{j,k}Q_{jk}-\frac{1}{2}\sum\nolimits_{t=1}^{T}\log \frac{%
\left\vert \tilde{\Lambda}_{m_{t},r_{(m_{t})}}\tilde{\Lambda}%
_{m_{t},r_{(m_{t})}}^{\prime }+\sigma ^{2}I_{N}\right\vert }{\left\vert
\Lambda _{z_{t}}^{0}\Lambda _{z_{t}}^{0\prime }+\sigma ^{2}I_{N}\right\vert }
\notag \\
&&-\frac{1}{2}\sum\nolimits_{t=1}^{T}x_{t}^{\prime }\tilde{\Lambda}%
_{m_{t},r_{(m_{t})}}(\sigma ^{2}I_{r_{(m_{t})}}+\tilde{\Lambda}%
_{m_{t},r_{(m_{t})}}^{\prime }\tilde{\Lambda}_{m_{t},r_{(m_{t})}})^{-1}(%
\tilde{\Lambda}_{m_{t},r_{(m_{t})}}^{\prime }\tilde{\Lambda}%
_{m_{t},r_{(m_{t})}})^{-1}\tilde{\Lambda}_{m_{t},r_{(m_{t})}}^{\prime }x_{t}
\notag \\
&&+\frac{1}{2}\sum\nolimits_{t=1}^{T}x_{t}^{\prime }\Lambda
_{z_{t}}^{0}(\sigma ^{2}I_{r_{z_{t}}^{0}}+\Lambda _{z_{t}}^{0\prime }\Lambda
_{z_{t}}^{0})^{-1}(\Lambda _{z_{t}}^{0\prime }\Lambda
_{z_{t}}^{0})^{-1}\Lambda _{z_{t}}^{0\prime }x_{t}  \label{bk}
\end{eqnarray}%
The first, the third and the seventh term on the right hand side of
expression (\ref{bk}) are negative, thus the inequality still holds when
these terms are thrown away. Steps (3.1), (3.2) and (3.4) in the proof of
Theorem \ref{consis} show that the fifth term is $O(T)$, the sixth term is $%
O_{p}(T\log N)$ and the eighth term is $O_{p}(T)$. For the fourth term, by
expression (\ref{g}) and Lemma \ref{E norm}, we have $\sum%
\nolimits_{t=1}^{T}\left\Vert P_{\tilde{\Lambda}_{m_{t},r_{(m_{t})}}}e_{t}%
\right\Vert ^{2}\leq \sum\nolimits_{j=1}^{J^{0}}r_{(j)}\left\Vert E^{\prime
}E\right\Vert =O_{p}(\frac{NT}{\delta _{NT}})$. The second term is $O_{p}(%
\frac{NT}{\delta _{NT}})$ because step (3.5) in the proof of Theorem \ref%
{consis} shows $\sum\nolimits_{t=1}^{T}e_{t}^{\prime }\Lambda
_{z_{t}}^{0}f_{t}^{0}=O_{p}(N^{\frac{1}{2}}T)$ and $\sum%
\nolimits_{t=1}^{T}e_{t}^{\prime }P_{\tilde{\Lambda}_{m_{t},r_{m_{t}}}}%
\Lambda _{z_{t}}^{0}f_{t}^{0}=O_{p}(N^{\frac{1}{2}}T^{\frac{1}{2}%
}(\sum\nolimits_{t=1}^{T}\left\Vert P_{\tilde{\Lambda}%
_{m_{t},r_{m_{t}}}}e_{t}\right\Vert ^{2})^{\frac{1}{2}})$. In summary, we
have%
\begin{equation}
l(\tilde{\Lambda}_{1,r_{(1)}},...,\tilde{\Lambda}_{\bar{J},r_{(\bar{J}%
)}},\sigma ^{2},Q,\phi )\leq l(\Lambda _{1}^{0},...,\Lambda
_{J^{0}}^{0},\sigma ^{2},Q,\phi )+O_{p}(\frac{NT}{\delta _{NT}}).  \label{bl}
\end{equation}%
Note that expression (\ref{bl}) holds no matter what the values of $%
r_{(1)},...,r_{(\bar{J})}$ are. When $r_{(1)},...,r_{(\bar{J})}$ equal the
true values, $l(\Lambda _{1}^{0},...,\Lambda _{J^{0}}^{0},\sigma ^{2},Q,\phi
)\leq l(\tilde{\Lambda}_{1},...,\tilde{\Lambda}_{J^{0}},\sigma ^{2},Q,\phi )$%
. Thus%
\begin{eqnarray*}
&&PC(r_{(1)},...,r_{(\bar{J})})-PC(r_{(1)}^{0},...,r_{(J^{0})}^{0}) \\
&\leq &O_{p}(\frac{1}{\delta _{NT}})-\sum\nolimits_{j=1}^{\bar{J}%
}(g(N,T))^{b(r_{(j)})}+\sum\nolimits_{j=1}^{J^{0}}(g(N,T))^{b(r_{(j)}^{0})}.
\end{eqnarray*}%
If $r_{(1)}>r_{(1)}^{0}$, $\Pr (PC(r_{(1)},...,r_{(\bar{J})})\leq
PC(r_{(1)}^{0},...,r_{(J^{0})}^{0}))\rightarrow 1$, i.e., $%
(g(N,T))^{b(r_{(1)})}$ would dominate if $r_{(1)}>r_{(1)}^{0}$, no matter
what the values of $r_{(2)},...,r_{(\bar{J})}$ are.

Proof of expression (\ref{bh}): Since $l(\Lambda _{1}^{0},...,\Lambda
_{J^{0}}^{0},\sigma ^{2},Q,\phi )\leq l(\tilde{\Lambda}_{1},...,\tilde{%
\Lambda}_{J^{0}},\sigma ^{2},Q,\phi )$,%
\begin{eqnarray}
&&l(\tilde{\Lambda}_{1,r_{(1)}},...,\tilde{\Lambda}_{\bar{J},r_{(\bar{J}%
)}},\sigma ^{2},Q,\phi )-l(\tilde{\Lambda}_{1},...,\tilde{\Lambda}%
_{J^{0}},\sigma ^{2},Q,\phi )  \notag \\
&\leq &l(\tilde{\Lambda}_{1,r_{(1)}},...,\tilde{\Lambda}_{\bar{J},r_{(\bar{J}%
)}},\sigma ^{2},Q,\phi )-l(\Lambda _{1}^{0},...,\Lambda _{J^{0}}^{0},\sigma
^{2},Q,\phi )  \notag \\
&\leq &\text{the right hand side of expression (\ref{bk})}  \notag \\
&\leq &-\frac{1}{2\sigma ^{2}}\sum\nolimits_{t=1}^{T}\left\Vert M_{\tilde{%
\Lambda}_{m_{t},r_{(m_{t})}}}\Lambda _{z_{t}}^{0}f_{t}^{0}\right\Vert
^{2}+O_{p}(\frac{NT}{\delta _{NT}})  \notag \\
&\leq &-\frac{1}{2\sigma ^{2}}\sum\nolimits_{t:z_{t}=(1)}\left\Vert M_{%
\tilde{\Lambda}_{m_{t},r_{(m_{t})}}}\Lambda _{(1)}^{0}f_{t}^{0}\right\Vert
^{2}+O_{p}(\frac{NT}{\delta _{NT}}),  \label{bm}
\end{eqnarray}%
where $\sum\nolimits_{t:z_{t}=(1)}$ denotes summation over the regime
corresponding to $r_{(1)}^{0}$, and there are $q_{(1)}^{0}T$ terms in this
summation since $q_{(1)}^{0}$ is the unconditional probability of $z_{t}=(1)$%
. For each $t$ with $z_{t}=(1)$, $\Lambda _{(1)}^{0}f_{t}^{0}$ is projected
on one of $\tilde{\Lambda}_{j,r_{(j)}},j=1,...,\bar{J}$, thus there exists $%
\tilde{\Lambda}_{k,r_{(k)}}$ such that $\Lambda _{(1)}^{0}f_{t}^{0}$ is
projected on $\tilde{\Lambda}_{k,r_{(k)}}$ at least $\frac{q_{(1)}^{0}T}{%
\bar{J}}$ times, i.e., $\sum\nolimits_{t=1}^{T}1_{m_{t}=(k)}1_{z_{t}=(1)}%
\geq \frac{q_{(1)}^{0}T}{\bar{J}}$. Thus by Assumption \ref{factors}(1), $%
\rho _{\min }(\frac{\sum\nolimits_{t=1}^{T}f_{t}^{0}f_{t}^{0\prime
}1_{m_{t}=(k)}1_{z_{t}=(1)}}{\sum%
\nolimits_{t=1}^{T}1_{m_{t}=(k)}1_{z_{t}=(1)}})\geq c$ for some $c>0$
w.p.a.1. It follows that%
\begin{eqnarray}
\sum\nolimits_{t:z_{t}=(1)}\left\Vert M_{\tilde{\Lambda}%
_{m_{t},r_{(m_{t})}}}\Lambda _{(1)}^{0}f_{t}^{0}\right\Vert ^{2} &\geq
&\sum\nolimits_{t=1}^{T}\left\Vert M_{\tilde{\Lambda}_{k,r_{(k)}}}\Lambda
_{(1)}^{0}f_{t}^{0}\right\Vert ^{2}1_{m_{t}=(k)}1_{z_{t}=(1)}  \notag \\
&=&tr(\Lambda _{(1)}^{0\prime }M_{\tilde{\Lambda}_{k,r_{(k)}}}\Lambda
_{(1)}^{0}\sum\nolimits_{t=1}^{T}f_{t}^{0}f_{t}^{0\prime
}1_{m_{t}=(k)}1_{z_{t}=(1)})  \notag \\
&\geq &tr(\Lambda _{(1)}^{0\prime }M_{\tilde{\Lambda}_{k,r_{(k)}}}\Lambda
_{(1)}^{0})\rho _{\min }(\sum\nolimits_{t=1}^{T}f_{t}^{0}f_{t}^{0\prime
}1_{m_{t}=(k)}1_{z_{t}=(1)})  \notag \\
&\geq &tr(\Lambda _{(1)}^{0\prime }M_{\tilde{\Lambda}_{k,r_{(k)}}}\Lambda
_{(1)}^{0})\frac{q_{(1)}^{0}T}{\bar{J}}c\text{ w.p.a.1.}  \label{bn}
\end{eqnarray}%
If $r_{(1)}<r_{(1)}^{0}$, then $r_{(k)}\leq r_{(1)}<r_{(1)}^{0}$, and
consequently$\frac{1}{N}tr(\Lambda _{(1)}^{0\prime }M_{\tilde{\Lambda}%
_{k,r_{(k)}}}\Lambda _{(1)}^{0})>c$ for some $c>0$. Since $%
(g(N,T))^{b(r_{(j)})}\rightarrow 0$ for $1\leq r_{(j)}\leq \bar{r}$, we have 
$\Pr (PC(r_{(1)},...,r_{(\bar{J})})>PC(r_{(1)}^{0},...,r_{(J^{0})}^{0}))%
\rightarrow 0$ if $r_{(1)}<r_{(1)}^{0}$.

Step (2):%
\begin{eqnarray}
\Pr (PC(r_{(1)}^{0},r_{(2)},...,r_{(\bar{J})})
&>&PC(r_{(1)}^{0},...,r_{(J^{0})}^{0}))\rightarrow 0\text{ if }%
r_{(2)}>r_{(2)}^{0}\text{,}  \label{bi} \\
\Pr (PC(r_{(1)}^{0},r_{(2)},...,r_{(\bar{J})})
&>&PC(r_{(1)}^{0},...,r_{(J^{0})}^{0}))\rightarrow 0\text{ if }%
r_{(2)}<r_{(2)}^{0}\text{.}  \label{bj}
\end{eqnarray}

Proof of expression (\ref{bi}): In the proof of expression (\ref{bg}) we
have shown that expression (\ref{bl}) holds no matter what the values of $%
r_{(1)},...,r_{(\bar{J})}$ are, thus (\ref{bl}) also holds here for $%
(r_{(1)}^{0},r_{(2)},...,r_{(\bar{J})})$. It follows that%
\begin{eqnarray*}
&&PC(r_{(1)}^{0},r_{(2)},...,r_{(\bar{J}%
)})-PC(r_{(1)}^{0},...,r_{(J^{0})}^{0}) \\
&\leq &O_{p}(\frac{1}{\delta _{NT}})-\sum\nolimits_{j=2}^{\bar{J}%
}(g(N,T))^{b(r_{(j)})}+\sum\nolimits_{j=2}^{J^{0}}(g(N,T))^{b(r_{(j)}^{0})}.
\end{eqnarray*}%
If $r_{(2)}>r_{(2)}^{0}$, $\Pr (PC(r_{(1)}^{0},r_{(2)},...,r_{(\bar{J}%
)})\leq PC(r_{(1)}^{0},...,r_{(J^{0})}^{0}))\rightarrow 1$, i.e., $%
(g(N,T))^{b(r_{(2)})}$ would dominate if $r_{(2)}>r_{(2)}^{0}$, no matter
what the values of $r_{(3)},...,r_{(\bar{J})}$ are.

Proof of expression (\ref{bj}): Similar to expressions (\ref{bm}) and (\ref%
{bn}), we have%
\begin{eqnarray*}
&&l(\tilde{\Lambda}_{1,r_{(1)}^{0}},\tilde{\Lambda}_{2,r_{(2)}},...,\tilde{%
\Lambda}_{\bar{J},r_{(\bar{J})}},\sigma ^{2},Q,\phi )-l(\tilde{\Lambda}%
_{1},...,\tilde{\Lambda}_{J^{0}},\sigma ^{2},Q,\phi ) \\
&\leq &-\frac{1}{2\sigma ^{2}}\sum\nolimits_{t:z_{t}=(1)}\left\Vert M_{%
\tilde{\Lambda}_{m_{t},r_{(m_{t})}}}\Lambda _{(1)}^{0}f_{t}^{0}\right\Vert
^{2}+O_{p}(\frac{NT}{\delta _{NT}}) \\
&\leq &-\frac{1}{2\sigma ^{2}}tr(\Lambda _{(1)}^{0\prime }M_{\tilde{\Lambda}%
_{k,r_{(k)}}}\Lambda _{(1)}^{0})\rho _{\min
}(\sum\nolimits_{t=1}^{T}f_{t}^{0}f_{t}^{0\prime
}1_{m_{t}=(k)}1_{z_{t}=(1)})+O_{p}(\frac{NT}{\delta _{NT}}),
\end{eqnarray*}%
and $\sum\nolimits_{t=1}^{T}1_{m_{t}=(k)}1_{z_{t}=(1)}\geq \frac{q_{(1)}^{0}T%
}{\bar{J}}$. The event $PC(r_{(1)}^{0},r_{(2)},...,r_{(\bar{J}%
)})>PC(r_{(1)}^{0},...,r_{(J^{0})}^{0})$ implies%
\begin{equation*}
\frac{l(\tilde{\Lambda}_{1,r_{(1)}^{0}},\tilde{\Lambda}_{2,r_{(2)}},...,%
\tilde{\Lambda}_{\bar{J},r_{(\bar{J})}},\sigma ^{2},Q,\phi )-l(\tilde{\Lambda%
}_{1},...,\tilde{\Lambda}_{J^{0}},\sigma ^{2},Q,\phi )}{NT}%
+\sum\nolimits_{j=2}^{J^{0}}(g(N,T))^{b(r_{(j)}^{0})}>0,
\end{equation*}%
i.e.,%
\begin{equation*}
tr(\frac{1}{N}\Lambda _{(1)}^{0\prime }M_{\tilde{\Lambda}_{k,r_{(k)}}}%
\Lambda _{(1)}^{0})<\frac{2\sigma ^{2}(O_{p}(\frac{1}{\delta _{NT}}%
)+\sum\nolimits_{j=2}^{J^{0}}(g(N,T))^{b(r_{(j)}^{0})})}{\rho _{\min }(\frac{%
1}{T}\sum\nolimits_{t=1}^{T}f_{t}^{0}f_{t}^{0\prime
}1_{m_{t}=(k)}1_{z_{t}=(1)})}.
\end{equation*}%
If $r_{(2)}<r_{(2)}^{0}$, then $r_{(k)}<r_{(2)}^{0}$ for $k=2,...,\bar{J}$.
Thus if $r_{(2)}<r_{(2)}^{0}$, 
\begin{equation*}
\Pr (PC(r_{(1)}^{0},r_{(2)},...,r_{(\bar{J}%
)})>PC(r_{(1)}^{0},...,r_{(J^{0})}^{0}),k\neq 1)\rightarrow 0.
\end{equation*}%
Similarly, the event $PC(r_{(1)}^{0},r_{(2)},...,r_{(\bar{J}%
)})>PC(r_{(1)}^{0},...,r_{(J^{0})}^{0})$ also implies that%
\begin{equation*}
tr(\frac{1}{N}\Lambda _{(2)}^{0\prime }M_{\tilde{\Lambda}_{k^{\prime
},r_{(k^{\prime })}}}\Lambda _{(2)}^{0})<\frac{2\sigma ^{2}(O_{p}(\frac{1}{%
\delta _{NT}})+\sum\nolimits_{j=2}^{J^{0}}(g(N,T))^{b(r_{(j)}^{0})})}{\rho
_{\min }(\frac{1}{T}\sum\nolimits_{t=1}^{T}f_{t}^{0}f_{t}^{0\prime
}1_{m_{t}=(k^{\prime })}1_{z_{t}=(2)})},
\end{equation*}%
with $\sum\nolimits_{t=1}^{T}1_{m_{t}=(k)}1_{z_{t}=(1)}\geq \frac{%
q_{(1)}^{0}T}{\bar{J}}$. If $r_{(2)}<r_{(2)}^{0}$, then $r_{(k^{\prime
})}<r_{(2)}^{0}$ for $k^{\prime }=2,...,\bar{J}$, thus 
\begin{equation*}
\Pr (PC(r_{(1)}^{0},r_{(2)},...,r_{(\bar{J}%
)})>PC(r_{(1)}^{0},...,r_{(J^{0})}^{0}),k^{\prime }\neq 1)\rightarrow 0.
\end{equation*}%
Finally, the event $(PC(r_{(1)}^{0},r_{(2)},...,r_{(\bar{J}%
)})>PC(r_{(1)}^{0},...,r_{(J^{0})}^{0}),k=1,k^{\prime }=1)$ implies 
\begin{eqnarray}
tr(\frac{1}{N}\Lambda _{(1)}^{0\prime }M_{\tilde{\Lambda}_{1,r_{(1)}^{0}}}%
\Lambda _{(1)}^{0}) &<&\frac{2\sigma ^{2}(O_{p}(\frac{1}{\delta _{NT}}%
)+\sum\nolimits_{j=2}^{J^{0}}(g(N,T))^{b(r_{(j)}^{0})})}{\rho _{\min }(\frac{%
1}{T}\sum\nolimits_{t=1}^{T}f_{t}^{0}f_{t}^{0\prime
}1_{m_{t}=(1)}1_{z_{t}=(1)})},  \label{bp} \\
tr(\frac{1}{N}\Lambda _{(2)}^{0\prime }M_{\tilde{\Lambda}_{1,r_{(1)}^{0}}}%
\Lambda _{(2)}^{0}) &<&\frac{2\sigma ^{2}(O_{p}(\frac{1}{\delta _{NT}}%
)+\sum\nolimits_{j=2}^{J^{0}}(g(N,T))^{b(r_{(j)}^{0})})}{\rho _{\min }(\frac{%
1}{T}\sum\nolimits_{t=1}^{T}f_{t}^{0}f_{t}^{0\prime
}1_{m_{t}=(1)}1_{z_{t}=(2)})},  \label{bq}
\end{eqnarray}%
with $\sum\nolimits_{t=1}^{T}1_{m_{t}=(1)}1_{z_{t}=(1)}\geq \frac{%
q_{(1)}^{0}T}{\bar{J}}$ and $\sum%
\nolimits_{t=1}^{T}1_{m_{t}=(1)}1_{z_{t}=(2)}\geq \frac{q_{(2)}^{0}T}{\bar{J}%
}$. From expression (\ref{bo}), we have 
\begin{equation}
\left\Vert P_{\Lambda _{(1)}^{0}}-P_{\tilde{\Lambda}_{1,r_{(1)}^{0}}}\right%
\Vert ^{2}\leq 2tr(\frac{1}{N}\Lambda _{(1)}^{0\prime }M_{\tilde{\Lambda}%
_{1,r_{(1)}^{0}}}\Lambda _{(1)}^{0})\left\Vert (\frac{1}{N}\Lambda
_{(1)}^{0\prime }\Lambda _{(1)}^{0})^{-\frac{1}{2}}\right\Vert _{F}^{2}, 
\notag
\end{equation}%
and it follows that%
\begin{eqnarray}
&&tr(\frac{1}{N}\Lambda _{(2)}^{0\prime }M_{\tilde{\Lambda}%
_{1,r_{(1)}^{0}}}\Lambda _{(2)}^{0})  \notag \\
&>&tr(\frac{1}{N}\Lambda _{(2)}^{0\prime }M_{\Lambda _{(1)}^{0}}\Lambda
_{(2)}^{0})-\left\vert tr(\frac{1}{N}\Lambda _{(2)}^{0\prime }(P_{\Lambda
_{(1)}^{0}}-P_{\tilde{\Lambda}_{1,r_{(1)}^{0}}})\Lambda
_{(2)}^{0})\right\vert  \notag \\
&\geq &tr(\frac{1}{N}\Lambda _{(2)}^{0\prime }M_{\Lambda _{(1)}^{0}}\Lambda
_{(2)}^{0})-\sqrt{2tr(\frac{1}{N}\Lambda _{(1)}^{0\prime }M_{\tilde{\Lambda}%
_{1,r_{(1)}^{0}}}\Lambda _{(1)}^{0})}\left\Vert (\frac{\Lambda
_{(1)}^{0\prime }\Lambda _{(1)}^{0}}{N})^{-\frac{1}{2}}\right\Vert
_{F}\left\Vert \frac{\Lambda _{(2)}^{0}}{\sqrt{N}}\right\Vert _{F}^{2}
\label{br}
\end{eqnarray}%
Expressions (\ref{bp})-(\ref{br}) together imply that $tr(\frac{1}{N}\Lambda
_{(2)}^{0\prime }M_{\Lambda _{(1)}^{0}}\Lambda _{(2)}^{0})\leq o_{p}(1)$.
Since we assume $\lim \frac{1}{N}\Lambda _{(2)}^{0\prime }M_{\Lambda
_{(1)}^{0}}\Lambda _{(2)}^{0}\neq 0$, we have $\Pr (tr(\frac{1}{N}\Lambda
_{(2)}^{0\prime }M_{\Lambda _{(1)}^{0}}\Lambda _{(2)}^{0})\leq
o_{p}(1))\rightarrow 0$, thus%
\begin{equation*}
\Pr (PC(r_{(1)}^{0},r_{(2)},...,r_{(\bar{J}%
)})>PC(r_{(1)}^{0},...,r_{(J^{0})}^{0}),k=1,k^{\prime }\neq 1)\rightarrow 0.
\end{equation*}%
In summary, we have proved expression (\ref{bj}). Similarly, we can continue
to prove that for $j=3,...,\bar{J}$, $\Pr
(PC(r_{(1)}^{0},...r_{(j-1)}^{0},r_{(j)},...,r_{(\bar{J}%
)})>PC(r_{(1)}^{0},...,r_{(J^{0})}^{0}))\rightarrow 0$ if $r_{(j)}\neq
r_{(j)}^{0}$.
\end{proof}

\section{Details on First Order Conditions\label{secfoc}}

\begin{description}
\item[First order condition of $\protect\sigma ^{2}$:] 
\begin{eqnarray*}
&&\frac{\partial \sum\nolimits_{t=1}^{T}\sum\nolimits_{j=1}^{J^{0}}\log
L(x_{t}\left\vert z_{t}=j;\Lambda _{j},\sigma ^{2}\right. )p_{tj\left\vert
T\right. }}{\partial \sigma ^{2}} \\
&=&\sum\nolimits_{t=1}^{T}\sum\nolimits_{j=1}^{J^{0}}p_{tj\left\vert
T\right. }\frac{\partial (-\frac{1}{2}\log \left\vert \Sigma _{j}\right\vert
-\frac{1}{2}x_{t}^{\prime }\Sigma _{j}^{-1}x_{t})}{\partial \sigma ^{2}} \\
&=&-\frac{1}{2}\sum\nolimits_{t=1}^{T}\sum\nolimits_{j=1}^{J^{0}}p_{tj\left%
\vert T\right. }tr(\Sigma _{j}^{-1}-\Sigma _{j}^{-1}x_{t}x_{t}^{\prime
}\Sigma _{j}^{-1}) \\
&=&-\frac{1}{2}\sum\nolimits_{j=1}^{J^{0}}tr(\sum\nolimits_{t=1}^{T}p_{tj%
\left\vert T\right. }\Sigma _{j}^{-1}-\Sigma
_{j}^{-1}\sum\nolimits_{t=1}^{T}p_{tj\left\vert T\right. }x_{t}x_{t}^{\prime
}\Sigma _{j}^{-1}) \\
&=&-\frac{1}{2}\sum\nolimits_{j=1}^{J^{0}}(\sum\nolimits_{t=1}^{T}p_{tj\left%
\vert T\right. })tr(\Sigma _{j}^{-1}-\Sigma _{j}^{-1}S_{j}\Sigma _{j}^{-1})
\\
&=&-\frac{1}{2\sigma ^{4}}\sum\nolimits_{j=1}^{J^{0}}(\sum%
\nolimits_{t=1}^{T}p_{tj\left\vert T\right. })tr(\Sigma _{j}-S_{j}) \\
&=&-\frac{1}{2\sigma ^{4}}tr(\sum\nolimits_{j=1}^{J^{0}}\sum%
\nolimits_{t=1}^{T}p_{tj\left\vert T\right. }\Lambda _{j}\Lambda
_{j}^{\prime }+T\sigma ^{2}I_{N}-\sum\nolimits_{t=1}^{T}x_{t}x_{t}^{\prime
})=0.
\end{eqnarray*}
\end{description}

The second equality is due to 
\begin{eqnarray}
\frac{\partial \log \left\vert \Sigma _{j}\right\vert }{\partial \sigma ^{2}}
&=&tr(\Sigma _{j}^{-1}),  \label{ak} \\
\frac{\partial x_{t}^{\prime }\Sigma _{j}^{-1}x_{t}}{\partial \sigma ^{2}}
&=&-tr(\Sigma _{j}^{-1}x_{t}x_{t}^{\prime }\Sigma _{j}^{-1}),  \label{al}
\end{eqnarray}%
and the fifth equality is due to%
\begin{equation}
\Sigma _{j}(\Sigma _{j}-S_{j})\Sigma _{j}=(\Lambda _{j}\Lambda _{j}^{\prime
}+\sigma ^{2}I_{N})(\Sigma _{j}-S_{j})(\Lambda _{j}\Lambda _{j}^{\prime
}+\sigma ^{2}I_{N})=\sigma ^{4}(\Sigma _{j}-S_{j}),  \label{am}
\end{equation}%
since $(\Sigma _{j}-S_{j})\Lambda _{j}\Lambda _{j}^{\prime }=(\Lambda
_{j}\Lambda _{j}^{\prime }+\sigma ^{2}I_{N}-S_{j})\Lambda _{j}\Lambda
_{j}^{\prime }=0$ by equation (\ref{foc1}). Thus we have $\sigma ^{2}=\frac{1%
}{N}tr(\frac{1}{T}\sum\nolimits_{t=1}^{T}x_{t}x_{t}^{\prime
}-\sum\nolimits_{j=1}^{J^{0}}\frac{1}{T}\sum\nolimits_{t=1}^{T}p_{tj\left%
\vert T\right. }\Lambda _{j}\Lambda _{j}^{\prime })$. The proof of
expression (\ref{ao}) is the same, with $p_{tj\left\vert T\right. }$
replaced by $\tilde{p}_{tj\left\vert T\right. }^{(h)}$ and $S_{j}$ replaced
by $\tilde{S}_{j}^{(h)}$.

\begin{description}
\item[First order condition of $\Lambda _{j}$ and $\Sigma _{e}$:] 
\end{description}

When equation (\ref{aq}) is replaced by $\Sigma _{j}=\Lambda _{j}\Lambda
_{j}^{\prime }+\Sigma _{e}$, equations (\ref{bw}) and (\ref{a}) are still
valid, i.e., $\Sigma _{j}^{-1}\Lambda _{j}=\Sigma _{j}^{-1}S_{j}\Sigma
_{j}^{-1}\Lambda _{j}$. Right multiple $\Sigma _{j}$ by $\Sigma
_{e}^{-1}\Lambda _{j}$, we have $\Sigma _{j}\Sigma _{e}^{-1}\Lambda
_{j}=\Lambda _{j}(\Lambda _{j}^{\prime }\Sigma _{e}^{-1}\Lambda
_{j}+I_{r_{j}^{0}})$. Left multiply $S_{j}\Sigma _{j}^{-1}$ on both sides of
this equation, we have $S_{j}\Sigma _{e}^{-1}\Lambda _{j}=S_{j}\Sigma
_{j}^{-1}\Lambda _{j}(\Lambda _{j}^{\prime }\Sigma _{e}^{-1}\Lambda
_{j}+I_{r_{j}^{0}})$. From equation (\ref{a}), we have $\Lambda
_{j}=S_{j}\Sigma _{j}^{-1}\Lambda _{j}$, thus $S_{j}\Sigma _{e}^{-1}\Lambda
_{j}=\Lambda _{j}(\Lambda _{j}^{\prime }\Sigma _{e}^{-1}\Lambda
_{j}+I_{r_{j}^{0}})$, i.e., $\Sigma _{e}^{-\frac{1}{2}}\Lambda _{j}$ is the
eigenvectors of $\Sigma _{e}^{-\frac{1}{2}}S_{j}\Sigma _{e}^{-\frac{1}{2}}$
and $\Lambda _{j}^{\prime }\Sigma _{e}^{-1}\Lambda _{j}+I_{r_{j}^{0}}$ is
corresponding eigenvalues.%
\begin{eqnarray*}
&&\frac{\partial \sum\nolimits_{t=1}^{T}\sum\nolimits_{j=1}^{J^{0}}\log
L(x_{t}\left\vert z_{t}=j;\Lambda _{j},\sigma ^{2}\right. )p_{tj\left\vert
T\right. }}{\partial diag(\Sigma _{e})} \\
&=&\sum\nolimits_{t=1}^{T}\sum\nolimits_{j=1}^{J^{0}}p_{tj\left\vert
T\right. }\frac{\partial (-\frac{1}{2}\log \left\vert \Sigma _{j}\right\vert
-\frac{1}{2}x_{t}^{\prime }\Sigma _{j}^{-1}x_{t})}{\partial diag(\Sigma _{e})%
} \\
&=&-\frac{1}{2}\sum\nolimits_{t=1}^{T}\sum\nolimits_{j=1}^{J^{0}}p_{tj\left%
\vert T\right. }diag(\Sigma _{j}^{-1}-\Sigma _{j}^{-1}x_{t}x_{t}^{\prime
}\Sigma _{j}^{-1}) \\
&=&-\frac{1}{2}\sum\nolimits_{j=1}^{J^{0}}diag(\sum\nolimits_{t=1}^{T}p_{tj%
\left\vert T\right. }\Sigma _{j}^{-1}-\Sigma
_{j}^{-1}\sum\nolimits_{t=1}^{T}p_{tj\left\vert T\right. }x_{t}x_{t}^{\prime
}\Sigma _{j}^{-1}) \\
&=&-\frac{1}{2}\sum\nolimits_{j=1}^{J^{0}}(\sum\nolimits_{t=1}^{T}p_{tj\left%
\vert T\right. })diag(\Sigma _{j}^{-1}-\Sigma _{j}^{-1}S_{j}\Sigma _{j}^{-1})
\\
&=&\Sigma _{e}^{-1}[-\frac{1}{2}\sum\nolimits_{j=1}^{J^{0}}(\sum%
\nolimits_{t=1}^{T}p_{tj\left\vert T\right. })diag(\Sigma _{j}-S_{j})]\Sigma
_{e}^{-1} \\
&=&\Sigma _{e}^{-1}[-\frac{1}{2}diag(\sum\nolimits_{j=1}^{J^{0}}\sum%
\nolimits_{t=1}^{T}p_{tj\left\vert T\right. }\Lambda _{j}\Lambda
_{j}^{\prime }+T\Sigma _{e}-\sum\nolimits_{t=1}^{T}x_{t}x_{t}^{\prime
})]\Sigma _{e}^{-1}=0.
\end{eqnarray*}%
The second equality is due to 
\begin{eqnarray}
\frac{\partial \log \left\vert \Sigma _{j}\right\vert }{\partial diag(\Sigma
_{e})} &=&diag(\Sigma _{j}^{-1}), \\
\frac{\partial x_{t}^{\prime }\Sigma _{j}^{-1}x_{t}}{\partial diag(\Sigma
_{e})} &=&-diag(\Sigma _{j}^{-1}x_{t}x_{t}^{\prime }\Sigma _{j}^{-1}),
\end{eqnarray}%
and the fifth equality is due to%
\begin{equation}
\Sigma _{j}\Sigma _{e}^{-1}(\Sigma _{j}-S_{j})\Sigma _{e}^{-1}\Sigma
_{j}=(\Lambda _{j}\Lambda _{j}^{\prime }+\Sigma _{e})\Sigma _{e}^{-1}(\Sigma
_{j}-S_{j})\Sigma _{e}^{-1}(\Lambda _{j}\Lambda _{j}^{\prime }+\Sigma
_{e})=(\Sigma _{j}-S_{j}),
\end{equation}%
since $(\Sigma _{j}-S_{j})\Sigma _{e}^{-1}\Lambda _{j}\Lambda _{j}^{\prime
}=0$, which follows from $S_{j}\Sigma _{e}^{-1}\Lambda _{j}=\Lambda
_{j}(\Lambda _{j}^{\prime }\Sigma _{e}^{-1}\Lambda _{j}+I_{r_{j}^{0}})$.
Thus we have $\Sigma _{e}=diag(\frac{1}{T}\sum%
\nolimits_{t=1}^{T}x_{t}x_{t}^{\prime }-\sum\nolimits_{j=1}^{J^{0}}\frac{1}{T%
}\sum\nolimits_{t=1}^{T}p_{tj\left\vert T\right. }\Lambda _{j}\Lambda
_{j}^{\prime })$.

\begin{description}
\item[First order condition of $Q$:] 
\end{description}

Since $\sum\nolimits_{j=1}^{J^{0}}Q_{jk}=1$, the Lagrangean is $%
\sum\nolimits_{t=2}^{T}\sum\nolimits_{j=1}^{J^{0}}\sum%
\nolimits_{k=1}^{J^{0}}\log Q_{jk}p_{tjk\left\vert T\right.
}+\sum\nolimits_{k=1}^{J^{0}}w_{k}(1-Q_{1k}-Q_{2k}-...-Q_{J^{0}k})$. The
first order derivative of the Lagrangean with respect to $Q_{jk}$ is $\frac{1%
}{Q_{jk}}\sum\nolimits_{t=2}^{T}p_{tjk\left\vert T\right. }-w_{k}$. Set it
to zero, we have $\sum\nolimits_{t=2}^{T}p_{tjk\left\vert T\right.
}=Q_{jk}w_{k}$. Take sum over $j$, we have $\sum\nolimits_{j=1}^{J^{0}}\sum%
\nolimits_{t=2}^{T}p_{tjk\left\vert T\right.
}=\sum\nolimits_{j=1}^{J^{0}}Q_{jk}w_{k}=w_{k}$. Thus $Q_{jk}=\sum%
\nolimits_{t=2}^{T}p_{tjk\left\vert T\right.
}/\sum\nolimits_{j=1}^{J^{0}}\sum\nolimits_{t=2}^{T}p_{tjk\left\vert
T\right. }$.

\begin{description}
\item[First order condition of $\protect\phi $:] 
\end{description}

Since $\sum\nolimits_{k=1}^{J^{0}}\phi _{k}=1$, the Lagrangean is $%
\sum\nolimits_{k=1}^{J^{0}}\log \phi _{k}p_{1k\left\vert T\right. }+w(1-\phi
_{1}-\phi _{2}-...-\phi _{J^{0}})$. The first order derivative of the
Lagrangean with respect to $\phi _{k}$ is $\frac{1}{\phi _{k}}%
p_{1k\left\vert T\right. }-w$. Set it to zero, we have $p_{1k\left\vert
T\right. }=\phi _{k}w$. Take sum over $k$, we have $1=\sum%
\nolimits_{k=1}^{J^{0}}p_{1k\left\vert T\right.
}=\sum\nolimits_{k=1}^{J^{0}}\phi _{k}w=w$, thus $\phi _{k}=p_{1k\left\vert
T\right. }=\sum\nolimits_{j=1}^{J^{0}}p_{2jk\left\vert T\right. }$.

\section{Smoother Algorithm for $\tilde{p}_{tjk\left\vert T\right. }^{(h)}$ 
\label{ptjkT}}

Step (1): Calculate conditional likelihoods $L(x_{t}\left\vert x_{1:t-1};%
\tilde{\theta}^{(h)}\right. )$ and filtered estimates $\tilde{p}%
_{tjk\left\vert t\right. }^{(h)}$ for $t=2,...,T$.%
\begin{eqnarray*}
\tilde{p}_{tjk\left\vert t\right. }^{(h)} &=&\Pr
(z_{t}=j,z_{t-1}=k\left\vert x_{1:t};\tilde{\theta}^{(h)}\right.
)=L(x_{t}\left\vert z_{t}=j;\tilde{\Lambda}^{(h)},\tilde{\sigma}%
^{2(h)}\right. )) \\
\times \Pr (z_{t} &=&j\left\vert z_{t-1}=k;Q\right. )\Pr
(z_{t-1}=k\left\vert x_{1:t-1};\tilde{\theta}^{(h)}\right.
)/L(x_{t}\left\vert x_{1:t-1};\tilde{\theta}^{(h)}\right. ),
\end{eqnarray*}%
where $\Pr (z_{1}=k\left\vert x_{1};\tilde{\theta}^{(h)}\right. )=\frac{%
L(x_{1}\left\vert z_{1}=k;\tilde{\Lambda}^{(h)},\tilde{\sigma}^{2(h)}\right.
)\phi _{k}}{\sum\nolimits_{j=1}^{J^{0}}L(x_{1}\left\vert z_{1}=j;\tilde{%
\Lambda}^{(h)},\tilde{\sigma}^{2(h)}\right. )\phi _{j}}$ and $\Pr
(z_{t-1}=k\left\vert x_{1:t-1};\tilde{\theta}^{(h)}\right.
)=\sum\nolimits_{z_{t-2}=1}^{J^{0}}\Pr (z_{t-1}=k,z_{t-2}\left\vert
x_{1:t-1};\tilde{\theta}^{(h)}\right. )$. The denominator $L(x_{t}\left\vert
x_{1:t-1};\tilde{\theta}^{(h)}\right. )$ equals the sum of the numerator
with respect to $z_{t}$ and $z_{t-1}$.

Step (2): Fix down $z_{t}=j,z_{t-1}=k$, for all $z_{t+1}$,%
\begin{eqnarray*}
\Pr (z_{t+1},z_{t} &=&j,z_{t-1}=k\left\vert x_{1:t+1};\tilde{\theta}%
^{(h)}\right. )=L(x_{t+1}\left\vert z_{t+1};\tilde{\Lambda}^{(h)},\tilde{%
\sigma}^{2(h)}\right. ))\Pr (z_{t+1}\left\vert z_{t}=j;Q\right. ) \\
\times \Pr (z_{t} &=&j,z_{t-1}=k\left\vert x_{1:t};\tilde{\theta}%
^{(h)}\right. )/L(x_{t+1}\left\vert x_{1:t};\tilde{\theta}^{(h)}\right. ),
\end{eqnarray*}%
for all $z_{t+1}$ and $z_{t+2}$, 
\begin{eqnarray*}
\Pr (z_{t+2},z_{t+1},z_{t} &=&j,z_{t-1}=k\left\vert x_{1:t+2};\tilde{\theta}%
^{(h)}\right. )=L(x_{t+2}\left\vert z_{t+2};\tilde{\Lambda}^{(h)},\tilde{%
\sigma}^{2(h)}\right. ))\Pr (z_{t+2}\left\vert z_{t+1};Q\right. ) \\
\times \Pr (z_{t+1},z_{t} &=&j,z_{t-1}=k\left\vert x_{1:t+1};\tilde{\theta}%
^{(h)}\right. )/L(x_{t+2}\left\vert x_{1:t+1};\tilde{\theta}^{(h)}\right. ),
\end{eqnarray*}%
and for $\tau =t+3,...,T$, for all $z_{\tau }$ and $z_{\tau -1}$, 
\begin{eqnarray*}
\Pr (z_{\tau },z_{\tau -1},z_{t} &=&j,z_{t-1}=k\left\vert x_{1:\tau };\tilde{%
\theta}^{(h)}\right. )=L(x_{\tau }\left\vert z_{\tau };\tilde{\Lambda}^{(h)},%
\tilde{\sigma}^{2(h)}\right. ))\Pr (z_{\tau }\left\vert z_{\tau -1};Q\right.
) \\
\times \Pr (z_{\tau -1},z_{t} &=&j,z_{t-1}=k\left\vert x_{1:\tau -1};\tilde{%
\theta}^{(h)}\right. )/L(x_{\tau }\left\vert x_{1:\tau -1};\tilde{\theta}%
^{(h)}\right. ),
\end{eqnarray*}%
where $\Pr (z_{\tau -1},z_{t}=j,z_{t-1}=k\left\vert x_{1:\tau -1};\tilde{%
\theta}^{(h)}\right. )=\sum\nolimits_{z_{\tau -2}=1}^{J^{0}}\Pr (z_{\tau
-1},z_{\tau -2},z_{t}=j,z_{t-1}=k\left\vert x_{1:\tau -1};\tilde{\theta}%
^{(h)}\right. )$.

Step (3): Calculate $\tilde{p}_{tjk\left\vert T\right.
}^{(h)}=\sum\nolimits_{z_{T}=1}^{J^{0}}\sum\nolimits_{z_{T-1}=1}^{J^{0}}\Pr
(z_{T},z_{T-1},z_{t}=j,z_{t-1}=k\left\vert x_{1:T};\tilde{\theta}%
^{(h)}\right. )$.

Repeat steps (1)-(3) for all $j$ and $k$.

\end{document}